\documentclass{mscs}
\usepackage[doublespacing]{setspace}
\usepackage{amsmath}
\usepackage{lineno,hyperref}
\usepackage{cedilleverbatim}
\usepackage{proof}
\usepackage{amssymb}
\usepackage{MnSymbol}
\usepackage{stmaryrd}
\usepackage{tikz-cd}
\usepackage{mathtools}
\modulolinenumbers[5]

\newcommand{\paragraphb}[1]{\paragraph{\textbf{#1}}}

\newtheorem{theorem}{Theorem}
\newtheorem{proposition}[theorem]{Proposition}

\newtheorem{definition}[theorem]{Definition}

\mathchardef\mhyph="2D % Define a "math hyphen"
\newcommand{\abs}[4]{{#1}\, #2\!: \! #3.\, #4}
\newcommand{\absu}[3]{{#1}\, #2.\, #3}

\newcommand{\ann}[2]{\mathit{#1}\! : \! #2}
\newcommand{\tpcheck}[0]{\overset{\leftarrow}{\in}}
\newcommand{\tpsynth}[0]{\overset{\rightarrow}{\in}}
\newcommand{\leadstoc}[0]{\ensuremath{\leadsto_{\mathbf{n}}}}
\newcommand{\leadstocs}[0]{\ensuremath{\leadsto_{\mathbf{n}}^*}}
\newcommand{\nleadstoc}[0]{\ensuremath{\nleadsto_{\mathbf{n}}}}
\newcommand{\tpsynthleads}[0]{\ensuremath{\overset{\leadstocs}{\in}}}

\usepackage{natbib}
\title[Mathematical Structures in Computer Science]
{Monotone Recursive Types and Recursive Data Representations in Cedille}
\author[C. Jenkins and A. Stump]{
  Christopher Jenkins and Aaron Stump \\
  Computer Science, 14 MacLean Hall, The University of Iowa, Iowa City, Iowa, USA \\
  email: \texttt{\string{firstname\string}-\string{lastname\string}@uiowa.edu}
}
\shorttitle{Recursive Types and Data Representations in CDLE}

\pubyear{2019}
\pagerange{\pageref{firstpage}--\pageref{lastpage}}
\volume{00}
\doi{}

\begin{document}

\maketitle

%\begin{frontmatter}

% \author{Christopher Jenkins}
% %\address{Computer Science, 14 MacLean Hall, The University of Iowa, Iowa City, Iowa, USA}
% \ead{christopher-jenkins@uiowa.edu}

% \author{Aaron Stump}
% \address{Computer Science, 14 MacLean Hall, The University of Iowa, Iowa City, Iowa, USA}
% \ead{aaron-stump@uiowa.edu}

\begin{abstract}
  Guided by Tarksi's fixpoint theorem in order theory, we show how to derive
  monotone recursive types with constant-time \textit{roll} and \textit{unroll}
  operations within Cedille, an impredicative, constructive, and logically
  consistent pure typed lambda calculus.
  This derivation takes place within the preorder on Cedille types induced by
  \emph{type inclusions}, a notion which is expressible within the theory
  itself.
  As applications, we use monotone recursive types to generically derive two
  recursive representations of data in lambda calculus, the Parigot and
  Scott encoding.
  For both encodings, we prove induction and examine the computational and
  extensional properties of their destructor, iterator, and primitive recursor
  in Cedille.
  For our Scott encoding in particular, we translate into Cedille a construction
  due to \citet{lepigre+19} that equips Scott naturals with primitive recursion,
  then extend this construction to derive a generic induction principle.
  This allows us to give efficient and provably unique (up to function
  extensionality) solutions for the iteration and primitive recursion schemes
  for Scott-encoded data.
\end{abstract}

% \begin{keyword}
% recursive types \sep type theory\sep lambda encodings 
% \end{keyword}

% \end{frontmatter}

%\linenumbers

\section{Introduction}
\label{sec:introduction}
% \subsection{Recursive Types}
In type theory and programming languages, recursive types \(\absu{\mu}{X}{T}\)
are types where the variable $X$ bound by $\mu$ in $T$ stands for the entire
type expression again.
The relationship between a recursive type and its one-step unrolling
\([\absu{\mu}{X}{T}/X]T\) (the substitution of \(\absu{\mu}{X}{T}\) for \(X\) in
\(T\)) is the basis for the important distinction of \emph{iso-} and
\emph{equi-recursive} types~\citep{Crary+99}~\citep[see also][Section~20.2]{pierce02}.
With iso-recursive types, the two types are related by constant-time functions
\(\mathit{unroll} : \absu{\mu}{X}{T} \to [\absu{\mu}{X}{T}/X]T\) and
\(\mathit{roll} : [\absu{\mu}{X}{T}/X]T \to \absu{\mu}{X}{T}\) which are mutual
inverses (composition of these two in either order produces a function that is
extensionally the identity function).
With equi-recursive types, the recursive type and its one-step unrolling are
considered definitionally equal, and \textit{unroll} and
\textit{roll} are not needed to pass between the two.

Without restrictions, adding recursive types as primitives
to an otherwise terminating theory allows the typing of diverging terms.
For example, let \(T\) be any type and let $B$ abbreviate $\absu{\mu}{X}{(X \to
  T)}$.
Then, we see that $B$ is equivalent to $B \to T$, allowing us to assign type $B
\to T$ to $\absu{\lambda}{x}{x\ x}$.
From that same type equivalence, we see that we may also
assign type $B$ to this term, allowing us to assign the type \(T\) to the
diverging term \(\Omega = (\absu{\lambda}{x}{x\ x})\ \absu{λ}{x}{x\ x}\).
This example shows that not only termination, but also soundness of the theory
when interpreted as a logic under the Curry-Howard
isomorphism~\citep{Sorensen06}, is lost when introducing unrestricted recursive
types (\(T\) here is arbitrary, so this would imply all types are inhabited).

The usual restriction on recursive types is to
require that to form (alternatively, to introduce or to eliminate) $\mu\ X.\ T$,
the variable $X$ must occur only positively in $T$, where the function-type
operator $\to$ preserves polarity in its codomain part and switches polarity in its
domain part.
For example, $X$ occurs only positively in $(X \to Y)\to Y$, while $Y$ occurs
both positively and negatively.
Since positivity is a syntactic condition, it is not compositional: if $X$
occurs positively in $T_1$ and in $T_2$, where \(T_2\) contains also the free
variable $Y$, this does not mean it will occur positively in $[T_1/Y]T_2$.
For example, take $T_1$ to be $X$ and $T_2$ to be $Y \to X$.

In search of a compositional restriction for ensuring termination in the
presence of recursive types, \citet{matthes98,matthes02} investigated monotone
iso-recursive types in a theory that requires evidence of
monotonicity equivalent to the following property of a type scheme $F$ (where
the center dot indicates application to a type):
\[
  \absu{∀}{X}{\absu{∀}{Y}{(X \to Y) \to F ·X \to F ·Y}}
\]

In Matthes's work, monotone recursive types are an addition to an underlying
type theory, and the resulting system must be analyzed anew for such properties
as subject reduction, confluence, and normalization. In the present paper, we
take a different approach by deriving monotone recursive types \emph{within an
existing type theory}, the Calculus of Dependent Lambda Eliminations
(CDLE)~\citep{stump17,stump18c}.
Given any type scheme \(F\) satisfying a form of monotonicity, we show how to
define a type \(\mathit{Rec} \cdot F\) together with constant-time rolling and unrolling
functions that witness the isomorphism between \(\mathit{Rec} \cdot F\) and \(F
\cdot (\mathit{Rec} \cdot F)\).
The definitions are carried out in Cedille, an implementation of
CDLE.\footnote{All code and proofs appearing in listings can be found in full at
  \url{https://github.com/cedille/cedille-developments/tree/master/recursive-representation-of-data}}
The main benefit to this approach is that the existing meta-theoretic results
for CDLE --- confluence, logical consistency, and normalization for a class
of types that includes ones defined here --- apply, since they hold globally and
hence perforce for the particular derivation of monotone recursive types.
% All code found within listings is available at
% \url{https://github.com/cedille/cedille-developments}.

\paragraphb{Recursive representations of data in typed lambda calculi.} One important
application of recursive types is their use in forming inductive datatypes,
especially within a pure typed lambda calculus where data must be encoded using
lambda expressions.
The most well-known method of lambda encoding is the \emph{Church
  encoding}, or \emph{iterative representation}, of data, which produces terms
typable in System F unextended by recursive types.
The main deficiency of Church-encoded data is
that data destructors, such as the predecessor function for naturals, can take
no better than linear time to compute~\citep{parigot89}.
Recursive types can be used to type lambda encodings with efficient data
destructors
\citep{SU99_Type-Fixpoints-Iteration-Recursion,SF16_Efficiency-of-Lambda-Encodings-in-Total-Type-Theory}. 

As practical applications of Cedille's derived recursive types, we generically
derive two \emph{recursive representations} of data described by 
\citet{parigot89,parigot1992}: the \emph{Scott encoding} and the \emph{Parigot
  encoding}.
While both encodings support efficient data destructors, the Parigot encoding
readily supports the primitive recursion scheme but suffers from exponential
representation size, whereas the Scott encoding has a linear representation
size but only readily supports the case-distinction scheme.
For both encodings, we derive the primitive recursion scheme and induction
principle.
That this can be done for the Scott encoding in CDLE is itself a
remarkable result that builds on the derivations by \citet{lepigre+19} and
\cite{parigot88} of a strongly normalizing recursor for Scott naturals in resp.\
a Curry style type theory with sophisticated subtyping and a logical framework.
To the best of our knowledge, the generic derivation for Scott-encoded data
is the second-ever demonstration of a lambda encoding with induction principles,
efficient data destructors, and linear-space representation (see
\cite{firsov18b} for the first).

\paragraphb{Overview of this paper.} We begin the remainder of this paper with
an introduction to CDLE (Section~\ref{sec:cdle}), before proceeding to the
derivation of monotone recursive types (Section~\ref{sec:rectypes}).
Our applications of recursive types in deriving lambda encodings with induction
follows a common structure.
This structure is outlined in Section~\ref{sec:datatypes-and-recursion-schemes},
where we elaborate on the connection between structured recursion schemes and
lambda encodings of datatypes using impredicative polymorphism, and we explain
the computational and extensional criteria we use to characterize
implementations of the iteration scheme for the Scott and Parigot encodings.
Section~\ref{sec:scott} covers the Scott encoding, introducing the
case-distinction scheme and giving first a concrete derivation of natural
numbers supporting proof by cases, then the fully generic derivation.
Section~\ref{sec:parigot} covers the Parigot encoding, introducing the primitive
recursion scheme and giving first a concrete example for Parigot-encoded
naturals with an induction principle, then the fully generic derivation.
Section~\ref{sec:lr} revisits the Scott encoding, showing concretely how to
derive primitive recursion for Scott naturals, then generalizing to the
derivation of the standard induction principle for generic Scott-encoded data.
Finally, Section~\ref{sec:related} discusses related work and
Section~\ref{sec:conclusion} concludes and discusses future work.

\section{Background: the Calculus of Dependent Lambda Eliminations}
\label{sec:cdle}

In this section, we review the Calculus of Dependent Lambda Eliminations (CDLE)
and its implementation in the Cedille programming
language~\citep{stump17,stump18c}.
CDLE is a logically consistent constructive type theory that contains as a
subsystem the impredicative and extrinsically typed Calculus of Constructions
(CC).
It is designed to serve as a tiny kernel theory for interactive theorem provers,
minimizing the trusted computing base.
CDLE can be described concisely in 20 typing rules and is implemented by
\emph{Cedille Core}~\citep{stump18b}, a type checker consisting of \(\sim\)1K
lines of Haskell.
For the sake of exposition, we present the less dense version of CDLE
implemented by Cedille, described by \citet{stump18c} (\citet{stump17} describes
an earlier version); here, we have slightly simplified the typing rule for the
\(\rho\) term construct.
 
To achieve compactness, CDLE is a pure typed lambda calculus.
In particular, it has no primitive constructs for inductive datatypes.
Instead, datatypes can be represented using lambda encodings.
\citet{geuvers01} showed that induction principles are not derivable for these in
second-order dependent type theory.
\citet{Stu18_From-Realizibility-to-Induction} showed how CDLE overcomes this
fundamental difficulty by extending CC with three new typing constructs: the
implicit products of~\citet{miquel01}; the dependent intersections
of~\citet{kopylov03}; and equality over untyped terms.
The formation rules for these new constructs are first shown in
Figure~\ref{fig:knd}, and their introduction and elimination rules are explained
in Section~\ref{sec:term-constructs}.

\subsection{Type and kind constructs}
There are two sorts of classifiers in CDLE.
Kinds \(\kappa\) classify type constructors, and types (type constructors of
kind \(\star\)) classify terms.
Figure~\ref{fig:knd} gives the inference rules for the judgment \(\Gamma \vdash
\kappa\) that kind \(\kappa\) is well-formed under context \(\Gamma\),
and for judgment \(\Gamma \vdash T \tpsynth \kappa\) that type constructor \(T\)
is well-formed and has kind \(\kappa\) under \(\Gamma\).
For brevity, we take these figures as implicitly specifying the grammar for
types and kinds.
In the inference rules, capture-avoiding substitution is written \([t/x]\) for
terms and \([T/X]\) for types, and convertibility of classifiers is notated with
\(\cong\).
The non-congruence rules for conversion for type constructors are given in
Figure~\ref{fig:conv} (the convertibility rules for kinds do not appear here as
they consist entirely of congruence rules).
In those rules, call-by-name reduction, written \(\leadstoc\) and
\(\leadstocs\) for its reflexive transitive closure, is used 
to reduce types to weak head normal form before checking convertibility
with the auxiliary relation \(\cong^{\text{t}}\), in which corresponding term
subexpressions are checked for \(\beta\eta\)-equivalence (written
\(=_{\beta\eta}\)), modulo erasure (see Figure~\ref{fig:eraser}).

\begin{figure}[!htbp]
  \[
    \begin{array}{c}
      \fbox{\(\Gamma \vdash \kappa\)}
      \\ \\
      \begin{array}{ccc}
        \infer{\Gamma \vdash \star}{\ }
        & \infer{\Gamma\vdash\abs{\Pi}{x}{T}{\kappa}}{\Gamma \vdash T \tpsynth
          \star & \Gamma,x:T\vdash\kappa}
        & \infer{\Gamma\vdash\abs{\Pi}{X}{\kappa'}{\kappa}}{\Gamma \vdash \kappa' & \Gamma,X:\kappa'\vdash\kappa}
      \end{array}
      \\ \\
      \fbox{\(\Gamma \vdash T \tpsynth \kappa\)}
      \\ \\
  \begin{array}{cc}
    \infer{\Gamma \vdash X \tpsynth \kappa}{(X : \kappa) \in \Gamma} &
    \infer{\Gamma\vdash \abs{\forall}{X}{\kappa}{T} \tpsynth \star}{\Gamma \vdash \kappa & \Gamma,X:\kappa\vdash T \tpsynth \star} \\ \\
    \infer{\Gamma\vdash\abs{\forall}{x}{T}{T'} \tpsynth \star}{\Gamma \vdash T \tpsynth \star & \Gamma,x:T\vdash T' \tpsynth \star} &
    \infer{\Gamma\vdash\abs{\Pi}{x}{T}{T'} \tpsynth \star}{\Gamma \vdash T \tpsynth \star & \Gamma,x:T\vdash T' \tpsynth \star} \\ \\
    \infer{\Gamma\vdash\abs{\lambda}{x}{T}{T'} \tpsynth \abs{\Pi}{x}{T}{\kappa}}{\Gamma \vdash T \tpsynth \star & \Gamma,x:T\vdash T'\tpsynth\kappa} &
    \infer{\Gamma\vdash\abs{\lambda}{X}{\kappa}{T} \tpsynth \abs{\Pi}{X}{\kappa}{\kappa'}}{\Gamma \vdash \kappa & \Gamma,X:\kappa\vdash T\tpsynth\kappa'} \\ \\
    \infer{\Gamma\vdash T\ t \tpsynth [t/x]\kappa}{\Gamma\vdash T \tpsynth \abs{\Pi}{x}{T'}{\kappa} & \Gamma\vdash t \tpcheck T'} &
    \infer{\Gamma\vdash T_1 \cdot T_2 \tpsynth [T_2/X]\kappa_1}
          {
           \begin{array}{c}
             \Gamma\vdash T_1 \tpsynth \abs{\Pi}{X}{\kappa_2}{\kappa_1}
             \\ \Gamma\vdash T_2 \tpsynth \kappa_2'
             \quad \kappa_2 \cong \kappa_2'
           \end{array}
          } \\ \\
    \infer{\Gamma\vdash\abs{\iota}{x}{T}{T'} \tpsynth \star}{\Gamma \vdash T \tpsynth \star & \Gamma,x:T\vdash T' \tpsynth \star} &
    \infer{\Gamma\vdash \{ t \simeq t' \} \tpsynth \star}{\textit{FV}(t\ t')\subseteq\textit{dom}(\Gamma)}
  \end{array}
    \end{array}
  \]
  \caption{Kind formation and kinding of types in CDLE}
  \label{fig:knd}
\end{figure}

\begin{figure}
  \[
    \begin{array}{c}
      \fbox{\(T_1 \cong T_2\)}
      \quad \fbox{\(T_1 \cong^{\text{t}} T_2\)}
      \\ \\
      \infer{
       T_1 \cong T_2
      }{
       T_1 \leadstocs T_1' \nleadstoc
       \quad T_2 \leadstocs T_2' \nleadstoc
       \quad T_1' \cong^{\text{t}} T_2'
      }
      \\ \\
      \begin{array}{ccc}
        \infer{
         X \cong^{\text{t}} X
        }{}
        & 
        \infer{
         T_1\ t_1 \cong^{\text{t}} T_2\ t_2
        }{
         T_1 \cong^{\text{t}} T_2 \quad |t_1| =_{\beta\eta} |t_2|
        }
        &
          \infer{
           \{ t_1 \simeq t_2 \} \cong^{\text{t}} \{ t_1'\ \simeq t_2' \}
          }{
           |t_1| =_{\beta\eta} |t_1'| \quad |t_2| =_{\beta\eta} |t_2'|
          }
      \end{array}
    \end{array}
    % \begin{array}{ll}
    %   \infer{T \cong T'}{T \leadsto^*_\beta T_1 & T' \leadsto^*_\beta T_2 & T_1\cong^t T_2}
    %   \\ \\ \infer{T\ t \cong^t T\ t'}{T \cong^t T' & |t| =_{\beta\eta} |t'|}
    %   & \infer{\{ t_1 \simeq t_2 \} \cong^t \{ t_1'\ \simeq t_2' \}}{|t_1| =_{\beta\eta} |t_1'| & |t_2| =_{\beta\eta} |t_2'|}
    % \end{array}
  \]
  \caption{Non-congruence rules for classifier convertibility}
  \label{fig:conv}
\end{figure}

The kinds of CDLE are the same as those of CC: \(\star\) classifies types,
\(\abs{\Pi}{X}{\kappa_1}{\kappa_2}\) classifies type-level functions that
abstract over type constructors, and \(\abs{\Pi}{x}{T}{\kappa}\) classifies type-level
functions that abstract over terms.
The type constructs that CDLE inherits from CC are type variables,
(impredicative) type constructor quantification \(\abs{\forall}{X}{\kappa}{T}\), dependent
function (or \emph{product}) types \(\abs{\Pi}{x}{T}{T'}\), abstractions over terms
\(\abs{\lambda}{x}{T}{T'}\) and over type constructors
\(\abs{\lambda}{X}{\kappa}{T}\), and applications of type constructors to
terms \(T\ t\) and to other type constructors \(T_1 \cdot T_2\).

The additional type constructs are: types for dependent functions with erased
arguments (or \emph{implicit product types}) \(\abs{\forall}{x}{T}{T'}\);
dependent intersection types \(\abs{\iota}{x}{T}{T'}\); and equality types \(\{t
\simeq t'\}\).
Kinding for the first two of these follows the same format as kinding of
dependent function types, e.g., \(\abs{\forall}{x}{T}{T'}\) has kind \(\star\)
if \(T\) has kind \(\star\) and if \(T'\) has kind \(\star\) under a typing
context extended by the assumption \(\ann{x}{T}\).
For equality types, the only requirement for the type \(\{t \simeq t'\}\) to be
well-formed is that the free variables of both \(t\) and \(t'\) (written
\(\textit{FV}(t\ t')\)) are declared in the typing context.
Thus, the equality is \emph{untyped}, as neither \(t\) nor \(t'\) need be
typable.

\subsection{Term constructs}
\label{sec:term-constructs}
\begin{figure}
  \[
    \begin{array}{c}
      \fbox{\(\Gamma \vdash t \tpsynth T\)} \quad \fbox{\(\Gamma \vdash t
      \tpcheck T\)}
      \\ \\
      \Gamma \vdash t \tpsynthleads T\ =\ \exists T'.\ (\Gamma
      \vdash t \tpsynth T') \land (T' \leadstocs T)
      \\ \\ \\ 
  \begin{array}{cc}
    \infer{\Gamma\vdash x\tpsynth T}{(x : T)\in\Gamma} &
    \infer{\Gamma\vdash t\tpcheck T}{\Gamma\vdash t\tpsynth T' & T' \cong T} \\ \\
    \infer{\Gamma\vdash \absu{\lambda}{x}{t} \tpcheck T}{T \leadstocs \abs{\Pi}{x}{T_1}{T_2} & \Gamma,x:T_1\vdash t\tpcheck T_2} &
    \infer{\Gamma\vdash t\ t' \tpsynth [t'/x]T}{\Gamma\vdash t \tpsynthleads \abs{\Pi}{x}{T'}{T} & \Gamma\vdash t' \tpcheck T'} \\ \\

    \infer{\Gamma\vdash \absu{\Lambda}{X}{t} \tpcheck T'}
          {T' \leadstocs \abs{\forall}{X}{\kappa}{T} & \Gamma,X:\kappa\vdash t \tpcheck T} &
    \infer{\Gamma\vdash t \cdot T' \tpsynth [T'/X]T}
          {
           \begin{array}{c}
             \Gamma\vdash t \tpsynthleads \abs{\forall}{X}{\kappa}{T}
             \\ \Gamma\vdash T' \tpsynth \kappa'
             \quad \kappa'\cong\kappa
           \end{array}
          } \\ \\
    \infer{\Gamma \vdash \chi\ T\ \mhyph\ t\ \tpsynth T}
          {\Gamma \vdash T \tpsynth \star & \Gamma \vdash t \tpcheck T}
    &
    \infer{
     \Gamma \vdash [x ◂ T_1 = t_1]\ \mhyph\ t_2 \tpcheck T_2
    }{
    \begin{array}{c}
      \Gamma \vdash T_1 \tpsynth \star
      \\ \Gamma \vdash t_1 \tpcheck T_1
      \quad \Gamma,\ann{x}{T_1} \vdash t_2 \tpcheck T_2
    \end{array}
    }
  \end{array}
    \end{array}
  \]
\caption{Standard term constructs}
\label{fig:tp}
\end{figure}

Figure~\ref{fig:tp} gives the type inference rules for the standard term
constructs of CDLE.
These type inference rules, as well as those listed in
Figures~\ref{fig:cdle-implicit-product}, \ref{fig:cdle-dependent-intersection},
and \ref{fig:cdle-equality}, are \emph{bidirectional}~\citep[c.f.][]{pierce+00}:
judgment \(\Gamma \vdash t \tpsynth T\) indicates term \(t\) \emph{synthesizes}
type \(T\) under typing context \(\Gamma\) and judgment \(\Gamma \vdash t
\tpcheck T\) indicates \(t\) can be \emph{checked} against type \(T\).
These rules are to be read bottom-up as an algorithm for type inference, with
\(\Gamma\) and \(t\) considered inputs in both judgments and the type \(T\) an
output in the synthesis judgment and input in the checking judgment.
As is common for a bidirectional system, Cedille has a mechanism allowing
the user to ascribe a type annotation to a term: \(\chi\ T\ \mhyph\ t\)
synthesizes type \(T\) if \(T\) is a well-formed type of kind \(\star\) and
\(t\) can be checked against this type.
During type inference, types may be call-by-name reduced to weak head normal
form in order to reveal type constructors.
For brevity, we use the shorthand \(\Gamma \vdash t \tpsynthleads T\) (defined
formally near the top of Figure~\ref{fig:tp}) in some premises to indicate that
\(t\) synthesizes some type \(T'\) that reduces to \(T\).

We assume the reader is familiar with the type constructs of CDLE inherited from
CC.
Abstraction over types in terms is written \(\absu{\Lambda}{X}{t}\), and
application of terms to types (polymorphic type instantiation) is written \(t
\cdot T\).
In code listings, type arguments are sometimes omitted when Cedille can infer
these from the types of term arguments.
Local term definitions are given with
\[[x ◂ T_1 = t_1]\ \mhyph\ t_2\]
to be read ``let \(x\) of type \(T_1\) be \(t_1\) in \(t_2\),'' and global
definitions are given with \(x ◂ T = t .\) (ended with a period), where
\(t\) is checked against type \(T\).

\begin{figure}
  \[
    \begin{array}{lllllll}
      |x| & = & x &\ &
                       |\absu{\lambda}{x}{t}| & = & \absu{\lambda}{x}{|t|} \\
      |t\ t'| & = & |t|\ |t'| &\ &
                                   |t\cdot T| & = & |t| \\
      |\absu{\Lambda}{x}{t}| & = & |t| &\ &
                                            |t\ \mhyph t'| & = & |t| \\
      |[t , t']| & = & |t| &\ &
                                |t.1| & = & |t| \\
      |t.2| & = & |t| &\ &
                           |\beta\{t\}| & = & |t| \\
      % |\delta\ \mhyph\ t| & = & \absu{\lambda}{x}{x}&\ &
      |\rho\ t\ @x.T'\ \mhyph\ t'| & = & |t'| &\ &
      |\varphi\ t\ \mhyph\ t'\ \{t''\}| & = & |t''|
      \\ |\chi\ T\ \mhyph\ t| & = & |t| &\ &
      |\varsigma\ t| & = & |t|
      \\ |[x ◂ T_1 = t_1]\ \mhyph\ t_2| & = & (\absu{\lambda}{x}{|t_2|})\ |t_1|
      &\ & |\delta\ \mhyph\ t| & = & \absu{\lambda}{x}{x}
    \end{array}
  \]
  \caption{Erasure for annotated terms}
  \label{fig:eraser}
\end{figure}

In describing the new type constructs of CDLE, we make reference to the erasures of the
corresponding annotations for terms.
The full definition of the erasure function \(|-|\), which extracts an untyped
lambda calculus term from a term with type annotations, is given in
Figure~\ref{fig:eraser}.
For the term constructs of CC, type abstractions \(\absu{\Lambda}{X}{t}\) erase
to \(|t|\) and type applications \(t \cdot T\) erase to \(|t|\).
As a Curry-style theory, the convertibility relation of Cedille is
\(\beta\eta\)-conversion for untyped lambda calculus terms --- there is no
notion of reduction or conversion for the type-annotated language of terms.

\begin{figure}[h]
  \centering
  \[
    \begin{array}{cc}
      \infer{\Gamma\vdash \absu{\Lambda}{x}{t} \tpcheck T}
      {
        T \leadstocs \abs{\forall}{x}{T_1}{T_2}
        \quad \Gamma,x:T_1\vdash t \tpcheck T_2
        \quad x\not\in\textit{FV}(|t|)
      }
      &
        \infer{
         \Gamma\vdash t\ \mhyph t' \tpsynth [t'/x]T_2
        }{
        \Gamma\vdash t \tpsynthleads \abs{\forall}{x}{T_1}{T_2}
        \quad \Gamma\vdash t' \tpcheck T_1
        }
    \end{array}
  \]
  \caption{Implicit products}
  \label{fig:cdle-implicit-product}
\end{figure}

\paragraph{\textbf{The implicit product type} \(\abs{\forall}{x}{T_1}{T_2}\)} of
\citet{miquel01} (Figure~\ref{fig:cdle-implicit-product}) is the type for
functions which accept an erased (computationally irrelevant) input of type
\(T_1\) and produce a result of type \(T_2\).
Implicit products are introduced with \(\absu{\Lambda}{x}{t}\), and the type
inference rule is the same as for ordinary function abstractions except for the side
condition that \(x\) does not occur free in the erasure of the body \(t\).
Thus, the argument can play no computational role in the function and exists
solely for the purposes of typing.
The erasure of the introduction form is \(|t|\).
For application, if \(t\) has type \(\abs{\forall}{x}{T_1}{T_2}\) and \(t'\) has
type \(T_1\), then \(t\ \mhyph t'\) has type \([t'/x]T_2\) and erases to
\(|t|\).
When \(x\) is not free in \(\abs{\forall}{x}{T_1}{T_2}\), we may write \(T_1
\Rightarrow T_2\), similarly to writing \(T_1 \to T_2\) for \(\abs{\Pi}{x}{T_1}{T_2}\).

Note that the notion of computational irrelevance here is not that of a
different sort of classifier for types (e.g. \(\mathit{Prop}\) in Coq,
c.f.~\citealp{coq}) that separates terms in the language into those which can be
used for computation and those which cannot.
Instead, it is similar to \emph{quantitative type
  theory}~\citep{Atk18_Quantitative-Type-Theory}: relevance and irrelevance are
properties of binders, indicating how a function may \emph{use} an argument.

\paragraph{\textbf{The dependent intersection type}
  \(\abs{\iota}{x}{T_1}{T_2}\)}
\begin{figure}[h]
  \centering
  \[
    \begin{array}{c}
      \infer{\Gamma\vdash [ t_1 , t_2 ] \tpcheck T}
      {
        T \leadstocs \abs{\iota}{x}{T_1}{T_2}
        \quad \Gamma \vdash t_1 \tpcheck T_1
        \quad \Gamma \vdash t_2 \tpcheck [t_1/x]T_2
        \quad |t_1| =_{\beta\eta} |t_2|
      }
      \\ \\
      \begin{array}{cc}
        \infer{
         \Gamma\vdash t.1 \tpsynth T_1
        }{
         \Gamma\vdash t \tpsynthleads \abs{\iota}{x}{T_1}{T_2}
        }
        &
          \infer{
           \Gamma\vdash t.2 \tpsynth [t.1/x] T_2
          }{
           \Gamma\vdash t \tpsynthleads \abs{\iota}{x}{T_1}{T_2}
          }
      \end{array}
    \end{array}
  \]
  \caption{Dependent intersections}
  \label{fig:cdle-dependent-intersection}
\end{figure}
of \citet{kopylov03}
(Figure~\ref{fig:cdle-dependent-intersection}) is the type for terms \(t\) which
can be assigned both type \(T_1\) and type \([t/x]T_2\).
It is a dependent generalization of intersection types \citep[c.f.][Part
3]{BDS13_Lambda-Calculus-with-Types} in Curry-style theories, which allow one to
express the fact that an untyped lambda calculus term can be assigned two
different types.
In Cedille's annotated language, the introduction form for dependent
intersections is written \([t_1,t_2]\), and can be checked
against type \(\abs{\iota}{x}{T_1}{T_2}\) if \(t_1\) can be checked against type
\(T_1\), \(t_2\) can be checked against \([t_1/x]T_2\), and \(t_1\) and \(t_2\)
are \(\beta\eta\)-equivalent modulo erasure.
For the elimination forms, if \(t\) synthesizes type \(\abs{\iota}{x}{T_1}{T_2}\)
then \(t.1\) (which erases to \(|t|\)) synthesizes type \(T_1\), and \(t.2\)
(erasing to the same) synthesizes type \([t.1/x]T_2\).

Dependent intersections can be thought of as a dependent pair type
where the two components are equal.
Thus, we may ``forget'' the second component: \([t_1,t_2]\) erases to \(|t_1|\).
Put another way, dependent intersections are a restricted form of
computationally transparent subset types where the proof 
that some term $t$ inhabits the subset must be definitionally equal to $t$.
A consequence of this restriction is that the proof may be recovered, in the
form of \(t\) itself, for use in computation.

\begin{figure}[h]
  \centering
  \[
    \begin{array}{c}
        \infer{\Gamma\vdash \beta\{t'\} \tpcheck T}
        {
        T \leadstocs \{t_1 \simeq t_2\}
        \quad \textit{FV}(t')\subseteq \textit{dom}(\Gamma)
        \quad |t_1| =_{\beta\eta} |t_2|
        }
      \\ \\
      \infer{
       \Gamma \vdash \rho\ t\ @x.T'\ \mhyph\ t' \tpcheck T
      }{
        \Gamma \vdash t \tpsynthleads \{t_1 \simeq t_2\}
        \quad \Gamma \vdash [t_2/x] T' \tpsynth \star
        \quad \Gamma \vdash t' \tpcheck [t_2/x]T'
        \quad [t_1/x]T' \cong T
      }
      \\ \\
      \infer{
      \Gamma\vdash \varphi\ t\ \mhyph\ t'\ \{t''\} \tpcheck T
      }{
      \Gamma\vdash t\tpcheck \{t' \simeq t''\}
      \quad \Gamma\vdash t' \tpcheck T
      \quad \textit{FV}(t'') \subseteq \textit{dom}(\Gamma)
      }
      \\ \\
      \begin{array}{cc}
        \infer{
        \Gamma \vdash \delta\ \mhyph\ t \tpcheck T
        }{
        \Gamma \vdash t \tpsynth T'
        \quad T' \cong \{\absu{\lambda}{x}{\absu{\lambda}{y}{x}} \simeq
        \absu{\lambda}{x}{\absu{\lambda}{y}{y}}\} 
        }
        &
          \infer{
          \Gamma \vdash \varsigma\ t \tpsynth \{t_2 \simeq t_1\}
          }{
          \Gamma \vdash t \tpsynthleads \{t_1 \simeq t_2\}
          }
      \end{array}
    \end{array}
  \]
  \caption{Equality type}
  \label{fig:cdle-equality}
\end{figure}

\paragraph{\textbf{The equality type} \(\{t_1 \simeq t_2\}\)} is the type of
proofs that \(t_1\) is propositionally equal to \(t_2\).
The introduction form \(\beta\{t'\}\) proves reflexive equations between
\(\beta\eta\)-equivalence classes of terms: it can be checked against the type
\(\{t_1 \simeq t_2\}\) if \(|t_1| =_{\beta\eta} |t_2|\) and if the
subexpression \(t'\) has no undeclared free variables.
We discuss the significance of the fact that \(t'\) is unrelated to the terms
being equated, dubbed the \emph{Kleene trick}, below.
In code listings, if \(t'\) is omitted from the introduction form, it defaults
to \(\absu{\lambda}{x}{x}\).

The elimination form \(\rho\ t\ @x.T'\ \mhyph\ t'\) for the equality type
\(\{t_1 \simeq t_2\}\) replaces occurrences of \(t_1\) in the checked type with
\(t_2\) before checking \(t'\).
The user indicates the occurrences of \(t_1\) to replace with \(x\) in the
annotation \(@x.T'\), which binds \(x\) in \(T'\).
The rule requires that \([t_2/x]T'\) has kind \(\star\), then checks \(t'\)
against this type, and finally confirms that \([t_1/x]T'\) is convertible with
the expected type \(T\).
The entire expression erases to \(|t'|\).

\begin{example*}
Assume \(m\) and \(n\) have type \(\mathit{Nat}\), \(\mathit{suc}\) and
\(\mathit{pred}\) have type \(\mathit{Nat} \to \mathit{Nat}\), and furthermore
that \(|\mathit{pred}\ (\mathit{suc}\ t)| =_{\beta\eta} |t|\) for all \(t\).
If \(e\) has type \(\{\mathit{suc}\ m \simeq \mathit{suc}\ n\}\), then \(\rho\
e\ @x.\{\mathit{pred}\ x \simeq \mathit{pred}\ (\mathit{suc}\ n)\}\ \mhyph\ \beta\)
can be checked with type \(\{m \simeq n\}\) as follows: we check that
\(\{\mathit{pred}\ (\mathit{suc}\ n)
\simeq \mathit{pred}\ (\mathit{suc}\ n)\}\), obtained from substituting \(x\)
in the annotation with the right-hand side of the equation of the type
of \(e\), has 
kind \(\star\); we check \(\beta\) against this type; and we check that
substituting \(x\) with \(\mathit{suc}\ m\) is convertible with the expected type \(\{m \simeq
n\}\).
By assumption \(|\mathit{pred}\ (\mathit{suc}\ m)| =_{\beta\eta} |m|\)
and \(|\mathit{pred}\ (\mathit{suc}\ n)| =_{\beta\eta} |n|\), so \(\{m \simeq
n\} \cong \{\mathit{pred}\ (\mathit{suc}\ m) \simeq \mathit{pred}\
(\mathit{suc}\ n)\}\).  
\end{example*}

Equality types in CDLE come with two additional axioms: a strong form of the
direct computation rule of NuPRL~\citep[see][Section~2.2]{allen+06} given by
\(\varphi\), and proof by contradiction given by \(\delta\).
The inference rule for an expression of the form \(\varphi\ t\ \mhyph\ t'\
\{t''\}\) says that the entire expression can be checked against type \(T\) if
\(t'\) can be, if there are no undeclared free variables in \(t''\) (so, \(t''\)
is a well-scoped but otherwise untyped term), and if \(t\) proves that \(t'\)
and \(t''\) are equal.
The crucial feature of \(\varphi\) is its erasure: the expression erases
to \(|t''|\), effectively enabling us to cast \(t''\) to the type of \(t'\).

An expression of the form \(\delta\ \mhyph\ t\) may be checked against any type
if \(t\) synthesizes a type convertible with a particular false equation,
\(\{\absu{\lambda}{x}{\absu{\lambda}{y}{x}} \simeq
\absu{\lambda}{x}{\absu{\lambda}{y}{y}}\}\).
To broaden the class of false equations to which one may apply \(\delta\), the
Cedille tool implements the \emph{B\"ohm-out} semi-decision procedure
\citep{BDPR79_Bohm-Algorithm} for discriminating between
\(\beta\eta\)-inequivalent terms.
We use \(\delta\) only once in this paper as part of a final comparison between
the Scott and Parigot encoding (see Section~\ref{sec:lr-gen}).

Finally, Cedille provides a symmetry axiom \(\varsigma\) for
equality types, with \(|t|\) the erasure of \(\varsigma\ t\).
This axiom is purely a convenience; without \(\varsigma\), symmetry for
equality types can be proven with \(\rho\).

\subsection{The Kleene trick}
\label{sec:kleene-trick}
As mentioned earlier, the introduction form
\(\beta\{t'\}\) for the equality type contains a subexpression \(t'\) that is
unrelated to the equated terms.
By allowing \(t'\) to be any closed (in context) term, we are able to define a
type of all untyped lambda calculus terms.

\begin{definition}[Top]
  \label{def:top}
  Let \(Top\) be the type \(\{\absu{\lambda}{x}{x} \simeq \absu{\lambda}{x}{x}\}\).
\end{definition}

\noindent We dub this the \emph{Kleene trick}, as one may find the idea in Kleene's later
definitions of numeric realizability in which any number is allowed as a
realizer for a true atomic formula~\citep{kleene65}.

Combined with dependent intersections, the Kleene trick also allows us to derive
computationally transparent equational subset types.
For example, let \(\mathit{Nat}\) again be the type of naturals with
\(\mathit{zero}\) the zero value,
\(\mathit{Bool}\) the type of Booleans with \(\mathit{tt}\) the truth
value, and \(\mathit{isEven} : \mathit{Nat} \to \mathit{Bool}\) a function
returning \(\mathit{tt}\) if and only if its argument is even.
Then, the type \(\mathit{Even}\) of even naturals can be defined as
\(\abs{\iota}{x}{\mathit{Nat}}{\{\mathit{isEven}\ x \simeq \mathit{tt}\}}\).
Since \(|\mathit{isEven}\ zero| =_{\beta\eta} |\mathit{tt}|\), we can
check \([\mathit{zero}, \beta\{\mathit{zero}\}]\) against type
\(\mathit{Even}\), and the expression erases to \(|\mathit{zero}|\).
More generally, if \(n\) is a \(\mathit{Nat}\) and \(t\) is a proof that
\(\{\mathit{isEven}\ n \simeq tt\}\), then \([ n, \rho\ t\ @x.\{x \simeq
\mathit{tt}\}\ -\ \beta\{n\}]\) can be checked against type \(\mathit{Even}\): 
the erasure of the first and second components are equal, and within the second
component \(\rho\) rewrites the expected type \(\{\mathit{isEven}\ n \simeq
\mathit{tt}\}\) to \(\{\mathit{tt} \simeq \mathit{tt}\}\) then checks
\(\beta\{n\}\) against this.

\subsection{Meta-theory}
\label{sec:meta}
It may concern the reader that, with the Kleene trick, it is
possible to type non-terminating terms, leading to a failure of normalization in
general in CDLE.
For example, the looping term \(\Omega\),
\(\beta\{(\absu{\lambda}{x}{x\ x})\ \absu{\lambda}{x}{x\ x}\}\), can be checked against
type \(\mathit{Top}\).
More subtly, the \(\varphi\) axiom allows non-termination in
inconsistent contexts.
Assume there is a typing context \(\Gamma\) and term \(t\) such that \(\Gamma
\vdash t \tpsynth \abs{\forall}{X}{\star}{X}\), and let \(\omega\) be the term
\[ \varphi\ (t
  \cdot \{\absu{\lambda}{x}{x} \simeq \absu{\lambda}{x}{x\ x}\})\ \mhyph\
  (\absu{\Lambda}{X}{\absu{\lambda}{x}{x}})\ \{\absu{\lambda}{x}{x\ x}\}\]
Under \(\Gamma\), \(\omega\) can be checked against the type
\(\abs{\forall}{X}{\star}{X \to X}\), and by the erasure rules \(\omega\)
erases to \(\absu{\lambda}{x}{x\ x}\).
We can then type the looping term \(\Omega\):
\[\Gamma \vdash \omega \cdot (\abs{\forall}{X}{\star}{X \to X})\ \omega
  \tpsynth \abs{\forall}{X}{\star}{X \to X}\]

Unlike the situation for unrestricted recursive types discussed in
Section~\ref{sec:introduction}, the existence of non-normalizing terms does not
threaten the logical consistency of CDLE.
For example, extensional Martin-L\"of
type theory is consistent but, due to a similar difficulty with inconsistent
contexts, is non-normalizing~\citep{dybjer16}.

\begin{proposition}[\citealp{stump18c}]
  \label{thm:consis}
  \ \\
  There is no term $t$ such that $\vdash t \tpsynth \abs{\forall}{X}{\star}{X}$.
\end{proposition}

Neither does non-termination from the Kleene trick or \(\varphi\) with
inconsistent contexts preclude the possibility of a qualified termination
guarantee.
In Cedille, closed terms of a function type are call-by-name normalizing.

\begin{proposition}[\citealp{stump18c}]
  \label{thm:cedille-termination}
  Suppose that \(\vdash t \tpsynth T\), and that there exists \(t'\)
  such that \(\vdash t' \tpsynth T \to \abs{\Pi}{x}{T_1}{T_2}\) and \(|t'| = \absu{\lambda}{x}{x}\).
  Then \(|t|\) is call-by-name normalizing.
\end{proposition}

Lack of normalization in general does, however, mean that type inference in
Cedille is formally undecidable, as there are several inference rules in which
full \(\beta\eta\)-equivalence of terms is checked.
In practice, this is not a significant impediment: even in implementations of
strongly normalizing dependent type theories, it is possible for type inference
to trigger conversion checking between terms involving astronomically slow functions,
effectively causing the implementation to hang.
For the recursive representations of inductive types we derive in this paper, we
show that closed lambda encodings do indeed satisfy the criterion required to
guarantee call-by-name normalization.

\section{Deriving Recursive Types in Cedille}
\label{sec:rectypes}

Having surveyed Cedille, we now turn to the derivation of recursive types within
it.
This is accomplished by implementing a version of Tarski's fixpoint theorem for
monotone functions over a complete lattice.
We first recall the simple corollary of Tarski's more general result \citep[c.f.][]{lassez82}. 

\begin{definition}[\(f\)-closed]
  Let \(f\) be a monotone function on a preorder \((S,\sqsubseteq)\).
  An element \(x \in S\) is said to be \emph{\(f\)-closed} when \(f(x)
  \sqsubseteq x\).
\end{definition}

\begin{theorem}[\citealp{tarski55}]
\label{thm:tarski}
  Suppose $f$ is a monotone function on complete lattice $(S,\mathbin{\sqsubseteq},\sqcap)$.
  Let $R$ be the set of \(f\)-closed elements of \(S\) and $r = \sqcap R$.  Then $f(r) = r$.
\end{theorem}

The version we implement is a strengthening of this corollary, in the sense that
it has weaker assumptions than Theorem~\ref{thm:tarski}: rather than require \(S\) be a complete
lattice, we only need that \(S\) is a preorder and \(R\) has a greatest lower bound.

\subsection{Tarski's Theorem}
\label{sec:tarski}
\begin{theorem}
  \label{thm:tarskip}
  Suppose $f$ is a monotone function on a preorder
  $(S,\mathbin{\sqsubseteq})$, and that the set $R$ of all \(f\)-closed elements
  has a greatest lower bound $r$.
  Then $f(r) \sqsubseteq r$ and $r \sqsubseteq f(r)$.
\end{theorem}
\begin{proof}
  \begin{enumerate}
    
  \item First prove \(f(r) \sqsubseteq r\).
    Since \(r\) is the greatest lower bound of \(R\), it suffices to prove
    \(f(r) \sqsubseteq x\) for every \(x \in R\).
    So, let \(x\) be an arbitrary element of \(R\).
    Since \(r\) is a lower bound of \(R\), \(r \sqsubseteq x\).
    By monotonicity, we therefore obtain \(f(r) \sqsubseteq f(x)\), and since \(x
    \in R\) we have that \(f(x) \sqsubseteq x\)
    By transitivity, we conclude that \(f(r) \sqsubseteq x\).
    
  \item Now prove \(r \sqsubseteq f(r)\).
    Using 1 above and monotonicity of \(f\), we have that \(f(f(r)) \sqsubseteq
    f(r)\).
    This means that \(f(r) \in R\), and since \(r\) is a lower bound of \(R\),
    we have \(r \sqsubseteq f(r)\).
  \end{enumerate}
\end{proof}

Notice in this proof \emph{prima facie} impredicativity: we pick a fixpoint
$r$ of $f$ by reference to a collection $R$ which (by 1) contains $r$.
We will see that this impredicativity carries over to Cedille.
We will instantiate the underlying set \(S\) of the preorder in
Theorem~\ref{thm:tarskip} to the set of Cedille types --- this is why we need to
relax the assumption of Theorem~\ref{thm:tarski} that \(S\) is a complete
lattice.
However, we must still answer several questions:
\begin{itemize}
\item how should the ordering $\sqsubseteq$ be implemented;
\item how do we express the idea of a monotone function; and
\item how do we obtain the greatest lower bound of \(R\)?
\end{itemize}

One possibility that is available in System F is to choose
functions \(A \to B\) as the ordering \(A \sqsubseteq B\), positive type
schemes \(T\) (having a free variable \(X\), and such that \(A \to B\) implies
\([A/X]T \to [B/X]T\)) as monotonic functions, and use universal quantification
to define the greatest lower bound as \(\absu{\forall}{X}{(T \to X) \to
  X}\).
This approach, described by \citet{Wad90_Recursive-Types-for-Free}, is
essentially a generalization of order theory to category theory, and the terms
inhabiting recursive types so derived are Church encodings.
However, the resulting recursive types lack the property that
\textit{roll} and \textit{unroll} are constant-time operations.

In Cedille, another possibility is available: we can interpret the ordering relation
as \emph{type inclusions}, in the sense that \(T_1\) is included into \(T_2\) if
and only if every term \(t\) of type \(T_1\) is definitionally equal to some
term of type \(T_2\).
To show how type inclusions can be expressed as a type within Cedille
(\(\mathit{Cast}\), Section~\ref{sec:casts}), we first demonstrate how to
internalize the property that some untyped term \(t\) can be viewed has having
type \(T\) (\(\mathit{View}\), Section~\ref{sec:views}): type inclusions are
thus a special case of internalized typing where we view
\(\absu{\lambda}{x}{x}\) has having type \(T_1 \to T_2\).

\subsection{Views}
\label{sec:views}

\begin{figure}[h]
  \centering
  \[
    \begin{array}{c}
      \begin{array}{cc}
        \infer{
        \Gamma \vdash \mathit{View} \cdot T\ t \tpsynth \star
        }{
        \Gamma \vdash T \tpsynth \star
        \quad \Gamma \vdash t \tpcheck \mathit{Top}
        }
        &
          \infer{
          \Gamma \vdash \mathit{elimView} \cdot T\ t\ \mhyph v \tpsynth T
          }{
          \Gamma \vdash T \tpsynth \star
          \quad \Gamma \vdash t \tpcheck \mathit{Top}
          \quad \Gamma \vdash v \tpcheck \mathit{View} \cdot T\ t
          }
      \end{array}
      \\ \\
      \infer{
        \Gamma \vdash \mathit{intrView} \cdot T\ t_1\ \mhyph t_2\ \mhyph t \tpsynth
        \mathit{View} \cdot T\ t_1
      }{
        \Gamma \vdash T \tpsynth \star
        \quad \Gamma \vdash t_1 \tpcheck \mathit{Top}
        \quad \Gamma \vdash t_2 \tpcheck T
        \quad \Gamma \vdash t \tpcheck \{t_2 \simeq t_1\}
      }
      \\ \\
      \infer{
       \Gamma \vdash \mathit{eqView} \cdot T\ \mhyph t\ \mhyph v \tpsynth
      \{t \simeq v\}
      }{
      \Gamma \vdash T \tpsynth \star
      \quad \Gamma \vdash t \tpcheck \mathit{Top}
      \quad \Gamma \vdash v \tpcheck \mathit{View} \cdot T\ t
      }
      \\ \\
      \begin{array}{lll}
        |\mathit{elimView} \cdot T\ t\ \mhyph v|
        & =
        & (\absu{\lambda}{x}{x})\ |t|
        \\ |\mathit{intrView} \cdot T\ t_1\ \mhyph t_2\ \mhyph t|
        & =
        & (\absu{\lambda}{x}{x})\ |t_1|
        \\ |\mathit{eqView} \cdot T\ \mhyph t\ \mhyph v|
        & =
        & \absu{\lambda}{x}{x}
      \end{array}
    \end{array}
  \]
  \caption{Internalized typing, axiomatically}
  \label{fig:view-ax}
\end{figure}

Figure~\ref{fig:view-ax} summarizes the derivation of the \(\mathit{View}\) type
family in Cedille, and Figure~\ref{fig:view} gives its implementation.
Type \(\mathit{View} \cdot T\ t\) is the subset of type \(T\) consisting of terms
provably equal to the untyped (more precisely
\(\mathit{Top}\)-typed, see Definition~\ref{def:top}) term \(t\).
It is defined using dependent intersection: \(\mathit{View} \cdot T\ t =
\abs{\iota}{x}{T}{\{x \simeq t\}}\).

\paragraphb{Axiomatic summary.}
The introduction form \(\mathit{intrView}\) takes an untyped term \(t_1\) and
two computationally irrelevant arguments: a term \(t_2\) of type \(T\) and a
proof \(t\) that \(t_2\) is equal to \(t_1\).
The expression \(\mathit{intrView}\ t_1\ \mhyph t_2\ \mhyph t\) erases to
\((\absu{\lambda}{x}{x})\ |t_1|\).
The elimination form \(\mathit{elimView}\) takes an untyped term \(t\) and an
erased argument \(v\) proving that \(t\) may be viewed as having type \(T\),
and produces a term of type \(T\).
The crucial property of \(\mathit{elimView}\) is its erasure:
\(\mathit{elimView}\ t\ \mhyph v\) also erases to \((\absu{\lambda}{x}{x})\
|t|\), and so the expression is definitionally equal to \(t\) itself.
Finally, \(\mathit{eqView}\) provides a reasoning principle for views.
It states that every proof \(v\) that \(t\) may be viewed as having type
\(T\) is equal to \(t\).

\begin{figure}[h]
  \centering
  \small
\begin{verbatim}
module view .

import utils/top .

View ◂ Π T: ★. Top ➔ ★ = λ T: ★. λ t: Top. ι x: T. { x ≃ t } .

intrView ◂ ∀ T: ★. Π t1: Top. ∀ t2: T. { t2 ≃ t1 } ➾ View ·T t1
= Λ T. λ t1. Λ t2. Λ t. [ φ t - t2 { t1 } , β{ t1 } ] .

elimView ◂ ∀ T: ★. Π t: Top. View ·T t ➾ T
= Λ T. λ t. Λ v. φ v.2 - v.1 { t } .

eqView ◂ ∀ T: ★. ∀ t: Top. ∀ v: View ·T t. { t ≃ v }
= Λ T. Λ t. Λ v. ρ v.2 @x.{ t ≃ x } - β .

selfView ◂ ∀ T: ★. Π t: T. View ·T β{ t }
= Λ T. λ t. intrView β{ t } -t -β .

extView
◂ ∀ S: ★. ∀ T: ★. Π t: Top. (Π x: S. View ·T β{ t x }) ➾ View ·(S ➔ T) t
= Λ S. Λ T. λ t. Λ v.
  intrView ·(S ➔ T) t -(λ x. elimView β{ t x } -(v x)) -β .
\end{verbatim}
  \caption{Internalized typing (\texttt{view.ced})}
  \label{fig:view}
\end{figure}

\paragraphb{Implementation.}
We now turn to the Cedille implementation of \(\mathit{View}\) and its
operations in Figure~\ref{fig:view}.
The definition of \(\mathit{intrView}\) uses the \(\varphi\) axiom
(Figure~\ref{fig:cdle-equality}) and the Kleene trick
(Section~\ref{sec:kleene-trick}) so that the resulting \(\mathit{View} \cdot T\
t_1\) erases to \(|t_1|\) (see Figure~\ref{fig:eraser} for erasure rules).
Because of the Kleene trick, the requirement that a term both has type \(T\) and
also proves itself equal to \(t_1\) does not restrict the terms and types over
which we may form a \(\mathit{View}\).
The elimination form uses \(\varphi\) to cast \(t\) to the type \(T\) of \(v.1\)
using the equality \(\{v.1 \simeq t\}\) given by \(v.2\).
The reasoning principle \(\mathit{eqView}\) uses \(\rho\) to rewrite the
expected type \(\{t \simeq v \}\) with the equation \(\{v.1 \simeq t\}\) (\(v\)
and \(v.1\) are convertible terms).

The last two definitions of Figure~\ref{fig:view}, \(\mathit{selfView}\) and
\(\mathit{extView}\), are auxiliary.
Since they can be derived solely from the introduction and elimination forms,
they are not included in the axiomatization given in Figure~\ref{fig:view-ax}.
Definition \(\mathit{selfView}\) reflects the typing judgment that
\(t\) has type \(T\) into the proposition that \(\beta\{t\}\) can be viewed as having
type \(T\).
Definition \(\mathit{extView}\) provides an extrinsic typing principle for
functions: if we can show of an untyped term \(t\) that for all inputs \(x\) of
type \(S\) we have that \(t\ x\) can be viewed as having type \(T\), then in fact \(t\) can
be viewed as having type \(S \to T\).

\subsection{Casts}
\label{sec:casts}

Type inclusions, or \emph{casts}, are represented by functions from \(S\) to
\(T\) that are provably equal to \(\absu{\lambda}{x}{x}\) \citetext{see
  \citealp{breitner+16}, and also \citealp{firsov18b} for the related notion of
  Curry-style ``identity functions''}.
With types playing the role of elements of
the preorder, existence of a cast from type $S$ to type $T$ will play the role of
the ordering $S \sqsubseteq T$ in the proof of Theorem~\ref{thm:tarskip}.
We summarize the derivation of casts in Cedille axiomatically in
Figure~\ref{fig:casts-ax}, and show their implementation in
Figure~\ref{fig:casts}.

\begin{figure}
  \centering
  \[
    \begin{array}{c}
      \infer{
      \Gamma \vdash \mathit{Cast} \cdot S \cdot T \tpsynth \star
      }{
      \Gamma \vdash S \tpsynth \star
      \quad \Gamma \vdash T \tpsynth \star
      }
      \\ \\
      \infer{
      \Gamma \vdash \mathit{intrCast} \cdot S \cdot T\ \mhyph t\ \mhyph t'
      \tpsynth \mathit{Cast} \cdot S \cdot T
      }{
      \Gamma \vdash S \tpsynth \star
      \quad \Gamma \vdash T \tpsynth \star
      \quad \Gamma \vdash t \tpcheck S \to T
      \quad \Gamma \vdash t' \tpcheck \abs{\Pi}{x}{S}{\{t\ x \simeq x\}}
      }
      \\ \\
      \infer{
       \Gamma \vdash \mathit{elimCast} \cdot S \cdot T\ \mhyph c\ \tpsynth S \to T
      }{
      \Gamma \vdash S \tpsynth \star
      \quad \Gamma \vdash T \tpsynth \star
      \quad \Gamma \vdash c \tpcheck \mathit{Cast} \cdot S \cdot T
      }
      \\ \\
      \infer{
       \Gamma \vdash \mathit{eqCast} \cdot S \cdot T\ \mhyph c \tpsynth
      \{\absu{\lambda}{x}{x} \simeq c\}
      }{
      \Gamma \vdash S \tpsynth \star
      \quad \Gamma \vdash T \tpsynth \star
      \quad \Gamma \vdash c \tpcheck \mathit{Cast} \cdot S \cdot T
      }
      \\ \\
      \begin{array}{lll}
        |\mathit{intrCast} \cdot S \cdot T\ \mhyph t\ \mhyph t'|
        & =
        & (\absu{\lambda}{x}{(\absu{\lambda}{x}{x})\ x})\ \absu{\lambda}{x}{x}
        \\ |\mathit{elimCast} \cdot S \cdot T\ \mhyph c|
        & =
        & (\absu{\lambda}{x}{x})\ \absu{\lambda}{x}{x}
        \\ |\mathit{eqCast}|
        & =
        & \absu{\lambda}{x}{x}
      \end{array}
    \end{array}
  \]
  \caption{Casts, axiomatically}
  \label{fig:casts-ax}
\end{figure}

\paragraphb{Axiomatic summary.}
The introduction form, \(\mathit{intrCast}\), takes two erased arguments: a
function \(t : S \to T\) and a proof that \(t\) is \emph{extensionally} the
identity function for all terms of type \(S\).
In intrinsic type theories, there would not be much more to say: identity
functions cannot map from $S$ to $T$ unless $S$ and $T$ are convertible 
types.
But in an extrinsic type theory like CDLE, there are many nontrivial
casts, especially in the presence of the \(\varphi\) axiom.
Indeed, by enabling the definition of \(\mathit{extView}\), we will see that
\(\varphi\) plays a crucial role in the implementation of \(\mathit{intrCast}\).

The elimination form, \(\mathit{elimCast}\), takes as an erased argument a proof
\(c\) of the inclusion of type \(S\) into type \(T\).
Its crucial property is its erasure: in the Cedille implementation,
\(\mathit{elimCast}\ \mhyph c\) erases to a term convertible with
\(\absu{\lambda}{x}{x}\).
The last construct listed in Figure~\ref{fig:casts-ax} is \(\mathit{eqCast}\),
the reasoning principle.
It states that every witness of a type inclusion is provably equal to
\(\absu{\lambda}{x}{x}\).

\begin{figure}
\small
\begin{verbatim}
module cast.

import view .

Cast ◂ ★ ➔ ★ ➔ ★ = λ S: ★. λ T: ★. View ·(S ➔ T) β{ λ x. x } .

intrCast ◂ ∀ S: ★. ∀ T: ★. ∀ t: S ➔ T. (Π x: S. { t x ≃ x }) ➾ Cast ·S ·T
= Λ S. Λ T. Λ t. Λ t'.
  extView ·S ·T β{ λ x. x } -(λ x. intrView β{ x } -(t x) -(t' x)) .

elimCast ◂ ∀ S: ★. ∀ T: ★. Cast ·S ·T ➾ S ➔ T
= Λ S. Λ T. Λ c. elimView β{ λ x. x } -c .

eqCast ◂ ∀ S: ★. ∀ T: ★. ∀ c: Cast ·S ·T. { λ x. x ≃ c }
= Λ S. Λ T. Λ c. eqView -β{ λ x. x } -c .
\end{verbatim}
\caption{Casts (\texttt{cast.ced})}
\label{fig:casts}
\end{figure}

\paragraphb{Implementation.}
We now turn to the implementation of \(\mathit{Cast}\) in
Figure~\ref{fig:casts}.
The type \(\mathit{Cast} \cdot S \cdot T\) itself is defined to be the type of
proofs that \(\absu{\lambda}{x}{x}\) may be viewed as having type \(S \to T\).
The elimination form \(\mathit{elimCast}\) and reasoning principle
\(\mathit{eqCast}\) are direct consequences of \(\mathit{elimView}\) and
\(\mathit{eqView}\), so we focus on the introduction form \(\mathit{intrCast}\).
Without the use of \(\mathit{extView}\), we might expect the definition of
\(\mathit{intrCast}\) to be of the form
\[\mathit{intrView} \cdot (S \to T)\ \beta\{\absu{\lambda}{x}{x}\}\ \mhyph t\ 
  \mhyph \bullet\]
\noindent where \(\bullet\) holds the place of a proof of \(\{t \simeq
\absu{\lambda}{x}{x}\}\).
However, Cedille's type theory is intensional, so from the assumption that \(t\)
behaves like the identity function on terms of type \(S\) we \emph{cannot} conclude
that \(t\) is equal to \(\absu{\lambda}{x}{x}\).
Since Cedille's operational semantics and definitional equality
is over untyped terms, our assumption regarding the behavior of \(t\) on terms
of type \(S\) gives us no guarantees about the behavior of \(t\) on terms of
other types, and thus no guarantee about the \emph{intensional} structure of \(t\).

The trick used in the definition of \(\mathit{intrCast}\) is rather to give an
\emph{extrinsic typing} of the identity function, using the typing and
property of \(t\).
In the body, we use \(\mathit{extView}\) with an erased proof that assumes an
arbitrary \(x: S\) and constructs a view of \(x\) having type \(T\) using the
typing of \(t\ x\) and the proof \(t'\ x\) that \(\{t\ x \simeq x\}\).
So rather than showing \(t\) is \(\absu{\lambda}{x}{x}\), we are showing that
\(t\) justifies giving \(\absu{\lambda}{x}{x}\) the type \(S \to T\).

\subsubsection{Casts form a preorder on types.}
\begin{figure}
  \small
  \centering
\begin{verbatim}
castRefl ◂ ∀ S: ★. Cast ·S ·S
= Λ S. intrCast -(λ x. x) -(λ _. β) .

castTrans ◂ ∀ S: ★. ∀ T: ★. ∀ U: ★. Cast ·S ·T ➾ Cast ·T ·U ➾ Cast ·S ·U
= Λ S. Λ T. Λ U. Λ c1. Λ c2.
intrCast -(λ x. elimCast -c2 (elimCast -c1 x)) -(λ x. β) .

castUnique ◂ ∀ S: ★. ∀ T: ★. ∀ c1: Cast ·S ·T. ∀ c2: Cast ·S ·T. { c1 ≃ c2 }
= Λ S. Λ T. Λ c1. Λ c2. ρ ς (eqCast -c1) @x.{ x ≃ c2 } - eqCast -c2 .
\end{verbatim}
  \caption{Casts form a preorder (\texttt{cast.ced})}
  \label{fig:cast-preorder}
\end{figure}

Recall that a preorder \((S,\sqsubseteq)\) consists of a set \(S\) and a
reflexive and transitive binary relation \(\sqsubseteq\) over \(S\).
In the proof-relevant setting of type theory, establishing that a relation on
types induces a preorder requires that we also show,
for all types \(T_1\) and \(T_2\), that proofs of \(T_1 \sqsubseteq T_2\) are unique
--- otherwise, we might only be working in a category.
We now show that \(\mathit{Cast}\) satisfies all three of these properties.

\begin{theorem}
  \label{thm:cast-preorder}
  \(\mathit{Cast}\) induces a preorder (or thin category) on Cedille types.
\end{theorem}
\begin{proof}
  Figure~\ref{fig:cast-preorder} gives the proofs in Cedille of reflexivity
  (\(\mathit{castRefl}\)), transitivity (\(\mathit{castTrans}\)), and uniqueness
  (\(\mathit{castUnique}\)) for \(\mathit{Cast}\).
\end{proof}

\subsubsection{Monotonicity.}
\label{sec:mono}
\begin{figure}[h]
  \centering
  \small
\begin{verbatim}
module mono .

import cast .

Mono ◂ (★ ➔ ★) ➔ ★
= λ F: ★ ➔ ★. ∀ X: ★. ∀ Y: ★. Cast ·X ·Y ➔ Cast ·(F ·X) ·(F ·Y) .
\end{verbatim}
  \caption{Monotonicity (\texttt{mono.ced})}
  \label{fig:mono}
\end{figure}

Monotonicity of a type scheme \(F : \star \to \star\) in this preorder is
defined as a lifting, for all types \(S\) and \(T\), of any cast from \(S\) to
\(T\) to a cast from \(F \cdot S\) to \(F \cdot T\).
This is \(\mathit{Mono}\) in Figure~\ref{fig:mono}.
In the subsequent derivations of Scott and Parigot encodings, we shall omit
the details of monotonicity proofs; once the general principle behind them is
understood, these proofs are mechanical and do not provide further insight into
the encoding.
Although these are somewhat burdensome to write by hand, their repetitive nature
makes us optimistic such proofs can be automated for the signatures of standard
datatypes.

\begin{figure}[h]
  \centering
  \small
\begin{verbatim}
NatF ◂ ★ ➔ ★ = λ N: ★. ∀ X: ★. X ➔ (N ➔ X) ➔ X.

monoNatF ◂ Mono ·NatF
= Λ X. Λ Y. λ c.
  intrCast
    -(λ n. Λ Z. λ z. λ s. n ·Z z (λ r. s (elimCast -c r)))
    -(λ n. β).
\end{verbatim}
  \caption{Monotonicity for \(\mathit{NatF}\)}
  \label{fig:mono-natf}
\end{figure}

We give an example in Figure~\ref{fig:mono-natf} to highlight the method.
Type scheme \(\mathit{NatF}\) is the impredicative encoding of the signature for
natural numbers.
To prove that it is monotonic, we assume arbitrary types \(X\) and \(Y\) such
that there is a cast \(c\) from the former to the latter, and must exhibit a
cast from \(\mathit{NatF} \cdot X\) to \(\mathit{NatF} \cdot Y\).
We do this using \(\mathit{intrCast}\) on a function that is
\emph{definitionally} equal (\(\beta\eta\)-convertible modulo erasure) to the
identity function.

After assuming \(n : \mathit{NatF} \cdot X\), we introduce a term of type
\(\mathit{NatF} \cdot Y\) by abstracting over a type \(Z\) and terms \(z : Z\)
and \(s : Y \to Z\).
Instantiating the type argument of \(n\) with \(Z\), the second term argument we
need to provide must have type \(X \to Z\).
For this, we precompose \(s\) with the assumed cast \(c\) from \(X\) to \(Y\).

% In the derivation of Scott naturals in Section~\ref{sec:scott-nat}, we shall
% detail some proofs that certain type schemes are monotonic, and subsequently
% omit discussion of similar proofs.
% This is because, once the general principle is understood, most of these
% proofs are mechanical and do not shed further light on the datatype encoding.

\subsection{Translating the proof of Theorem~\ref{thm:tarskip} to Cedille}
\label{sec:tarskitrans}

Figure~\ref{fig:recType} shows the translation of the proof of
Theorem~\ref{thm:tarskip} to Cedille, deriving monotone recursive
types.
Cedille's module system allows us to parametrize the module shown in
Figure~\ref{fig:recType} by the type scheme \(F\), and all definitions
implicitly take \(F\) as a type argument.
For the axiomatic presentation in Figure~\ref{fig:recType-ax}, we give \(F\)
explicitly.

\begin{figure}
  \small
\begin{verbatim}
module recType (F : ★ ➔ ★).

import cast .
import mono .

Rec ◂ ★ = ∀ X: ★. Cast ·(F ·X) ·X ➾ X.

recLB ◂ ∀ X: ★. Cast ·(F ·X) ·X ➾ Cast ·Rec ·X
= Λ X. Λ c. intrCast -(λ x. x ·X -c) -(λ x. β) .

recGLB ◂ ∀ Y: ★. (∀ X: ★. Cast ·(F ·X) ·X ➾ Cast ·Y ·X) ➾ Cast ·Y ·Rec
= Λ Y. Λ u. intrCast -(λ y. Λ X. Λ c. elimCast -(u -c) y) -(λ x. β) .

recRoll ◂ Mono ·F ➾ Cast ·(F ·Rec) ·Rec
= Λ mono.
  recGLB ·(F ·Rec)
    -(Λ X. Λ c. castTrans ·(F ·Rec) ·(F ·X) ·X -(mono (recLB -c)) -c) .

recUnroll ◂ Mono ·F ➾ Cast ·Rec ·(F ·Rec)
= Λ mono. recLB ·(F ·Rec) -(mono (recRoll -mono)).
\end{verbatim}
  \caption{Monotone recursive types derived in Cedille (\texttt{recType.ced})}
  \label{fig:recType}
\end{figure}

As noted in Section~\ref{sec:tarski}, it is enough to require that the set of
\(f\)-closed elements (here, \(F\)-closed types) has a greatest lower bound.
In Cedille's meta-theory~\citep{stump18c}, types are interpreted as sets
of (\(\beta\eta\)-equivalence classes of) closed terms, and in particular
the meaning of an impredicative quantification \(\abs{\forall}{X}{\star}{T}\) is
the intersection of the meanings (under different assignments of meanings to the
variable \(X\)) of the body.
Such an intersection functions as the greatest lower bound, as we will
see.

The definition of \(\mathit{Rec}\) in Figure~\ref{fig:recType} expresses the
intersection of the set of all \(F\)-closed types.
This \(\mathit{Rec}\) corresponds to \(r\) in the proof of Theorem~\ref{thm:tarskip}.
Semantically, we are taking the intersection of all those sets \(X\) which
are \(F\)-closed.
So the greatest lower bound of the set of all $f$-closed
elements in the context of a partial order is translated to the intersection of
all \(F\)-closed types, where \(X\) being \(F\)-closed means there is a cast
from \(F \cdot X\) to \(X\).
We require just an erased argument of type \(\mathit{Cast} \cdot (F \cdot X)
\cdot X\).
By making the argument erased, we express the idea of taking the
intersection of sets of terms satisfying a property, and not a set of functions
that take a proof of the property as an argument.

\begin{theorem}
  \(\mathit{Rec}\) is the greatest lower bound of the set of \(F\)-closed types.
\end{theorem}
\begin{proof}
  In Figure~\ref{fig:recType}, definition \(\mathit{recLB}\) establishes that
  \(\mathit{Rec}\) is a lower bound of this set and \(\mathit{recGLB}\)
  establishes that it greater than or equal to any other lower bound, i.e., for
  any other lower bound \(Y\) there is a cast from \(Y\) to \(\mathit{Rec}\).
  For \(\mathit{recLB}\), assume we have an \(F\)-closed type \(X\) and some \(x
  : \mathit{Rec}\).
  It suffices to give a term of type \(X\) that is definitionally equal to \(x\).
  Instantiate the type argument of \(x\) to \(X\) and use the proof that \(X\)
  is \(F\)-closed as an erased argument.

  For \(\mathit{recGLB}\), assume we have some \(Y\) which is a lower bound
  of the set of all \(F\)-closed types, witnessed by \(u\), and a term \(y: Y\).
  It suffices to give a term of type \(\mathit{Rec}\) that is definitionally
  equal to \(y\).
  By the definition of \(\mathit{Rec}\), we assume an arbitrary type \(X\) that
  is \(F\)-closed, witnessed by \(c\), and must produce a term of type \(X\).
  Use the assumption \(u\) and \(c\) to cast \(y\) to the type \(X\), noting that
  abstraction over \(X\) and \(c\) is erased.
\end{proof}

In Figure~\ref{fig:recType}, \(\mathit{recRoll}\) implements part 1 of the proof
of Theorem~\ref{thm:tarskip}, and \(\mathit{recUnroll}\) implements part 2.
In \(\mathit{recRoll}\), we invoke the property that \(\mathit{Rec}\) contains
any other lower bound of the set of \(F\)-closed types in order to show
\(\mathit{Rec}\) contains \(F \cdot \mathit{Rec}\), and must show
that \(F \cdot \mathit{Rec}\) is included into any arbitrary \(F\)-closed type \(X\).
We do so using the fact that \(Rec\) is also a lower bound of this set
(\(\mathit{recLB}\)) and so is contained in \(X\), monotonicity of \(F\), and
transitivity of \(\mathit{Cast}\) with the assumption that \(F \cdot X\) is
contained in \(X\).
In \(\mathit{recUnroll}\), we use \(\mathit{recRoll}\) and monotonicity of \(F\)
to obtain that \(F \cdot Rec\) is \(F\)-closed, then use \(\mathit{recLB}\) to
conclude.
It is here we see the impredicativity noted earlier: in \(\mathit{recLB}\), we
instantiate the type argument of \(x : \mathit{Rec}\) to the
given type \(X\); in \(\mathit{recUnroll}\), the given type is \(F \cdot
\mathit{Rec}\).
This would not be possible in a predicative type theory.

\subsection{Operational semantics for \(\mathit{Rec}\)}
\label{sec:rec-opsem}

We conclude this section by giving the definitions of the constant-time
\(\mathit{roll}\) and \(\mathit{unroll}\) operators for recursive types in
Figure~\ref{fig:roll-unroll}.
The derivation of recursive types with these operators is summarized
axiomatically in Figure~\ref{fig:recType-ax}.

\begin{figure}
  \centering
  \small
\begin{verbatim}
roll ◂ Mono ·F ➾ F ·Rec ➔ Rec
= Λ m. elimCast -(recRoll -m) .

unroll ◂ Mono ·F ➾ Rec ➔ F ·Rec
= Λ m. elimCast -(recUnroll -m) .

_ ◂ { roll   ≃ λ x. x } = β.
_ ◂ { unroll ≃ λ x. x } = β.

recIso1 ◂ { λ x. roll (unroll x) ≃ λ x. x} = β.
recIso2 ◂ { λ x. unroll (roll x) ≃ λ x. x} = β.
\end{verbatim}
  \caption{Operators \(\mathit{roll}\) and \(\mathit{unroll}\)
    (\texttt{recType.ced})}
  \label{fig:roll-unroll}
\end{figure}

Operations \(\mathit{roll}\) and \(\mathit{unroll}\) are implemented simply by
using the elimination form for casts on resp.\ \(\mathit{recRoll}\) and
\(\mathit{recUnroll}\), assuming a proof \(m\) that \(F\) is monotonic.
By erasure, this means both operations erase to \((\absu{\lambda}{x}{x})\
\absu{\lambda}{x}{x}\), and thus they are both definitionally equal to
\(\absu{\lambda}{x}{x}\).
This is show in Figure~\ref{fig:roll-unroll} with two anonymous proofs
(indicated by \(\_\)) of equality types that hold by \(\beta\) alone.
This fact makes trivial the proof that these operators for recursive types
satisfy the desired laws.

\begin{theorem}
  For all \(F : \star \to \star\) and monotonicity witnesses \(m : \mathit{Mono}
  \cdot F\), function \(\mathit{roll} \cdot F\ \mhyph m : F \cdot (\mathit{Rec} \cdot F) \to
  \mathit{Rec} \cdot F\) has a two-sided inverse \(\mathit{unroll}\cdot F\ \mhyph m : \mathit{Rec} \cdot F
  \to F \cdot (\mathit{Rec} \cdot F)\).
\end{theorem}
\begin{proof}
  By definitional equality; see \(\mathit{recIso1}\) and \(\mathit{recIso2}\) in Figure~\ref{fig:roll-unroll}.
\end{proof}

\begin{figure}
  \centering
  \[
    \begin{array}{c}
      \infer{
      \Gamma \vdash \mathit{Rec} \cdot F \tpsynth \star
      }{
      \Gamma \vdash F \tpsynth \star \to \star
      }
      \\ \\
      \begin{array}{cc}
        \infer{
        \Gamma \vdash \mathit{roll}\cdot F\ \mhyph t
        \tpsynth F \cdot (\mathit{Rec} \cdot F) \to \mathit{Rec} \cdot F
        }{
        \Gamma \vdash F \tpsynth \star \to \star
        \quad \Gamma \vdash t \tpcheck \mathit{Mono} \cdot F
        }
        &
          \infer{
          \Gamma \vdash \mathit{unroll} \cdot F\ \mhyph t \tpsynth \mathit{Rec}
          \cdot F \to F \cdot (\mathit{Rec} \cdot F)
          }{
          \Gamma \vdash F \tpsynth \star \to \star
          \quad \Gamma \vdash t \tpcheck \mathit{Mono} \cdot F
          }
      \end{array}
      \\ \\
      \begin{array}{lll}
        |\mathit{roll}|
        & =
        & (\absu{\lambda}{x}{x})\ \absu{\lambda}{x}{x}
        \\ |\mathit{unroll}|
        & =
        & (\absu{\lambda}{x}{x})\ \absu{\lambda}{x}{x}
      \end{array}
    \end{array}
  \]
  \caption{Monotone recursive types, axiomatically}
  \label{fig:recType-ax}
\end{figure}

We remark that, given the erasures of \(\mathit{roll}\) and \(\mathit{unroll}\),
the classification of \(\mathit{Rec}\) as either being iso-recursive or
equi-recursive is unclear.
On the one hand, \(\mathit{Rec} \cdot F\) and its one-step unrolling are not
definitionally equal types and require explicit operations to pass between the
two.
On the other hand, their \emph{denotations} as sets of \(\beta\eta\)-equivalence
classes of untyped lambda terms are equal, and in intensional type theories with
iso-recursive types it is not usual that the equation \(\mathit{roll}\
(\mathit{unroll}\ t) = t\) holds by the operational semantics (however, the
equation \(\mathit{unroll}\ (\mathit{roll}\ t) = t\) should, unless one is
satisfied with Church encodings).
Instead, we view \(\mathit{Rec}\) as a synthesis of these two formulations of
recursive types.

\section{Datatypes and recursion schemes}
\label{sec:datatypes-and-recursion-schemes}
Before we proceed with the application of derived recursive types to encodings of
datatypes with induction in Cedille, we first elaborate on the close connection
between datatypes, structured recursion schemes, and impredicative encodings.
An inductive datatype $D$ can be understood semantically as the least
fixpoint of a signature functor $F$.
Together with \(D\) comes a generic constructor, \(\mathit{inD} : F \cdot D \to
D\), which we can understand as building a new value of \(D\) from an
``\(F\)-collection'' of predecessors.
For example, the datatype \(\mathit{Nat}\) of natural numbers has the signature
\(\abs{\lambda}{X}{\star}{1 + X}\), where \(+\) is the binary coproduct type
constructor and \(1\) is the unitary type.
The more familiar constructors \(\mathit{zero} : \mathit{Nat}\) and \(\mathit{suc} :
\mathit{Nat} \to \mathit{Nat}\) can be merged together into a single constructor
\(\mathit{inNat} : (1 + \mathit{Nat}) \to \mathit{Nat}\).

What separates our derived monotone recursive types (which also constructs a
least fixpoint of \(F\)) and inductive datatypes is, essentially, the difference
between preorder theory and category theory: 
\emph{proof relevance}.
In moving from the first setting to the second, we observe the following
correspondences.
\begin{itemize}
\item The ordering corresponds to \emph{morphisms}.
  Here, this means working with functions \(S \to T\), not type inclusions
  \(\mathit{Cast} \cdot S \cdot T\).
  While there is at most one witness of an inclusion from one type to another,
  there of course may be multiple functions.
  
\item Monotonicity corresponds to \emph{functoriality}.
  Here, this means that a type scheme \(F\) comes with an operation
  \(\mathit{fmap}\) that lifts, for all types \(S\) and \(T\), functions \(S \to
  T\) to functions \(F \cdot S \to F \cdot T\). This lifting must respect
  identity and composition of functions.
  We give a formal definition in Cedille of functors later in
  Section~\ref{sec:functor-sigma}.
  
\item Where we had \(F\)-closed sets, we now have \emph{\(F\)-algebras}.
  Here, this means a type \(T\) together with a function \(t : F \cdot T \to
  T\).
\end{itemize}

Carrying the correspondence further, in Section~\ref{sec:tarskitrans} we proved
that \(\mathit{Rec} \cdot F\) is a lower bound of the set of \(F\)-closed types.
The related property for a datatype \(D\) with signature functor \(F\)
is the existence of a \emph{iteration operator}, \(\mathit{foldD}\), satisfying
both a typing and (because of the proof-relevant setting) a computation law.
This is shown in Figure~\ref{fig:scheme-iter}.

\begin{figure}[!h]
  \centering
  \[
    \begin{array}{c}
      \infer{
      \Gamma \vdash \mathit{foldD} \cdot T\ t \tpsynth D \to T
      }{
      \Gamma \vdash T \tpsynth \star
      \quad \Gamma \vdash t \tpcheck F \cdot T \to T
      }
      \\ \\
      |\mathit{foldD}\ t\ (\mathit{inD}\ t')| \rightsquigarrow |t\ (\mathit{fmap}\ (\mathit{foldD}\ t)\ t')|
    \end{array}
  \]
  \caption{Generic iteration scheme}
  \label{fig:scheme-iter}
\end{figure}

For the typing law, we read \(T\) as the type of results we wish to iteratively
compute and \(t\) as a function that constructs a result from an
\(F\)-collection of previous results recursively computed from predecessors.
This reading is confirmed by the computation law, which states that for all
\(T\), \(t\), and \(t'\), the function \(\mathit{foldD}\ t\) acts on
\(\mathit{inD}\ t'\) by first making recursive calls on the subdata of \(t'\)
(accessed using \(\mathit{fmap}\)) then using \(t\) to compute a result from this.
Instantiating \(F\) with the signature for \(\mathit{Nat}\), and working through
standard isomorphisms, we can specialize the typing and computation laws of the
generic iteration scheme to the usual laws for iteration over \(\mathit{Nat}\),
shown in Figure~\ref{fig:scheme-iter-nat}.

\begin{figure}[h]
  \centering
  \[
    \begin{array}{c}
      \infer{
      \Gamma \vdash \mathit{foldNat} \cdot T\ t_1\ t_2 \tpsynth \mathit{Nat} \to T
      }{
      \Gamma \vdash T \tpsynth \star
      \quad \Gamma \vdash t_1 \tpcheck T
      \quad \Gamma \vdash t_2 \tpcheck T \to T
      }
      \\ \\
      \begin{array}{lcl}
        |\mathit{foldNat}\ t_1\ t_2\ \mathit{zero}|
        & \rightsquigarrow
        & |t_1|
        \\ |\mathit{foldNat}\ t_1\ t_2\ (\mathit{suc}\ t)|
        & \rightsquigarrow
        & |t_2\ (\mathit{foldNat} \cdot T\ t_1\ t_2\ t)|
      \end{array}
    \end{array}
  \]
  \caption{Iteration scheme for \(\mathit{Nat}\)}
  \label{fig:scheme-iter-nat}
\end{figure}

Following the approach of \citet{Ge14_Church-Scott-Encoding}, we present lambda
encodings as solutions to structured recursion schemes over datatypes and
evaluate them by how well they simulate the computation laws for these schemes.
Read Figure~\ref{fig:scheme-iter} as a collection of constraints with unknowns \(D\),
\(\mathit{foldD}\), and \(\mathit{inD}\).
From these constraints, we can calculate an encoding of \(D\) in
System~F\(^{\omega}\) that is a variant of the Church encoding.
This is shown in Figure~\ref{fig:data-iter}.
For comparison, the figure also shows the familiar Church encoding of naturals,
which we may similarly read from the iteration scheme for \(\mathit{Nat}\).

\begin{figure}[h]
  \centering
  \[
    \begin{array}{lll}
      D
      & =
      & \abs{\forall}{X}{\star}{(F \cdot X \to X) \to X}
      \\
      \mathit{foldD}
      & =
      & \absu{\Lambda}{X}{\absu{\lambda}{a}{\absu{\lambda}{x}{x \cdot X\ a}}}
      \\
      \mathit{inD}
      & =
      & \absu{\lambda}{x}{\absu{\Lambda}{X}{\absu{\lambda}{a}{a\ (\mathit{fmap} \cdot
        D \cdot X\ (\mathit{foldD}\ \cdot X\ a)\ x)}}}
      \\ \\
      \mathit{Nat}
      & =
      & \abs{\forall}{X}{\star}{X \to (X \to X) \to X}
      \\ \mathit{foldNat}
      & =
      &
        \absu{\Lambda}{X}{\absu{\lambda}{z}{\absu{\lambda}{s}{\absu{\lambda}{n}{n
        \cdot X\ z\ s}}}}
      \\ \mathit{zero}
      & =
      & \absu{\Lambda}{X}{\absu{\lambda}{z}{\absu{\lambda}{s}{z}}}
      \\ \mathit{suc}
      & =
      &
        \absu{\lambda}{n}{\absu{\Lambda}{X}{\absu{\lambda}{z}{\absu{\lambda}{s}{s\
        (\mathit{foldNat} \cdot X\ z\ s\ n)}}}}
    \end{array}
  \]
  \caption{Church encoding of \(D\) and \(\mathit{Nat}\)}
  \label{fig:data-iter}
\end{figure}

For the typing law, we let \(D\) be the type of functions, polymorphic
in \(X\), that take functions of type \(F \cdot X \to X\) to a result of type \(X\).
For the iterator \(\mathit{foldD}\), we provide the given \(a : F \cdot X
\to X\) as an argument to the encoding.
Finally, we use the right-hand side of the computation law to give a definition for 
constructor \(\mathit{inD}\), and we can confirm that \(\mathit{inD}\) has type
\(F \cdot D \to D\).
Thus, the Church encoding of \(D\) arises as a direct solution to the iteration
scheme for \(D\) in (impredicative) polymorphic lambda calculi.

Notice that with the definitions of \(\mathit{inD}\) and
\(\mathit{foldD}\), we simulate the computation law for iteration in a
constant number of \(\beta\)-reduction steps.
For call-by-name operational semantics, we see that 
\[
  \begin{array}{clclcl}
    & |\mathit{foldD}\ t\ (\mathit{inD}\ t')|
    & \rightsquigarrow
    & (\absu{\lambda}{x}{x\ |t|})\ |\mathit{inD}\ t'|
    & \rightsquigarrow
    & |\mathit{inD}\ t'\ t|
    \\ \rightsquigarrow
    & (\absu{\lambda}{a}{a\ |(\mathit{fmap}\ (\mathit{foldD}\ a)\ t')|})\ |t|
    & \rightsquigarrow
    & |t\ (\mathit{fmap}\ (\mathit{foldD}\ t)\ t')|
  \end{array}
\]
For call-by-value semantics, we would assume that \(t\) and \(t'\) are values
and first reduce \(\mathit{inD}\ t'\).

The issue of inefficiency in computing predecessor for Church naturals has an
analogue for an arbitrary datatype \(D\) that only supports iteration.
The destructor for datatype \(D\) is an operator \(\mathit{outD}\) satisfying
the typing and computation law of Figure~\ref{fig:destruct-laws1}.
\begin{figure}[h]
  \centering
  \[
    \begin{array}{cc}
      \infer{
      \Gamma \vdash \mathit{outD} : D \to F \cdot D
      }{
      }
      & \quad \quad
      |\mathit{outD}\ (\mathit{inD}\ t)| \rightsquigarrow |t|
    \end{array}
  \] 
  \caption{Generic destructor}
  \label{fig:destruct-laws1}
\end{figure}

With \(\mathit{foldD}\), we can define a candidate for the destructor that
satisfies the desired typing as \(\mathit{outD} = \mathit{foldD}\
(\mathit{fmap}\ \mathit{inD})\).
However, this definition of \(\mathit{outD}\) does not efficiently simulate the
computation law.
By definitional equality alone, we have only
\[|\mathit{outD}\ (\mathit{inD}\ t)| \rightsquigarrow^* |\mathit{fmap}\ \mathit{inD}\
  (\mathit{fmap}\ \mathit{outD}\ t)|\]
which means we recursively destruct predecessors of \(t\) only to reconstruct them
with \(\mathit{inD}\).
In particular, if \(t\) is a variable the recursive call becomes stuck, and we
cannot reduce further to obtain a right-hand side of \(t\).

\subsection{Characterizing datatype encodings}
\label{sec:data-char}

Throughout the remainder of this paper, we will give a thorough characterization
of our Scott and Parigot encodings.
Specifically, in addition to the iteration scheme and destructor discussed
above, we will later introduce the \emph{case-distinction} and
\emph{primitive recursion} schemes to explain the definitions of the
Scott and Parigot encodings, respectively.
Then, for each combination of encoding and scheme we consider two properties: the
efficiency (under call-by-name and call-by-value semantics) with which we can
simulate the computation law of the scheme with the encoding, and the
provability of the \emph{extensionality} law for that scheme using the derived
induction principle for the encoding. 

\begin{figure}[h]
  \centering
  \small
\begin{verbatim}
module data-char/iter-typing (F: ★ ➔ ★) .

Alg ◂ ★ ➔ ★ = λ X: ★. F ·X ➔ X .

Iter ◂ ★ ➔ ★ = λ D: ★. ∀ X: ★. Alg ·X ➔ D ➔ X .
\end{verbatim}
  \caption{Iteration typing (\texttt{data-char/iter-typing.ced})}
  \label{fig:iter-typing}
\end{figure}

We now detail the criteria we shall use, and the corresponding Cedille
definitions, for the iteration scheme and destructor.
This begins with Figure~\ref{fig:iter-typing}, which takes as a module parameter
a type scheme \(F : \star \to \star\) and gives type definitions for the typing
law of the iteration scheme.
Type family \(\mathit{Alg}\) gives the shape of the types of functions used in
iteration, and family \(\mathit{Iter}\) gives the shape of the type of the
combinator \(\mathit{foldD}\) itself.

\paragraphb{Iteration scheme.}

\begin{figure}
  \centering
  \small
\begin{verbatim}
module data-char/iter
  (F: ★ ➔ ★) (fmap: ∀ X: ★. ∀ Y: ★. (X ➔ Y) ➔ F ·X ➔ F ·Y)
  (D: ★) (inD: F ·D ➔ D).

import data-char/iter-typing ·F .

AlgHom ◂ Π X: ★. Alg ·X ➔ (D ➔ X) ➔ ★
= λ X: ★. λ a: Alg ·X. λ h: D ➔ X.
  ∀ xs: F ·D. { h (inD xs) ≃ a (fmap h xs) } .

IterBeta ◂ Iter ·D ➔ ★
= λ foldD: Iter ·D.
  ∀ X: ★. ∀ a: Alg ·X. AlgHom ·X a (foldD a) .

IterEta ◂ Iter ·D ➔ ★
= λ foldD: Iter ·D.
  ∀ X: ★. ∀ a: Alg ·X. ∀ h: D ➔ X. AlgHom ·X a h ➔
  Π x: D. { h x ≃ foldD a x } .
\end{verbatim}
  \caption{Iteration characterization (\texttt{data-char/iter.ced})}
  \label{fig:iter-laws}
\end{figure}

Figure~\ref{fig:iter-laws} lists the computation and extensionality laws for the
iteration scheme.
For the module parameters, read \(\mathit{fmap}\) as the functorial operation lifting
functions over type scheme \(F\), \(D\) as the datatype, and \(\mathit{inD}\) as
the constructor.
\(\mathit{IterBeta}\) expresses the computation law using the auxiliary
definition \(\mathit{AlgHom}\) (the category-theoretic notion of an
\(F\)-algebra homomorphism from \(\mathit{inD}\)): for all \(\mathit{xs} : F \cdot D\), 
\(X\), and \(a : \mathit{Alg} \cdot X\), we have that \(\mathit{foldD}\ a\ (\mathit{inD}\
\mathit{x})\) is propositionally equal to \(a\ (\mathit{fmap}\ (\mathit{foldD}\ a)\
\mathit{xs})\).
For all datatype encodings and recursion schemes, we will be careful to note
whether the computation law is efficiently simulated (under both call-by-name
and call-by-value operational semantics) by the implementation we give for that
scheme.

The extensionality law, given by \(\mathit{IterEta}\), expresses the property that
a candidate for the combinator \(\mathit{foldD}\) for iteration is a
\emph{unique} solution to the computation law. 
More precisely, if there is any other function \(h : D \to X\) that satisfies
the computation law with respect to some \(a : \mathit{Alg} \cdot X\), then
\(\mathit{foldD}\ a\) and \(h\) are equal up to function extensionality.
Extensionality laws are proved with induction.

\begin{figure}[h]
  \centering
  \small
\begin{verbatim}
module data-char/destruct (F: ★ ➔ ★) (D: ★) (inD : F ·D ➔ D) .

Destructor ◂ ★ = D ➔ F ·D .

Lambek1 ◂ Destructor ➔ ★
= λ outD: Destructor. Π xs: F ·D. { outD (inD xs) ≃ xs } .

Lambek2 ◂ Destructor ➔ ★
= λ outD: Destructor. Π x: D. { inD (outD x) ≃ x } .
\end{verbatim}
  \caption{Laws for the datatype destructor (\texttt{data-char/destruct.ced})}
  \label{fig:destruct-laws}
\end{figure}

\paragraphb{Destructor.}
In Figure~\ref{fig:destruct-laws}, \(\mathit{Destructor}\) gives the type
of the generic data destructor, and \(\mathit{Lambek1}\) and
\(\mathit{Lambek2}\) together state that the property that candidate destructor
\(\mathit{outD}\) is a two-sided inverse of the constructor \(\mathit{inD}\).
The names for these properties come from \emph{Lambek's lemma}
\cite{Lam68_A-Fixpoint-Theorem-for-Complete-Categories}, which states that the
action of the initial algebra is an isomorphism.
For all encodings of datatypes, we will be careful to note whether
\(\mathit{Lambek1}\) (the computation law) is efficiently simulated by our
solution for the encoding's destructor under call-by-name and call-by-value
operational semantics.
As we noted earlier, for the generic Church encoding the solution for
\(\mathit{outD}\) is \emph{not} an efficient one.
Proofs of \(\mathit{Lambek2}\) (the extensionality law) will be given by
induction.

\section{Scott encoding}
\label{sec:scott}
The Scott encoding was first described in unpublished lecture notes by
\citet{Sco62_A-System-of-Functional-Abstraction}, and appears also in the work
of \citet{parigot89, parigot1992}.
Unlike Church naturals, an efficient predecessor function is
definable for Scott naturals \citep{parigot89}, but it is not known how to
express the type of Scott-encoded data in
System~F\footnote{In an unpublished note, \citet{ACP93_Types-for-Scott} claim to
present a type for Scott naturals in System~F that lacks an efficient
predecessor; our view is that this type is a form of Church encoding.}
\citep[point towards a negative
result]{SU99_Type-Fixpoints-Iteration-Recursion}.
Furthermore, it is not obvious how to define recursive functions over Scott
encodings without a general fixpoint mechanism for terms.

As a first application of monotone recursive types in Cedille -- and as a
warm-up for the generic derivations to come -- we show how to derive
Scott-encoded natural numbers with a weak form of induction.
By \emph{weak}, we mean that this form of induction does not provide an
inductive hypothesis, only a mechanism for proof by case analysis. 
In Section~\ref{sec:lr}, we will derive both primitive recursion and standard
induction for Scott encodings.

The Scott encoding can be seen as a solution to the case-distinction scheme in
polymorphic lambda calculi with recursive types.
We give the typing and computation laws for this scheme for \(\mathit{Nat}\) in
Figure~\ref{fig:nat-case}.
Unlike the iteration scheme, in the successor case the case-distinction
scheme provides direct access to the predecessor itself but no apparent form of
recursion.
\begin{figure}[h]
  \centering
  \[
    \begin{array}{c}
      \infer{
      \Gamma \vdash \mathit{caseNat} \cdot T\ t_1\ t_2 \tpsynth \mathit{Nat} \to T
      }{
      \Gamma \vdash T \tpsynth \star
      \quad \Gamma \vdash t_1 \tpcheck T
      \quad \Gamma \vdash t_2 \tpcheck \mathit{Nat} \to T
      }
      \\ \\
      \begin{array}{lll}
        |\mathit{caseNat} \cdot T\ t_1\ t_2\ \mathit{zero}|
        & \rightsquigarrow
        & |t_1|
        \\ |\mathit{caseNat} \cdot T\ t_1\ t_2\ (\mathit{suc}\ n)|
        & \rightsquigarrow
        &  |t_2\ n|
      \end{array}
    \end{array}
  \]  
  \caption{Typing and computation laws for case distinction on \(\mathit{Nat}\)}
  \label{fig:nat-case}
\end{figure}

Using \(\mathit{caseNat}\), we can define the predecessor function
\(\mathit{pred}\) for naturals.
\[\mathit{pred} = \mathit{caseNat} \cdot \mathit{Nat}\ \mathit{zero}\ \mathit{\absu{\lambda}{x}{x}}\]
Function \(\mathit{pred}\) then computes as follows over the constructors of
\(\mathit{Nat}\).
\[
  \begin{array}{lclcl}
    |\mathit{pred}\ \mathit{zero}|
    & \rightsquigarrow
    & |\mathit{zero}|
    \\ |\mathit{pred}\ (\mathit{suc}\ n)|
    & \rightsquigarrow
    & |(\absu{\lambda}{x}{x})\ n|
    & \rightsquigarrow
    & |n|
  \end{array}
\]
Thus we see that with an efficient simulation of \(\mathit{caseNat}\), we have
an efficient implementation of the predecessor function.

Using the same method as discussed in Section~\ref{sec:data-char}, from the
typing and computation laws we obtain the solutions for \(\mathit{Nat}\),
\(\mathit{caseNat}\), \(\mathit{zero}\), and \(\mathit{suc}\) given in
Figure~\ref{fig:scott-nat}.
For the definition of \(\mathit{Nat}\), the premises of the typing law mention
\(\mathit{Nat}\) itself, so a direct solution requires some form of recursive
types.
For readability, the solution in Figure~\ref{fig:scott-nat} uses isorecursive
types --- so \(|\mathit{unroll}(\mathit{roll}(t))| \rightsquigarrow |t|\) for
all \(t\).
It is then a mechanical exercise to confirm that the computation laws are
efficiently simulated by these definitions.

\begin{figure}[h]
  \centering
  \[
    \begin{array}{lcl}
      \mathit{Nat}
      & =
      & \absu{\mu}{N}{\absu{\forall}{X}{X \to (N \to X) \to X}}
      \\ \mathit{caseNat}
      & =
      &
        \absu{\Lambda}{X}{\absu{\lambda}{z}{\absu{\lambda}{s}{\absu{\lambda}{x}{\mathit{unroll}(x)
        \cdot X\ z\ s}}}}
      \\ \mathit{zero}
      & =
      & \mathit{roll}(\absu{\Lambda}{X}{\absu{\lambda}{z}{\absu{\lambda}{s}{z}}})
      \\ \mathit{suc}
      & =
      &
        \absu{\lambda}{x}{\mathit{roll}(\absu{\Lambda}{X}{\absu{\lambda}{z}{\absu{\lambda}{s}{s\ x}}})}
    \end{array}
  \]
  \caption{Scott naturals in System~F with isorecursive types}
  \label{fig:scott-nat}
\end{figure}

\subsection{Scott-encoded naturals, concretely}
\label{sec:scott-nat}

Our construction of Scott-encoded naturals supporting weak induction
consists of three stages.
In Figure~\ref{fig:scott-nat-1} we give the definition
of the non-inductive datatype signature \(\mathit{NatF}\) with its constructors.
In Figure~\ref{fig:scott-nat-2} we define a predicate \(\mathit{WkIndNatF}\) over
types \(N\) and terms \(n\) of type \(\mathit{NatF} \cdot N\) that says a
certain form of weak induction suffices to prove properties about \(n\).
Finally, in Figure~\ref{fig:scott-nat-3} the type \(\mathit{Nat}\) is given
using recursive types and the desired weak induction principle for
\(\mathit{Nat}\) is derived.

\begin{figure}[h]
  \centering
  \small
\begin{verbatim}
module scott/concrete/nat .

import view .
import cast .
import mono .
import recType .

NatF ◂ ★ ➔ ★ = λ N: ★. ∀ X: ★. X ➔ (N ➔ X) ➔ X.

zeroF ◂ ∀ N: ★. NatF ·N
= Λ N. Λ X. λ z. λ s. z.

sucF ◂ ∀ N: ★. N ➔ NatF ·N
= Λ N. λ n. Λ X. λ z. λ s. s n.

monoNatF ◂ Mono ·NatF = <..>
\end{verbatim}
  \caption{Scott naturals (part 1) (\texttt{scott/concrete/nat.ced})}
  \label{fig:scott-nat-1}
\end{figure}

\paragraphb{Signature \(\mathit{NatF}\).}
In Figure~\ref{fig:scott-nat-1}, type scheme \(\mathit{NatF}\) is the usual
impredicative encoding of the signature functor for natural numbers.
Terms \(\mathit{zeroF}\) and \(\mathit{sucF}\) are its constructors, quantifying
over the parameter \(N\); using the erasure rules
(Figure~\ref{fig:eraser}) we can confirm that these have the erasures
\(\absu{\lambda}{z}{\absu{\lambda}{s}{z}}\) and
\(\absu{\lambda}{n}{\absu{\lambda}{z}{\absu{\lambda}{s}{s\ n}}}\) --- these are
the constructors for Scott naturals in untyped lambda calculus.
The proof that \(\mathit{NatF}\) is monotone is omitted, indicated by
\texttt{<..>} in the figure (we detailed the proof in Section~\ref{sec:mono}).

\begin{figure}[h]
  \centering
  \small
\begin{verbatim}
WkIndNatF ◂ Π N: ★. NatF ·N ➔ ★
= λ N: ★. λ n: NatF ·N.
  ∀ P: NatF ·N ➔ ★. P (zeroF ·N) ➔ (Π m: N. P (sucF m)) ➔ P n .

zeroWkIndNatF ◂ ∀ N: ★. WkIndNatF ·N (zeroF ·N)
= Λ N. Λ P. λ z. λ s. z .

sucWkIndNatF ◂ ∀ N: ★. Π n: N. WkIndNatF ·N (sucF n)
= Λ N. λ n. Λ P. λ z. λ s. s n .
\end{verbatim}
  \caption{Scott naturals (part 2) (\texttt{scott/concrete/nat.ced})}
  \label{fig:scott-nat-2}
\end{figure}

\paragraphb{Predicate \(\mathit{WkIndNatF}\).}
We next define a predicate, parametrized by a type \(N\), over terms of type
\(\mathit{NatF} \cdot N\).
For such a term \(n\), \(\mathit{WkIndNat} \cdot N\ n\) is the property that, to
prove \(P\ n\) for arbitrary \(P : \mathit{NatF} \cdot N \to \star\), it
suffices to show certain cases for \(\mathit{zeroF}\) and \(\mathit{sucF}\).
\begin{itemize}
\item In the base case, we must show that \(P\) holds for \(\mathit{zeroF}\).
  
\item In the step case, we must show that for arbitrary \(m : N\) that
  \(P\) holds for \(\textit{sucF}\ m\).
\end{itemize}

Next in Figure~\ref{fig:scott-nat-2} are proofs \(\mathit{zeroWkIndNatF}\) and
\(\mathit{sucWkIndNatF}\), which show resp.\ that \(\mathit{zeroF}\) satisfies
the predicate \(\mathit{WkIndNatF}\) and \(\mathit{sucF}\ n\) satisfies this
predicate for all \(n\).
Notice that \(\mathit{zeroWkIndNatF}\) is definitionally equal to
\(\mathit{zeroF}\) and \(\mathit{sucWkIndNatF}\) is definitionally equal to
\(\mathit{sucF}\).
We can confirm this fact by having Cedille check that \(\beta\) proves they are
propositionally equal (\_ denotes an anonymous proof):
\begin{verbatim}
_ ◂ { zeroF ≃ zeroWkIndNatF } = β .
_ ◂ { sucF ≃ sucWkIndNatF } = β .
\end{verbatim}
This correspondence, first observed for Church encodings by \citet{leivant83},
between lambda-encoded data and the proofs that elements of the datatype satisfy
the datatype's induction principle, is an essential part of the recipe of
\cite{Stu18_From-Realizibility-to-Induction} for deriving inductive types in
CDLE.

\begin{figure}
  \centering
  \small
\begin{verbatim}
NatFI ◂ ★ ➔ ★ = λ N: ★. ι x: NatF ·N. WkIndNatF ·N x .

monoNatFI ◂ Mono ·NatFI = <..>

Nat ◂ ★ = Rec ·NatFI .
rollNat   ◂ NatFI ·Nat ➔ Nat = roll -monoNatFI .
unrollNat ◂ Nat ➔ NatFI ·Nat = unroll -monoNatFI .

zero ◂ Nat
= rollNat [ zeroF ·Nat , zeroWkIndNatF ·Nat ] .

suc ◂ Nat ➔ Nat
= λ m. rollNat [ sucF m , sucWkIndNatF m ] .

LiftNat ◂ (Nat ➔ ★) ➔ NatF ·Nat ➔ ★
= λ P: Nat ➔ ★. λ n: NatF ·Nat.
  ∀ v: View ·Nat β{ n }. P (elimView β{ n } -v) .

wkIndNat ◂ ∀ P: Nat ➔ ★. P zero ➔ (Π m: Nat. P (suc m)) ➔ Π n: Nat. P n
= Λ P. λ z. λ s. λ n.
  (unrollNat n).2 ·(LiftNat ·P) (Λ v. z) (λ m. Λ v. s m) -(selfView n) .
\end{verbatim}
  \caption{Scott naturals (part 3) (\texttt{scott/concrete/nat.ced})}
  \label{fig:scott-nat-3}
\end{figure}

\paragraphb{\(\mathit{Nat}\), the type of Scott naturals.}
Figure~\ref{fig:scott-nat-3} gives the third and final phase of the derivation
of Scott naturals.
The datatype signature \(\mathit{NatFI}\) is defined using a dependent
intersection, producing the subset of those terms of type
\(\mathit{NatF} \cdot N\) that are definitionally equal to some proof that they
satisfy the predicate \(\mathit{WkIndNatF} \cdot N\) (since
\(\mathit{WkIndNatF}\) is not an equational constraint, this
restriction is not trivialized by the Kleene trick).
Monotonicity of \(\mathit{NatFI}\) is given by \(\mathit{monoNatFI}\) (proof
omitted).

Using the recursive type former \(\mathit{Rec}\) derived in
Section~\ref{sec:rectypes}, we define \(\mathit{Nat}\) as the least fixpoint of
\(\mathit{NatFI}\), and specializing the rolling and unrolling operators to 
\(\mathit{NatFI}\) we obtain \(\mathit{rollNat}\) and \(\mathit{unrollNat}\).
The operators, along with the facts that \(|\mathit{zeroF}| =_{\beta\eta}
|\mathit{zeroWkIndNatF}|\) and \(|\mathit{sucF}| =_{\beta\eta}
|\mathit{sucWkIndNatF}|\), are then used to define the constructors
\(\mathit{zero}\) and \(\mathit{suc}\).
From the fact that \(|\mathit{rollNat}| =_{\beta\eta} \absu{\lambda}{x}{x}\),
and by the erasure of dependent intersection introductions, we see that the
constructors for \(\mathit{Nat}\) are in fact definitionally equal to the
corresponding constructors for \(\mathit{NatF}\).
We can again confirm this by using Cedille to check that \(\beta\) proves they are
propositionally equal.
\begin{verbatim}
_ ◂ { zero ≃ zeroF } = β .
_ ◂ { suc ≃ sucF } = β .
\end{verbatim}

\paragraphb{Weak induction for \(\mathit{Nat}\).}
Weak induction for \(\mathit{Nat}\), given by \(\mathit{wkIndNat}\) in
Figure~\ref{fig:scott-nat-3}, allows us to prove \(P\ n\) for arbitrary \(n :
\mathit{Nat}\) and \(P : \mathit{Nat} \to \star\) if we can show that \(P\)
holds of \(\mathit{zero}\) and that for arbitrary \(m\) we can construct a proof of
\(P\ (\mathit{suc}\ m)\).
However, there is a gap between this proof principle and the proof principle
\(\mathit{WkIndNatF} \cdot \mathit{Nat}\) associated to \(n\) --- the latter allows us to
prove properties over terms of type \(\mathit{NatF} \cdot \mathit{Nat}\), not
terms of type \(\mathit{Nat}\)!
To bridge this gap, we introduce a predicate transformer \(\mathit{LiftNat}\)
that takes properties of kind \(\mathit{Nat} \to \star\) to properties of kind
\(\mathit{NatF} \cdot \mathit{Nat} \to \star\).
For any \(n : \mathit{NatF} \cdot \mathit{Nat}\), the new property
\(\mathit{LiftNat} \cdot P\ n\) states that \(P\) holds for \(n\) if we
have a way of viewing \(n\) at type \(\mathit{Nat}\) (\(\mathit{View}\),
Figure~\ref{fig:view-ax}).

The key to the proof of weak induction for \(\mathit{Nat}\) is
that by using \(\mathit{View}\), the retyping operation of terms of type \(\mathit{NatF} \cdot
\mathit{Nat}\) to the type \(\mathit{Nat}\) is definitionally equal to \(\absu{\lambda}{x}{x}\), and the fact that
the constructors for \(\mathit{NatF}\) are definitionally equal to the
constructors for \(\mathit{Nat}\).
We elaborate further on this point.
Let \(z\) and \(s\) be resp.\ the assumed proofs of the base and inductive
cases.
From the second projection of the unrolling of \(n\) we have a proof of
\(\mathit{WkIndNatF} \cdot \mathit{Nat}\ (\mathit{unroll}\ n).1\).
Instantiating the predicate argument of this with \(\mathit{LiftNat} \cdot P\)
gives us three subgoals:
\begin{itemize}
\item \(\mathit{LiftNat} \cdot P\ (\mathit{zeroF} \cdot \mathit{Nat})\)

  Assuming \(v : \mathit{View} \cdot \mathit{Nat}\ \beta\{\mathit{zeroF}\}\),
  we wish to give a proof of \(P\ (\mathit{elimView}\ \beta\{\mathit{zeroF}\}\
  \mhyph v)\).
  This is convertible with the type \(P\ \mathit{zero}\) of \(z\), since
  \(|\mathit{elimView}\ \beta\{\mathit{zeroF}\}\ \mhyph v| =_{\beta\eta} |\mathit{zero}|\).
  
\item \(\abs{\Pi}{m}{\mathit{Nat}}{\mathit{LiftNat} \cdot P\ (\mathit{sucF}\ m)}\)

  Assume we have such an \(m\), and that \(v\) is a proof of \(\mathit{View}
  \cdot \mathit{Nat}\ \beta\{\mathit{sucF}\ m\}\).
  We are expected to give a proof of \(P\ (\mathit{elimView}\ \beta\{\mathit{sucF}\
  m.1\}\ \mhyph v)\).
  The expression \(s\ m\) has type \(P\ (\mathit{suc}\ m)\), which is
  convertible with that expected type.

\item \(\mathit{View} \cdot \mathit{Nat}\ \beta\{(\mathit{unrollNat}\ n).1\}\)

  This holds by \(\mathit{selfView}\ n\) of type \(\mathit{View} \cdot
  \mathit{Nat}\ \beta\{n\}\), since \(|\beta\{n\}| =_{\beta\eta} |\beta\{(\mathit{unrollNat}\ n).1\}|\).
\end{itemize}
The whole expression synthesizes type
\[\mathit{P}\ (\mathit{elimView}\ \beta\{(\mathit{unrollNat}\ n).1\}\ \mhyph
  (\mathit{selfView}\ n))\]
which is convertible with the expected type \(P\ n\).

\subsubsection{Computational and extensional character.}
\label{sec:scott-nat-comp}

As mentioned at the outset of this section, one of the crucial characteristics
of Scott-encoded naturals is that they may be used to efficiently simulate the
computation laws for case distinction.
We now demonstrate this is the case for the Scott naturals we have derived.
Additionally, we prove using weak induction that the solution we give for the combinator for case
distinction satisfies the corresponding \emph{extensionality} law, i.e., it is
the unique such solution up to function extensionality.

\begin{figure}[h]
  \small
  \centering
\begin{verbatim}
caseNat ◂ ∀ X: ★. X ➔ (Nat ➔ X) ➔ Nat ➔ X
= Λ X. λ z. λ s. λ n. (unrollNat n).1 z s .

caseNatBeta1 ◂ ∀ X: ★. ∀ z: X. ∀ s: Nat ➔ X. { caseNat z s zero ≃ z }
= Λ X. Λ z. Λ s. β .

caseNatBeta2
◂ ∀ X: ★. ∀ z: X. ∀ s: Nat ➔ X. ∀ n: Nat. { caseNat z s (suc n) ≃ s n }
= Λ X. Λ z. Λ s. Λ n. β .

pred ◂ Nat ➔ Nat = caseNat zero λ p. p .

predBeta1 ◂ { pred zero ≃ zero } = β .
predBeta2 ◂ ∀ n: Nat. { pred (suc n) ≃ n } = Λ n. β .

wkIndNatComp ◂ { caseNat ≃ wkIndNat } = β .
\end{verbatim}
  \caption{Computation laws for case distinction and predecessor (\texttt{scott/concrete/nat.ced})}
  \label{fig:scott-nat-comp}
\end{figure}

\paragraphb{Computational laws.}
The definition of the operator \(\mathit{caseNat}\) for case distinction is
given in Figure~\ref{fig:scott-nat-comp}, along with predecessor
\(\mathit{pred}\) (defined using \(\mathit{caseNat}\)) and proofs 
for both that they satisfy the desired computation laws (or \(\beta\)-laws) by definition.
As both \(\mathit{rollNat}\) and \(\mathit{unrollNat}\) erase to
\((\absu{\lambda}{x}{x})\ (\absu{\lambda}{x}{x})\)
(Figure~\ref{fig:recType-ax}), our overhead in simulating case distinction is
only a constant number of reductions.
With an efficient operation for case distinction, we obtain an efficient
predecessor \(\mathit{pred}\).

To confirm the efficiency of this implementation, we consider the erasure of the
right-hand side of the computation law for case distinction for the successor
case (\(\mathit{caseNatBeta2}\)). 
\[
  \underbrace{(\absu{\lambda}{z}{\absu{\lambda}{s}{\absu{\lambda}{n}{\underbrace{(\absu{\lambda}{x}{x})\
            (\absu{\lambda}{x}{x})}_{\mathit{unrollNat}}\ n\ z\ s}}})}_{\mathit{caseNat}}\ z\ s\
  (\underbrace{(\absu{\lambda}{n}{\underbrace{(\absu{\lambda}{x}{x})\
        (\absu{\lambda}{x}{x})}_{\mathit{rollNat}}\
      (\underbrace{(\absu{\lambda}{n}{\absu{\lambda}{z}{\absu{\lambda}{s}{s\
              n}}})}_{\mathit{sucF}}\ n)})}_{\mathit{suc}}\ n) 
\]
This both call-by-name and (under the assumption that \(n\) is a value)
call-by-value reduces to \(\mathit{s\ n}\) in a constant number of steps.

Additionally, it is satisfying to note that the computational content underlying the
weak induction principle is precisely that which underlies the case-distinction
operator \(\mathit{caseNat}\).
This is proven by \(\mathit{wkIndComp}\), which shows not just that they satisfy
the same computation laws, but in fact that the two terms are definitionally equal.

\paragraphb{Extensional laws.}
\begin{figure}[h]
  \centering
  \small
\begin{verbatim}
caseNatEta
◂ ∀ X: ★. ∀ z: X. ∀ s: Nat ➔ X.
  ∀ h: Nat ➔ X. { h zero ≃ z } ➾ (Π n: Nat. { h (suc n) ≃ s n }) ➾
  Π n: Nat. { caseNat z s n ≃ h n }
= Λ X. Λ z. Λ s. Λ h. Λ hBeta1. Λ hBeta2.
  wkIndNat ·(λ x: Nat. { caseNat z s x ≃ h x })
    (ρ hBeta1 @x.{ z ≃ x } - β)
    (λ m. ρ (hBeta2 m) @x.{ s m ≃ x } - β) .

reflectNat ◂ Π n: Nat. { caseNat zero suc n ≃ n }
= caseNatEta ·Nat -zero -suc -(λ x. x) -β -(λ m. β) .
\end{verbatim}
  \caption{Extensional laws for case distinction
    (\texttt{scott/concrete/nat.ced})}
  \label{fig:scott-nat-ext}
\end{figure}
Using weak induction, we can prove the extensionality law (or \(\eta\)-law) of the
case-distinction scheme.
This is \(\mathit{caseNatEta}\) in Figure~\ref{fig:scott-nat-ext}.
The precise statement of uniqueness is that, for every type \(X\) and terms \(z : X\)
and \(s : \mathit{Nat} \to X\), if there exists a function \(h : \mathit{Nat}
\to X\) satisfying the computation laws of the case-distinction scheme with
respect to \(z\) and \(s\), then \(h\) is extensionally equal to
\(\mathit{caseNat}\ z\ s\).

From uniqueness, we can obtain the proof \(\mathit{reflectNat}\) that using case
distinction with the constructors themselves reconstructs the given number.
The name for this is taken from the \emph{reflection law}
\citep{UV99_Primitive-Corecurusion-and-CoV-Coiteration} of the iteration scheme
for datatypes.
The standard formulation of the reflection law and the variation given by
\(\mathit{reflectNat}\) both express that the only elements of a datatype are
those generated by its constructors.
This idea plays a crucial role in the future derivations of (full) induction for
both the Parigot and Scott encodings, and is elaborated on in
Section~\ref{sec:parigot}.

\subsection{Scott-encoded data, generically}
\label{sec:scott-gen}

For the generic Scott encoding, we begin our discussion by phrasing the
case-distinction scheme \emph{generically} (meaning \emph{parametrically}) in a
signature \(F : \star \to \star\).
Let \(D\) be the datatype whose signature is \(F\) and whose constructor is
\(\mathit{inD} : F \cdot D \to D\).
Datatype \(D\) satisfies the case-distinction scheme if there exists an
operator \(\mathit{caseD}\) satisfying the typing and computation laws
listed in Figure~\ref{fig:scheme-case}.

\begin{figure}[h]
  \centering
  \[
    \begin{array}{c}
      \infer{
      \Gamma \vdash \mathit{caseD} \cdot T\ t \tpsynth D \to T
      }{
      \Gamma \vdash T \tpsynth \star
      \quad \Gamma \vdash t \tpcheck F \cdot D \to T
      }
      \\ \\
      |\mathit{caseD} \cdot T\ t\ (\mathit{inD}\ t')| \rightsquigarrow |t\ t'|
    \end{array}
  \]
  \caption{Generic case-distinction scheme}
  \label{fig:scheme-case}
\end{figure}

We understand the computation law as saying that when acting on data constructed
with \(\mathit{inD}\), the case-distinction scheme gives to its function
argument \(t\) the revealed subdata \(t' : F \cdot D\) directly.
Notice that unlike the iteration scheme, for case distinction we do not require
that \(F\) comes together with an operation \(\mathit{fmap}\) as there is no
recursive call to be made on the predecessors.

From these laws, we can define \(\mathit{outD} : D \to F \cdot D\), the generic
destructor that reveals the \(F\)-collection of predecessors used to construct a
term of type \(D\).
\[\mathit{outD} = \mathit{caseD}\ (\absu{\lambda}{x}{x})\]
This satisfies the expected computation law for the destructor in a number of
steps that is constant with respect to \(t\).

In System F\(^\omega\) extended with isorecursive types, we can read an encoding
directly from these laws, resulting in the solutions for \(D\),
\(\mathit{caseD}\), and \(\mathit{inD}\) given in Figure~\ref{fig:data-discr}.
This is the generic Scott encoding, and is the basis for the
developments in Section~\ref{ssec:derive-scott}.

\begin{figure}[h]
  \centering
  \[
    \begin{array}{lll}
      D
      & =
      & \absu{\mu}{D}{\abs{\forall}{X}{\star}{(F \cdot D \to X) \to X}}
      \\ \mathit{caseD}
      & =
      & \absu{\Lambda}{X}{\absu{\lambda}{a}{\absu{\lambda}{x}{\mathit{unroll}(x) \cdot X\ a}}}
      \\ \mathit{inD}
      & =
      & \absu{\lambda}{x}{\mathit{roll}(\absu{\Lambda}{X}{\absu{\lambda}{a}{a\ x}})}
    \end{array}
  \]
  \caption{Generic Scott encoding of \(D\) in System \(\text{F}^\omega\) with
    isorecursive types}
  \label{fig:data-discr}
\end{figure}

\subsubsection{Characterization criteria.}
We formalize in Cedille the above description of the generic case-distinction
scheme in Figures~\ref{fig:data-char-case-typing} and \ref{fig:data-char-case}.
Definitions for the typing law of case distinction for datatype \(D\) are given in
Figure~\ref{fig:data-char-case-typing}, where the module is parametrized by a
type scheme \(F\) that gives the datatype signature.
The type family \(\mathit{AlgCase}\) gives the shape of the types of functions
used for case distinction, and \(\mathit{Case}\) gives the shape of the type of
operator \(\mathit{caseD}\) itself.

\begin{figure}[h]
  \centering
  \small
\begin{verbatim}
module data-char/case-typing (F: ★ ➔ ★) .

AlgCase ◂ ★ ➔ ★ ➔ ★
= λ D: ★. λ X: ★. F ·D ➔ X .

Case ◂ ★ ➔ ★
= λ D: ★. ∀ X: ★. AlgCase ·D ·X ➔ D ➔ X .
\end{verbatim}
  \caption{Case distinction typing (\texttt{data-char/case-typing.ced})}
  \label{fig:data-char-case-typing}
\end{figure}

\begin{figure}
  \centering
  \small
\begin{verbatim}
import data-char/iter-typing .

module data-char/case
  (F: ★ ➔ ★) (D: ★) (inD: Alg ·F ·D).

import data-char/case-typing ·F .

AlgCaseHom ◂ Π X: ★. AlgCase ·D ·X ➔ (D ➔ X) ➔ ★
= λ X: ★. λ a: AlgCase ·D ·X. λ h: D ➔ X.
  ∀ xs: F ·D. { h (inD xs) ≃ a xs } .

CaseBeta ◂ Case ·D ➔ ★
= λ caseD: Case ·D.
  ∀ X: ★. ∀ a: AlgCase ·D ·X. AlgCaseHom ·X a (caseD a) .

CaseEta ◂ Case ·D ➔ ★
= λ caseD: Case ·D.
  ∀ X: ★. ∀ a: AlgCase ·D ·X. ∀ h: D ➔ X. AlgCaseHom ·X a h ➔
  Π x: D. { h x ≃ caseD a x } .
\end{verbatim}
  \caption{Case distinction characterization (\texttt{data-char/case.ced})}
  \label{fig:data-char-case}
\end{figure}

In Figure~\ref{fig:data-char-case}, we take the signature \(F\) as
well as the datatype \(D\) and its constructor \(\mathit{inD}\) as parameters.
We import the definitions of Figure~\ref{fig:iter-typing} without specifying that
module's type scheme parameter, so in the type of \(\mathit{inD}\),
\(\mathit{Alg} \cdot F \cdot D\), we give this explicitly as \(F\).
For a candidate \(\mathit{caseD} : \mathit{Case} \cdot D\) for the operator for case
distinction, the property \(\mathit{CaseBeta}\ \mathit{caseD}\) states that it
satisfies the desired computation law, where the shape of the computation law is
given by \(\mathit{AlgCaseHom}\).
\(\mathit{CaseEta}\ \mathit{caseD}\) is the property that \(\mathit{caseD}\)
satisfies the extensionality law, i.e., any other function \(h : D \to X\) that
satisfies the computation law with respect to \(a : \mathit{AlgCase} \cdot D
\cdot X\) is extensionally equal to \(\mathit{caseD}\ a\).

\subsubsection{Generic Scott encoding.}
We now detail the generic derivation of Scott-encoded data supporting weak
induction.
The developments in this section are parametrized by a type scheme \(F : \star
\to \star\) that is monotonic (the curly braces around \(\mathit{mono}\)
indicate that it is an erased module parameter).
As we did for the concrete derivation of naturals,
the construction is separated into several phases.
In Figure~\ref{fig:scott-generic-1}, we give the unrefined signature
\(\mathit{DF}\) for Scott-encoded data and its constructor.
In Figure~\ref{fig:scott-generic-2}, we define the predicate \(\mathit{WkIndDF}\)
expressing that \(\mathit{DF}\) terms satisfy a certain weak induction principle, and
show that the constructor satisfies this predicate.
In Figure~\ref{fig:scott-generic-3}, we take the fixpoint of the refined
signature, defining the datatype \(D\), and prove weak induction for \(D\).

\label{ssec:derive-scott}
\begin{figure}[h]
  \centering
  \small
\begin{verbatim}
import mono .

module scott/encoding (F: ★ ➔ ★) {mono: Mono ·F} .

import view .
import cast .
import recType .
import utils .

import data-char/typing ·F .

DF ◂ ★ ➔ ★ = λ D: ★. ∀ X: ★. AlgCase ·D ·X ➔ X .

inDF ◂ ∀ D: ★. AlgCase ·D ·(DF ·D)
= Λ D. λ xs. Λ X. λ a. a xs .

monoDF ◂ Mono ·DF = <..>
\end{verbatim}
  \caption{Generic Scott encoding (part 1) (\texttt{scott/generic/encoding.ced})}
  \label{fig:scott-generic-1}
\end{figure}

\paragraphb{Signature \(\mathit{DF}\).}
In Figure~\ref{fig:scott-generic-1}, \(\mathit{DF}\) is the type scheme whose fixpoint
is the solution to the equation for \(D\) in Figure~\ref{fig:data-discr}.
Definition \(\mathit{inDF}\) is the polymorphic constructor for signature
\(\mathit{DF}\), i.e., the generic grouping together of the collection of
constructors for the datatype signature (e.g., \(\mathit{zeroF}\) and
\(\mathit{sucF}\), Figure~\ref{fig:scott-nat-1}).
Finally, \(\mathit{monoDF}\) is a proof that type scheme \(\mathit{DF}\) is
monotonic (definition omitted, indicated by \texttt{<..>}).

% The remaining definitions of Figure~\ref{fig:scott-generic-1} are given in preparation
% for \(\mathit{WkIndDF}\), the predicate parametrized by a type \(D\) that terms
% of type \(\mathit{DF} \cdot D\) support a weak form of induction.
% Compare these to the inductive case in the definition of \(\mathit{WkIndNatF}\)
% in Figure~\ref{fig:scott-nat-2}.

\paragraphb{Predicate \(\mathit{WkIndDF}\).}
\begin{figure}[h]
  \centering
  \small
\begin{verbatim}
WkPrfAlg ◂ Π D: ★. (DF ·D ➔ ★) ➔ ★
= λ D: ★. λ P: DF ·D ➔ ★. Π xs: F ·D. P (inDF xs) .

WkIndDF ◂ Π D: ★. DF ·D ➔ ★
= λ D: ★. λ x: DF ·D.
  ∀ P: DF ·D ➔ ★. WkPrfAlg ·D ·P ➔ P x .

inWkIndDF ◂ ∀ D: ★. WkPrfAlg ·D ·(WkIndDF ·D)
= Λ D. λ xs. Λ P. λ a. a xs .
\end{verbatim}
  \caption{Generic Scott encoding (part 2) (\texttt{scott/generic/encoding.ced})}
  \label{fig:scott-generic-2}
\end{figure}

In Figure~\ref{fig:scott-generic-2}, we give the definition of
\(\mathit{WkIndDF}\), the property (parametrized by type \(D\)) that terms of
type \(\mathit{DF} \cdot D\) satisfy a certain weak induction principle.
More precisely, \(\mathit{WkIndDF} \cdot D\ t\) is the type of proofs that, for
all properties \(P : \mathit{DF} \cdot D \to \star\), \(P\) holds for \(t\) if a
weak inductive proof can be given for \(P\).
The type of weak inductive proofs is \(\mathit{WkPrfAlg} \cdot D \cdot P\), to
be read ``weak \((F,D)\)-proof-algebras for \(P\)''.
A term \(a\) of this type takes an \(F\)-collection of \(D\) values and produces
a proof that \(P\) holds for the value constructed from this using \(\mathit{inDF}\).

In the concrete derivation of Scott naturals, the predicate
\(\mathit{WkIndNatF}\) (Figure~\ref{fig:scott-nat-2}) required terms of the
following types be given in proofs by weak induction:
\[
  \begin{array}{l}
    P\ (\mathit{zeroF} \cdot N)
    \\ \abs{\Pi}{m}{N}{P\ (\mathit{sucF}\ m)}
  \end{array}
\]
We can understand \(\mathit{WkPrfAlg}\) as combing these types together into a
single type, parametrized by the signature \(F\):
\[\abs{\Pi}{\mathit{xs}}{F \cdot D}{P\ (\mathit{inDF}\ xs)}\]

Next in the figure is \(\mathit{inWkIndDF}\), which for all types \(D\) is a
weak \((F,D)\)-proof-algebra for \(\mathit{WkIndDF} \cdot D\).
Put more concretely, it is a proof that every term of type \(\mathit{DF} \cdot
D\) constructed by \(\mathit{inDF}\) admits the weak induction
principle given by \(\mathit{WkIndDF} \cdot D\).
The corresponding definitions from Section~\ref{sec:scott-gen} are
\(\mathit{zeroWkIndNatF}\) and \(\mathit{sucWkIndNatF}\). 

\paragraphb{The Scott-encoded datatype \(D\).}
\begin{figure}[h]
  \centering
  \small
\begin{verbatim}
DFI ◂ ★ ➔ ★ = λ D: ★. ι x: DF ·D. WkIndDF ·D x .

monoDFI ◂ Mono ·DFI = <..>

D ◂ ★ = Rec ·DFI .
rollD   ◂ DFI ·D ➔ D = roll -monoDFI .
unrollD ◂ D ➔ DFI ·D = unroll -monoDFI .

inD ◂ AlgCase ·D ·D
= λ xs. rollD [ inDF xs , inWkIndDF xs ] .

LiftD ◂ (D ➔ ★) ➔ DF ·D ➔ ★
= λ P: D ➔ ★. λ x: DF ·D. ∀ v: View ·D β{ x }. P (elimView β{ x } -v) .

wkIndD ◂ ∀ P: D ➔ ★. (Π xs: F ·D. P (inD xs)) ➔ Π x: D. P x
= Λ P. λ a. λ x.
  (unrollD x).2 ·(LiftD ·P) (λ xs. Λ v. a xs) -(selfView x) .
\end{verbatim}
  \caption{Generic Scott encoding (part 3) (\texttt{scott/generic/encoding.ced})}
  \label{fig:scott-generic-3}
\end{figure}

With the inductivity predicate \(\mathit{WkIndDF}\) and weak proof algebra
\(\mathit{inWkIndDF}\) for it, we are now able to form a signature refining
\(\mathit{DF}\) such that its fixpoint supports weak induction (proof by cases).
Observe that the proof \(\mathit{inWkIndDF}\) in
Figure~\ref{fig:scott-generic-2} is definitionally equal to \(\mathit{inDF}\). 
\begin{verbatim}
_ ◂ { inDF ≃ inWkIndDF } = β .
\end{verbatim}
This effectively allows us to use dependent intersections to form, for all
\(D\), the subset of the type \(\mathit{DF} \cdot D\) whose elements satisfy
\(\mathit{WkIndDF} \cdot D\).
This is \(\mathit{DFI}\) in Figure~\ref{fig:scott-generic-3}.

Since \(\mathit{DFI}\) is monotonic (proof omitted), we can form the datatype
\(D\) as the fixpoint of \(\mathit{DFI}\) using \(\mathit{Rec}\), with rolling
and unrolling operations \(\mathit{rollD}\) and \(\mathit{unrollD}\) that are
definitionally equal to \(\absu{\lambda}{x}{x}\).
The constructor \(\mathit{inD}\) for \(D\) takes an \(F\)-collection of \(D\)
predecessors \(\mathit{xs}\) and constructs a value of type \(D\) using 
the fixpoint rolling operator, the constructor \(\mathit{inDF}\), and the
proof \(\mathit{inWkIndDF}\).
Note again that, by the erasure of dependent intersections, we have that
\(\mathit{inD}\) and \(\mathit{inDF}\) are definitionally equal.

\paragraphb{Weak induction for \(D\).}
As was the case for the concrete encoding of Scott naturals, we must now
bridge the gap between the desired weak induction principle, where we wish to
prove properties of kind \(D \to \star\), and what we are given by
\(\mathit{WkIndDF}\) (the ability to prove properties of kind \(\mathit{DF}
\cdot D \to \star\)).
This is achieved using the predicate transformer \(\mathit{LiftD}\) that maps
predicates over \(\mathit{D}\) to predicates over \(\mathit{DF} \cdot D\) by
requiring an additional assumption that the given \(x : \mathit{DF} \cdot D\)
can be viewed as having type \(D\).

The weak induction principle \(\mathit{wkIndD}\) for \(D\) states that a
property \(P\) holds for term \(x : D\) if we can provide a function \(a\)
which, when given an arbitrary \(F\)-collection of \(D\) predecessors, produces
a proof that \(P\) holds for the successor of this collection constructed from
\(\mathit{inD}\).
In the body of \(\mathit{wkIndD}\), we invoke the proof principle
\(\mathit{WkIndDF} \cdot D\ (\mathit{unroll}\ x).1\), given by
\((\mathit{unroll}\ x).2\), on the lifting of \(P\).
For the weak proof algebra, we apply the assumption \(a\) to the revealed
predecessors \(\mathit{xs}\).
This expression has type \(P\ (\mathit{inD}\ xs)\), and the expected type is
\(P\ (\mathit{elimView}\ \beta\{\mathit{inDF}\ xs\}\ \mhyph v)\).
These two types are convertible, since the two terms in question are
definitionally equal:
\[|\mathit{elimView}\ \beta\{\mathit{inDF}\ \mathit{xs}\}\ \mhyph v| =_{\beta\eta}
  |\mathit{inD}\ \mathit{xs}|\]
since in particular \(|\mathit{inDF}| =_{\beta\eta} |\mathit{inD}|\).
The whole expression, then, has type
\[P\ (\mathit{elimView}\ \beta\{(\mathit{unrollD}\ x).1\}\ \mhyph
  (\mathit{selfView}\ x))\]
which is convertible with the expected type \(P\ x\).

\subsubsection{Computational and extensional character.}
\label{sec:scott-gen-char}
We now analyze the properties of our generic Scott encoding.
In particular, we give the normalization guarantee for terms of type
\(D\) and confirm that we can give implementations of the case-distinction
scheme and destructor that both efficiently simulate their computation
laws and provably satisfy their extensionality laws.
Iteration and primitive recursion are treated in Section~\ref{sec:lr-gen-char}.

\begin{figure}
  \centering
  \small
\begin{verbatim}
import cast .
import mono .
import recType .
import utils .

module scott/generic/props
  (F: ★ ➔ ★) {mono: Mono ·F} .

import data-char/case-typing ·F .
import scott/generic/encoding ·F -mono .

normD ◂ Cast ·D ·(AlgCase ·D ·D ➔ D)
= intrCast -(λ x. (unrollD x).1 ·D) -(λ x. β) .

import data-char/case ·F ·D inD .

caseD ◂ Case ·D
= Λ X. λ a. λ x. (unrollD x).1 a .

caseDBeta ◂ CaseBeta caseD
= Λ X. Λ a. Λ xs. β .

caseDEta ◂ CaseEta caseD
= Λ X. Λ a. Λ h. λ hBeta.
  wkIndD ·(λ x: D. { h x ≃ caseD a x })
    (λ xs. ρ (hBeta -xs) @x.{ x ≃ a xs } - β) .

reflectD ◂ Π x: D. { caseD inD x ≃ x }
= λ x. ρ ς (caseDEta ·D -inD -(id ·D) (Λ xs. β) x) @y.{ y ≃ x } - β .
\end{verbatim}
  \caption{Characterization of \(\mathit{caseD}\)
    (\texttt{scott/generic/props.ced})}
  \label{fig:scott-props-case}
\end{figure}

\paragraphb{Normalization guarantee.}
Recall that Proposition~\ref{thm:cedille-termination} guarantees call-by-name
normalization for closed Cedille terms whose types can be included into some
function type.
The proof \(\mathit{normD}\) of Figure~\ref{fig:scott-props-case}
establishes the existence of a cast from \(D\) to \(\mathit{AlgCase} \cdot D
\cdot D \to D\), meaning that closed terms of type \(D\) satisfy this criterion.

\paragraphb{Case-distinction scheme.}
We next bring into scope the definitions for characterizing the case-distinction
scheme (Figure~\ref{fig:data-char-case}).
For our solution \(\mathit{caseD}\) in Figure~\ref{fig:scott-props-case},
\(\mathit{caseDBeta}\) proves it satisfies the computation law and
\(\mathit{caseDEta}\) proves it satisfies the extensionality law.
As we saw for the concrete example of Scott naturals in
Section~\ref{sec:scott-nat-comp}, the proof \(\mathit{caseDBeta}\) of the
computation law holds by \emph{definitional} equality, not just propositional
equality, since the propositional equality is proved by \(\beta\).
By inspecting the definitions of \(\mathit{inD}\) and \(\mathit{caseD}\), and
the erasures of \(\mathit{roll}\) and \(\mathit{unroll}\)
(Figure~\ref{fig:recType-ax}), we can confirm that in fact \(\mathit{caseD}\ t\
(\mathit{inD}\ t')\) reduces to \(t\) in a number of steps that is constant with
respect to \(t'\) under both call-by-name and call-by-value operational semantics
(for call-by-value semantics, we would first assume \(t'\) is a value).

For \(\mathit{caseDEta}\), we proceed by weak induction and must show that
\(h\ (\mathit{inD}\ \mathit{xs})\) is propositionally equal to \(\mathit{caseD}\
a\ (\mathit{inD}\ xs)\).
This follows from the assumption that \(h\) satisfies the computation law with
respect to \(a : \mathit{AlgCase} \cdot D \cdot X\).
As an expected result of uniqueness, \(\mathit{reflectD}\) shows that applying
\(\mathit{caseD}\) to the constructor produces a function extensionally equal to
the identity function.

\begin{figure}[h]
  \centering
  \small
\begin{verbatim}
import data-char/destruct ·F ·D inD .

outD ◂ Destructor = caseD (λ xs. xs) .

lambek1D ◂ Lambek1 outD = λ xs. β .

lambek2D ◂ Lambek2 outD
= wkIndD ·(λ x: D. { inD (outD x) ≃ x }) (λ xs. β) .
\end{verbatim}
  \caption{Characterization of \(\mathit{outD}\)
    (\texttt{scott/generic/props.ced})}
  \label{fig:scott-props-out}
\end{figure}

\paragraphb{Destructor.}
In Figure~\ref{fig:scott-props-out}, we give the definition we proposed earlier
for the datatype destructor \(\mathit{outD}\).
The proof \(\mathit{lambek1D}\) establishes that \(|\mathit{outD}\ (\mathit{inD}\
t)|\) is definitionally equal to \(|t|\) for all terms \(t : F \cdot D\).
As \(\mathit{caseD}\) is an efficient implementation of case distinction,
we know that \(\mathit{outD}\) is an efficient destructor.
The proof \(\mathit{lambek2D}\) establishes the other side of the isomorphism
between \(D\) and \(F \cdot D\), and follows by weak induction.

\section{Parigot encoding}
\label{sec:parigot}
In this section we derive Parigot-encoded data with (non-weak) induction,
illustrating with a concrete example in Section~\ref{sec:parigot-nat} the 
techniques we use before proceeding with the generic derivation in
Section~\ref{ssec:parigot-gen}.
The Parigot encoding was first described by \citet{parigot88, parigot1992} for
natural numbers and later for a more general class of datatypes by
\citet{Ge14_Church-Scott-Encoding} (wherein it is called the \emph{Church-Scott}
encoding).
It is a combination of the Church and Scott encoding: it directly
supports recursion with access to previously computed results as well as to
predecessors.
For the inductive versions we derive in this section, this means that unlike the
weakly inductive Scott encoding of Section~\ref{sec:scott}, we have access to an
inductive hypothesis.
However, for the Parigot encoding this additional power comes at a cost: the
size of the encoding of natural number \(n\) is exponential in \(n\).

The Parigot encoding can be seen as a direct solution to the \emph{primitive recursion}
scheme in polymorphic lambda calculi with recursive types.
For natural numbers, the typing and computation laws for this scheme are given in
Figure~\ref{fig:nat-rec}. 
\begin{figure}
  \centering
  \[
    \begin{array}{c}
      \infer{
       \Gamma \vdash \mathit{recNat} \cdot T\ t_1\ t_2 \tpsynth \mathit{Nat} \to T
      }{
      \Gamma \vdash T \tpsynth \star
      \quad \Gamma \vdash t_1 \tpcheck T
      \quad \Gamma \vdash t_2 \tpcheck \mathit{Nat} \to T \to T
      }
      \\ \\
      \begin{array}{lcl}
        |\mathit{recNat} \cdot T\ t_1\ t_2\ \mathit{zero}|
        & \rightsquigarrow
        & |t_1|
        \\ |\mathit{recNat} \cdot T\ t_1\ t_2\ (\mathit{suc}\ n)|
        & \rightsquigarrow
        & |t_2\ n\ (\mathit{recNat} \cdot T\ t_1\ t_2\ n)|
      \end{array}
    \end{array}
  \]
  \caption{Typing and computation laws for primitive recursion on \(\mathit{Nat}\)}
  \label{fig:nat-rec}
\end{figure}
The significant feature of primitive recursion is that in the successor case,
the user-supplied function \(t_2\) has access both to the predecessor \(n\)
\emph{and} the result recursively computed from \(n\).

With primitive recursion, we can give the following implementation of the
predecessor function \(\mathit{pred}\).
\[\mathit{pred} = \mathit{recNat} \cdot \mathit{Nat}\ \mathit{zero}\
  (\absu{\lambda}{x}{\absu{\lambda}{y}{x}})\]
The efficiency of this definition of \(\mathit{pred}\) depends on the
operational semantics of the language.
Under call-by-name semantics, we have that \(|\mathit{pred}\
(\mathit{suc}\ n)|\) reduces to \(n\) in a constant number of steps since the
recursively computed result (the predecessor of \(n\)) is discarded before it
can be further evaluated.
This is not the case for call-by-value semantics: for closed \(n\) we
would compute \emph{all} predecessors of \(n\), then discard these results.

In System~F with isorecursive types, we can obtain Parigot naturals from the
typing and computation laws for primitive recursion over naturals.
The solutions for \(\mathit{Nat}\), \(\mathit{recNat}\), \(\mathit{zero}\), and
\(\mathit{suc}\) we acquire in this way are shown in Figure~\ref{fig:parigot-nat}.
\begin{figure}[h]
  \centering
  \[
    \begin{array}{lcl}
      \mathit{Nat}
      & =
      & \absu{\mu}{N}{\absu{\forall}{X}{X \to (N \to X \to X) \to X}}
      \\ \mathit{recNat}
      & =
      &
        \absu{\Lambda}{X}{\absu{\lambda}{z}{\absu{\lambda}{s}{\absu{\lambda}{x}{\mathit{unroll}(x)
        \cdot X\ z\ s}}}}
      \\ \mathit{zero}
      & =
      & \mathit{roll}(\absu{\lambda}{X}{\absu{\lambda}{z}{\absu{\lambda}{s}{z}}})
      \\ \mathit{suc}
      & =
      &
        \absu{\lambda}{n}{\mathit{roll}(\absu{\Lambda}{X}{\absu}{\lambda}{z}{\absu{\lambda}{s}{
        s\ n\ (\mathit{recNat} \cdot X\ z\ s\ n)}})}
    \end{array}
  \]
  \caption{Parigot naturals in System~F with isorecursive types}
  \label{fig:parigot-nat}
\end{figure}

\paragraphb{Parigot naturals and canonicity.}
In addition to showing yet another application of derived
recursive types in Cedille, this section serve two pedagogical purposes.
First, the Parigot encoding more readily supports primitive recursion
than the Scott encoding, for which the construction is rather complex (see
Section~\ref{sec:lr}).
Second, the derivation of induction for Parigot-encoded data involves a very
different approach than that used in Section~\ref{sec:scott}, taking full
advantage of the embedding of untyped lambda calculus in Cedille.

We elaborate on this second point further: as observed by
\cite{Ge14_Church-Scott-Encoding}, there is a deficiency in the definition of
the type of Parigot-encoded data in polymorphic type theory with recursive types.
For example, the type \(\mathit{Nat}\) is not precise enough: it admits the
definition of the following bogus constructor.
\[\mathit{suc'} = \absu{\lambda}{n}{\absu{\lambda}{z}{\absu{\lambda}{s}{s\
        \mathit{zero}\ (\mathit{recNat}\ z\ s\ n)}}}\]
The difficulty is that the type \(\mathit{Nat}\) does not enforce that the first
argument to the bound \(s\) is the same number that we use to compute the second
argument.
Put another way, this represents a failure to secure the extensionality law
(uniqueness) for the primitive recursion scheme with this encoding.

To address this, we observe that there is a purely computational
characterization of the subset of \(\mathit{Nat}\) that contains only the
canonical Parigot naturals.
This characterization is the \emph{reflection law}.
As we will see by proving induction, the set of canonical
Parigot naturals is precisely the set of terms \(n\) of type 
\(\mathit{Nat}\) satisfying the following equality:
\[\{\mathit{recNat}\ \mathit{zero}\ (\absu{\lambda}{m}{\mathit{suc}})\ n \simeq n\}\]
As an example, the non-canonical Parigot natural \(\mathit{suc'}\ (\mathit{suc}\
\mathit{zero})\) does \emph{not} satisfy this criterion: rebuilding it with the
constructors \(\mathit{zero}\) and \(\mathit{suc}\) produces \(\mathit{suc}\
(\mathit{suc}\ \mathit{zero})\) (see also \citet[][Section
2]{GJF12_Generic-Fibrational-Induction}, where it is shown that the reflection
law together with dependent products guarantees induction for naturals).

With \(\mathit{Top}\) and the Kleene trick (Section~\ref{sec:kleene-trick}), we can
express the property that a term satisfies the reflection law \emph{before} we
give a type for Parigot naturals.
This is good, because we wish to use the reflection law \emph{in the definition}
of the type of Parigot naturals!

\subsection{Parigot-encoded naturals, concretely}
\label{sec:parigot-nat}

We split the derivation of inductive Parigot naturals into three parts.
In Figure~\ref{fig:parigot-nat-1}, we define untyped operations for Parigot
naturals and prove that the untyped constructors preserve the reflection law.
In Figure~\ref{fig:parigot-nat-2}, we define the type \(\mathit{Nat}\) of
canonical Parigot naturals and its constructors.
Finally, in Figure~\ref{fig:parigot-nat-3} we define the subset of Parigot
naturals supporting induction, then show that the type \(\mathit{Nat}\) is
included in this subset.

\begin{figure}
  \centering
  \small
\begin{verbatim}
import cast .
import mono .
import recType .
import view .
import utils/top .

module parigot/concrete/nat .

recNatU ◂ Top
= β{ λ z. λ s. λ n. n z s } .

zeroU ◂ Top
= β{ λ z. λ s. z } .

sucU ◂ Top ➔ Top
= λ n. β{ λ z. λ s. s n (recNatU z s n) } .

reflectNatU ◂ Top
= β{ recNatU zeroU (λ m. sucU) } .

NatC ◂ Top ➔ ★ = λ n: Top. { reflectNatU n ≃ n } .

zeroC ◂ NatC zeroU = β{ zeroU } .

sucC ◂ Π n: Top. NatC n ➾ NatC (sucU n)
= λ n. Λ nc. ρ nc @x.{ sucU x ≃ sucU n } - β{ sucU n } .
\end{verbatim}
  \caption{Parigot naturals (part 1) (\texttt{parigot/concrete/nat.ced})}
  \label{fig:parigot-nat-1}
\end{figure}

\paragraphb{Reflection law.}
The first definitions in Figure~\ref{fig:parigot-nat-1} are untyped operations
for Parigot naturals.
Definition \(\mathit{recNatU}\) is the combinator for
primitive recursion, and \(\mathit{zeroU}\) and \(\mathit{sucU}\) are the
constructors (compare these to the corresponding definitions in
Figure~\ref{fig:parigot-nat}).
The term \(\mathit{reflectNatU}\) is the function which rebuilds Parigot naturals
with their constructors, and the predicate \(\mathit{NatC}\) expresses the
reflection law for untyped Parigot naturals.

The proofs \(\mathit{zeroC}\) and \(\mathit{sucC}\) show respectively that
\(\mathit{zeroU}\) satisfies the reflection law, and that if \(n\) satisfies the
reflection law then so does \(\mathit{sucU\ n}\).
In the proof for \(\mathit{sucU}\), the expected type reduces to an equality type whose left-hand
side is convertible with:
\[ \mathit{sucU}\ (\mathit{reflectNatU}\ n) \]
We finish the proof by rewriting with the assumption that \(n\) satisfies the reflection law.

Note that in addition to using the Kleene trick (Section~\ref{sec:kleene-trick})
to have a type \(\mathit{Top}\) for terms of the untyped lambda calculus, we are
also using it so that the proofs \(\mathit{zeroC}\) and \(\mathit{sucC}\) are 
definitionally equal to the untyped constructors \(\mathit{zeroU}\) and
\(\mathit{sucU}\) (see Figure~\ref{fig:eraser} for the erasure of \(\rho\)).
\begin{verbatim}
_ ◂ { zeroC ≃ zeroU } = β .
_ ◂ { sucC  ≃ sucU  } = β .
\end{verbatim}
(where \(\_\) indicates an anonymous proof).
This is so that we may define Parigot naturals as an equational subset type with
dependent intersection, which we will see next.

\begin{figure}
  \centering
  \small
\begin{verbatim}
NatF' ◂ ★ ➔ ★
= λ N: ★. ∀ X: ★. X ➔ (N ➔ X ➔ X) ➔ X .

NatF ◂ ★ ➔ ★
= λ N: ★. ι n: NatF' ·N. NatC β{ n } .

monoNatF ◂ Mono ·NatF = <..>

Nat ◂ ★ = Rec ·NatF .
rollNat   ◂ NatF ·Nat ➔ Nat = roll -monoNatF .
unrollNat ◂ Nat ➔ NatF ·Nat = unroll -monoNatF .

recNat ◂ ∀ X: ★. X ➔ (Nat ➔ X ➔ X) ➔ Nat ➔ X
= Λ X. λ z. λ s. λ n. (unrollNat n).1 z s .

zero' ◂ NatF' ·Nat = Λ X. λ z. λ s. z .

zero ◂ Nat = rollNat [ zero' , zeroC ] .

suc' ◂ Nat ➔ NatF' ·Nat
= λ n. Λ X. λ z. λ s. s n (recNat z s n) .

suc ◂ Nat ➔ Nat
= λ n. rollNat [ suc' n , sucC β{ n } -(unrollNat n).2 ] .
\end{verbatim}
  \caption{Parigot naturals (part 2) (\texttt{parigot/concrete/nat.ced})}
  \label{fig:parigot-nat-2}
\end{figure}

\paragraphb{\(\mathit{Nat}\), the type of Parigot naturals.}
In Figure~\ref{fig:parigot-nat-2}, we first define the type scheme
\(\mathit{NatF'}\) whose fixpoint over-approximates the type of Parigot
naturals.
Using dependent intersection, we then define the type scheme \(\mathit{NatF}\)
to map types \(N\) to the subset of terms of type \(\mathit{NatF'} \cdot N\)
which satisfy the reflection law.
This type scheme is monotonic (\(\mathit{monoNatF}\), definition omitted), so we
may use the recursive type former \(\mathit{Rec}\) to define the type
\(\mathit{Nat}\) with rolling and unrolling operators \(\mathit{rollNat}\) and
\(\mathit{unrollNat}\) that are definitionally equal to \(\absu{\lambda}{x}{x}\)
(see Figure~\ref{fig:recType-ax}).

\paragraphb{Constructors of \(\mathit{Nat}\).}
Definitions \(\mathit{recNat}\), \(\mathit{zero}\), and \(\mathit{suc}\) are the
typed versions of the primitive recursion combinator and constructors for
Parigot naturals.
The definitions of the constructors are split into two parts, with
\(\mathit{zero'}\) and \(\mathit{suc'}\) constructing terms of type
\(\mathit{NatF'} \cdot \mathit{Nat}\) and the unprimed constructors combining their
primed counterparts with the respective proofs that they satisfy the reflection
law.
For example, in \(\mathit{suc}\) the second component of the dependent
intersection is a proof of
\[\{\mathit{reflectNatU}\ (\mathit{sucU}\ \beta\{n\}) \simeq \mathit{sucU}\
  \beta\{n\}\}\]
obtained from invoking the proof \(\mathit{sucC}\) with
\[ (\textit{unrollNat}\ n).2 : \{\mathit{reflectNatU}\ \beta\{(\mathit{unrollNat}\ n).1\}
  \simeq \beta\{(\mathit{unrollNat}\ n).1\}\}\]
This is accepted by Cedille by virtue of the following definition equalities:
\[
  \begin{array}{lclcl}
    |\mathit{sucU}|
    & =_{\beta\eta}
    & |\mathit{suc'}|
    & =_{\beta\eta}
    & |\mathit{sucC}|
    \\ |\beta\{(\mathit{unrollNat}\ n).1\}|
    & =_{\beta\eta}
    & |n|
    & =_{\beta\eta}
    &  |\beta\{n\}|
  \end{array}
\]
Finally, as expected the typed and untyped versions of each of these three
operations are definitionally equal.

\begin{figure}
  \centering
  \small
\begin{verbatim}
IndNat ◂ Nat ➔ ★
= λ n: Nat. ∀ P: Nat ➔ ★. P zero ➔ (Π m: Nat. P m ➔ P (suc m)) ➔ P n .

NatI ◂ ★ = ι n: Nat. IndNat n .

recNatI
◂ ∀ P: Nat ➔ ★. P zero ➔ (Π m: Nat. P m ➔ P (suc m)) ➔ Π n: NatI. P n.1
= Λ P. λ z. λ s. λ n. n.2 z s .

indZero ◂ IndNat zero
= Λ P. λ z. λ s. z .

zeroI ◂ NatI = [ zero , indZero ] .

indSuc ◂ Π n: NatI. IndNat (suc n.1)
= λ n. Λ P. λ z. λ s. s n.1 (recNatI z s n) .

sucI ◂ NatI ➔ NatI
= λ n. [ suc n.1 , indSuc n ] .

reflectNatI ◂ Nat ➔ NatI
= recNat zeroI (λ _. sucI) .

toNatI ◂ Cast ·Nat ·NatI
= intrCast -reflectNatI -(λ n. (unrollNat n).2) .

indNat ◂ ∀ P: Nat ➔ ★. P zero ➔ (Π m: Nat. P m ➔ P (suc m)) ➔ Π n: Nat. P n
= Λ P. λ z. λ s. λ n. recNatI z s (elimCast -toNatI n) .
\end{verbatim}
  \caption{Parigot naturals (part 3) (\texttt{parigot/concrete/nat.ced})}
  \label{fig:parigot-nat-3}
\end{figure}

\paragraphb{\(\mathit{NatI}\), the type of inductive Parigot naturals.}
The definition of the inductive subset of Parigot naturals begins in
Figure~\ref{fig:parigot-nat-3} with the predicate \(\mathit{IndNat}\) over
\(\mathit{Nat}\).
For all \(n : \mathit{Nat}\), \(\mathit{IndNat}\ n\) is the type of proofs that,
in order to prove an arbitrary predicate \(P\) holds for \(n\), it suffices to
give corresponding proofs for the constructors \(\mathit{zero}\) and
\(\mathit{suc}\).
We again note that in the successor case, we have access to both the predecessor
\(m\) and a proof of \(P\ m\) as an inductive hypothesis.
The type \(\mathit{NatI}\) is then defined with dependent intersection as
the subset of \(\mathit{Nat}\) for which the predicate \(\mathit{IndNat}\) holds.

Definition \(\mathit{recNatI}\) brings us close to the derivation of an
induction principle for \(\mathit{Nat}\), but does not quite achieve it.
With an inductive proof of predicate \(P\), we have only that \(P\) holds for
the Parigot naturals in the inductive subset.
It remains to show that \emph{every} Parigot natural is in this subset.
We begin this with proofs \(\mathit{indZero}\) and
\(\mathit{indSuc}\) stating resp.\ that \(\mathit{zero}\) satisfies
\(\mathit{IndNat}\) and, for every \(n\) in the inductive subset
\(\mathit{NatI}\), \(\mathit{suc}\ n.1\) satisfies \(\mathit{IndNat}\).
We see that \(|\mathit{indZero}| =_{\beta\eta} |\mathit{zero}|\), and also that
\(|\mathit{indSuc}| =_{\beta\eta} |\mathit{suc}|\) since \(|\mathit{recNatI}|
=_{\beta\eta} |\mathit{recNat}|\).
Thus, the constructors \(\mathit{zeroI}\) and \(\mathit{sucI}\) for
\(\mathit{NatI}\) can be formed with dependent intersection introduction.

\paragraphb{Reflection and induction.}
We can now show that every term of type \(\mathit{Nat}\) also has type
\(\mathit{NatI}\), i.e., every (canonical) Parigot natural is in the inductive subset.
We do this by leveraging the fact that satisfaction of the reflection law is
baked into the type Parigot naturals.
First, we define \(\mathit{reflectNatI}\) which uses \(\mathit{recNat}\) to
recursively rebuild a Parigot natural with the constructors \(\mathit{zeroI}\)
and \(\mathit{sucI}\) of the inductive subset.
Next, we observe that \(|\mathit{reflectNatI}| =_{\beta\eta}
|\mathit{reflectNatU}|\), so we define a cast \(\mathit{toNatI}\) where the
given proof
\[(\mathit{unrollNat}\ n).2 : \{\mathit{reflectNatU}\ \beta\{(\mathit{unrollNat}\ n).1\}
  \simeq \beta\{(\mathit{unrollNat}\ n).1\}\}\]
has a type convertible with the expected type \(\{\mathit{reflectNatI}\ n \simeq
n\}\).
So here we are using the full power of \(\mathit{intrCast}\) by exhibiting a
type inclusion using a function that behaves extensionally like identity, but is
\emph{not} intensionally the identity function.
The proof \(\mathit{indNat}\) of the induction principle for Parigot
naturals follows from \(\mathit{recNatI}\) and the use of \(\mathit{toNatI}\) to
convert the given \(n : \mathit{Nat}\) to the type \(\mathit{NatI}\).

\subsubsection{Computational and extensional character.}
\label{sec:parigot-nat-char}
\begin{figure}[h]
  \centering
  \small
\begin{verbatim}
recNatBeta1
◂ ∀ X: ★. ∀ z: X. ∀ s: Nat ➔ X ➔ X.
  { recNat z s zero ≃ z }
= Λ X. Λ z. Λ s. β .

recNatBeta2
◂ ∀ X: ★. ∀ z: X. ∀ s: Nat ➔ X ➔ X. ∀ n: Nat.
  { recNat z s (suc n) ≃ s n (recNat z s n) }
= Λ X. Λ z. Λ s. Λ n. β .

indNatComp ◂ { indNat ≃ recNat } = β .

pred ◂ Nat ➔ Nat = recNat zero (λ n. λ r. n) .

predBeta1 ◂ { pred zero ≃ zero } = β .

predBeta2 ◂ ∀ n: Nat. { pred (suc n) ≃ n }
= Λ n. β .
\end{verbatim}
  \caption{Computation laws for primitive recursion and predecessor
    (\texttt{parigot/concrete/nat.ced})}
  \label{fig:parigot-nat-comp}
\end{figure}
We now give a characterization of \(\mathit{Nat}\).
From the code listing in Figure~\ref{fig:parigot-nat-2}, it is clear
\(\mathit{recNat}\) satisfies the typing law.
Figure~\ref{fig:parigot-nat-comp} shows the proofs of the computation laws
(\(\mathit{recNatBeta1}\) and \(\mathit{recNatBeta2}\)), which hold by
definitional equality.
By inspecting the definitions of \(\mathit{recNat}\), \(\mathit{zero}\), and
\(\mathit{suc}\), and from the erasures of \(\mathit{roll}\) and
\(\mathit{unroll}\), we can confirm that the computation laws are simulated in
a constant number of reduction steps under both call-by-name and call-by-value
semantics.

Figure~\ref{fig:parigot-nat-comp} also includes \(\mathit{indNatComp}\), which
shows that the computational content underlying the induction principle is
precisely the recursion scheme, and the predecessor function \(\mathit{pred}\)
with its expected computation laws.
As mentioned earlier, we must qualify that with an efficient simulation of
\(\mathit{recNat}\) we do not obtain an efficient solution for the predecessor
function under call-by-value operational semantics.

\begin{figure}[h]
  \centering
  \small
\begin{verbatim}
recNatEta
◂ ∀ X: ★. ∀ z: X. ∀ s: Nat ➔ X ➔ X.
  ∀ h: Nat ➔ X. { h zero ≃ z } ➾ (Π n: Nat. { h (suc n) ≃ s n (h n) }) ➾
  Π n: Nat. { h n ≃ recNat z s n }
= Λ X. Λ z. Λ s. Λ h. Λ hBeta1. Λ hBeta2.
  indNat ·(λ x: Nat. { h x ≃ recNat z s x })
    (ρ hBeta1 @x.{ x ≃ z } - β)
    (λ m. λ ih.
       ρ (hBeta2 m) @x.{ x ≃ s m (recNat z s m) }
     - ρ ih @x.{ s m x ≃ s m (recNat z s m) } - β) .
\end{verbatim}
  \caption{Extensional law for primitive recursion
    (\texttt{parigot/concrete/nat.ced})}
  \label{fig:parigot-nat-ext}
\end{figure}

Finally, in Figure~\ref{fig:parigot-nat-ext} we use induction to prove that for
all \(z : X\) and \(s : \mathit{Nat} \to X \to X\), \(\mathit{recNat}\ z\ s\) is
the unique solution satisfying the computation laws for primitive recursion with
respect to \(z\) and \(s\).
Unlike the analogous proof for \(\mathit{caseNat}\) in
Section~\ref{sec:scott-nat-comp}, in the successor case we reach a subgoal where
we must prove \(s\ m\ (h\ m)\) is equal to \(s\ m\ (\mathit{recNat}\ z\ s\ m)\)
for an arbitrary \(m : \mathit{Nat}\) and function \(h : \mathit{Nat} \to X\)
which satisfies the computation laws with respect to \(z\) and \(s\).
At that point, we use the inductive hypothesis, which is unavailable with weak
induction, to conclude the proof.

\subsection{\(\mathit{Functor}\) and \(\mathit{Sigma}\)}
\label{sec:functor-sigma}

The statement of the generic primitive recursion scheme requires functors and
pair types, and the induction principle additionally requires dependent pair
types.
As we will use this scheme to define the generic Parigot encoding, in this
section we first show our formulation of functors and the functor laws and give
an axiomatic presentation of the derivation of dependent pair types with
induction in Cedille.

\paragraphb{Functors.}
\begin{figure}[h]
  \centering
\small
\begin{verbatim}
module functor (F : ★ ➔ ★).

Fmap ◂ ★ = ∀ X: ★. ∀ Y: ★. (X ➔ Y) ➔ (F ·X ➔ F ·Y).

FmapId ◂ Fmap ➔ ★ = λ fmap: Fmap.
  ∀ X: ★. ∀ Y: ★. Π c: X ➔ Y. (Π x: X. {c x ≃ x}) ➔ Π x: F ·X . {fmap c x ≃ x}.

FmapCompose ◂ Fmap ➔ ★ = λ fmap: Fmap.
  ∀ X: ★. ∀ Y: ★. ∀ Z: ★. Π f: Y ➔ Z. Π g: X ➔ Y. Π x: F ·X.
  {fmap f (fmap g x) ≃ fmap (λ x. f (g x)) x}.
\end{verbatim}
  \caption{Functors (\texttt{functor.ced})}
  \label{fig:functor}
\end{figure}

Functors and the associated identity and composition laws for them are given in
Figure~\ref{fig:functor}.
Analogous to monotonicity, functoriality of a type scheme \(F : \star \to \star\)
means that \(F\) comes together with an operation \(\mathit{fmap} :
\mathit{Fmap} \cdot F\) lifting functions \(S \to T\) to functions \(F \cdot S
\to F \cdot T\), for all types \(S\) and \(T\).
Unlike monotonicity, in working with functions we find ourselves in a
proof-relevant setting, so we will require that this lifting respects identity
(\(\mathit{FmapId}\ \mathit{fmap}\)) and composition (\(\mathit{FmapCompose}\
\mathit{fmap}\)).

\begin{figure}
  \centering
  \small
\begin{verbatim}
import functor.

module functorThms (F: ★ ➔ ★) (fmap: Fmap ·F)
  {fmapId: FmapId ·F fmap} {fmapCompose: FmapCompose ·F fmap}.

import cast .
import mono .

monoFunctor ◂ Mono ·F
= Λ X. Λ Y. λ c.
  intrCast
    -(λ d. fmap (elimCast -c) d)
    -(λ d. fmapId (elimCast -c) (λ x. β) d).
\end{verbatim}
  \caption{Functors and monotonicity (\texttt{functorThms})}
  \label{fig:functorThms}
\end{figure}

Notice also that our definition of the identity law has an extrinsic twist: the
domain and codomain of the lifted function \(c\) need not be convertible types
for us to satisfy the constraint that \(c\) acts extensionally like the identity
function.
Phrasing the identity law in this way allows us to derive a useful lemma,
\(\mathit{monoFunctor}\) in Figure~\ref{fig:functorThms}, that establishes that
every functor is a monotone type scheme (see Figure~\ref{fig:mono} for the
definition of \(\mathit{Mono}\)).

\paragraphb{Dependent pair types.}
\begin{figure}
  \centering
  \[
    \begin{array}{c}
      \infer{
      \Gamma \vdash \mathit{Sigma} \cdot S \cdot T \tpsynth \star
      }{
      \Gamma \vdash S \tpsynth \star
      \quad \Gamma \vdash T \tpsynth S \to \star
      }
      \\ \\
      \infer{
      \Gamma \vdash \mathit{mksigma} \cdot S \cdot T\ s\ t \tpsynth
      \mathit{Sigma} \cdot S \cdot T
      }{
      \Gamma \vdash \mathit{Sigma} \cdot S \cdot T \tpsynth \star
      \quad \Gamma \vdash s \tpcheck S
      \quad \Gamma \vdash t \tpcheck T\ s
      }
      \\ \\
      \begin{array}{cc}
          \infer{
          \Gamma \vdash \mathit{proj1}\ p \tpsynth S
          }{
           \Gamma \vdash p \tpsynth \mathit{Sigma} \cdot S \cdot T
          }
        & 
        \infer{
         \Gamma \vdash \mathit{proj2}\ p \tpsynth T\ (\mathit{proj1}\ p)
        }{
         \Gamma \vdash p \tpsynth \mathit{Sigma} \cdot S \cdot T
        }
      \end{array}
      \\ \\
      \infer{
       \Gamma \vdash \mathit{indsigma}\ p \cdot P\ f \tpsynth P\ p
      }{
      \Gamma \vdash p \tpsynth \mathit{Sigma} \cdot S \cdot T
      \quad \Gamma \vdash P \tpsynth \mathit{Sigma} \cdot S \cdot T \to \star
      \quad \Gamma \vdash f \tpcheck \abs{\Pi}{x}{S}{\abs{\Pi}{y}{T\ x}{P\ (\mathit{mksigma}\ x\ y)}}
      }
      \\ \\
      \begin{array}{lcl}
        |\mathit{proj1}\ (\mathit{mksigma}\ s\ t)|
        & =_{\beta\eta}
        & |s|
        \\ |\mathit{proj2}\ (\mathit{mksigma}\ s\ t)|
        & =_{\beta\eta}
        & |t|
        \\ |\mathit{indsigma}\ (\mathit{mksigma}\ s\ t)\ f|
        & =_{\beta\eta}
        & |f\ s\ t|
      \end{array}
    \end{array}
  \]
  \caption{\(\mathit{Sigma}\), axiomatically (\texttt{utils/sigma.ced})}
    \label{fig:sigma}
\end{figure}

\begin{figure}
  \centering
  \small
\begin{verbatim}
Pair ◂ ★ ➔ ★ ➔ ★
= λ S: ★. λ T: ★. Sigma ·S ·(λ _: S. T).

fork ◂ ∀ X: ★. ∀ S: ★. ∀ T: ★. (X ➔ S) ➔ (X ➔ T) ➔ X ➔ Pair ·S ·T
= Λ X. Λ S. Λ T. λ f. λ g. λ x. mksigma (f x) (g x) .
\end{verbatim}
  \caption{\(\mathit{Pair}\) (\texttt{utils/sigma.ced})}
  \label{fig:pair}
\end{figure}

Figure~\ref{fig:sigma} gives an axiomatic presentation of the dependent pair
type \(\mathit{Sigma}\) (see the code repository for the full derivation).
The constructor is \(\mathit{mksigma}\), the first and second projections are
\(\mathit{proj1}\) and \(\mathit{proj2}\), and the induction principle is
\(\mathit{indsigma}\). 
Below the type inference rules, we confirm that the projection functions and
induction principle compute as expected over pairs formed from the constructor.
In Figure~\ref{fig:pair}, the type \(\mathit{Pair}\) is defined in terms of
\(\mathit{Sigma}\) as the special case where the type of the second component
does not depend upon the first component.
Additionally, the figure also gives the utility function \(\mathit{fork}\) for
constructing non-dependent pairs to help express the computation law of the
primitive recursion scheme.

\subsection{Parigot-encoded data, generically}
\label{ssec:parigot-gen}

% \subsubsection{Recursive algebras}
% \label{sec:parigot-gen-rec}

In this section we derive inductive Parigot-encoded datatypes generically.
The derivation is parametric in a signature functor \(F\) for the datatype with
an operation \(\mathit{fmap} : \mathit{Fmap} \cdot F\) that satisfies the
functor identity and composition laws.

\begin{figure}[h]
  \centering
  \[
    \begin{array}{c}
      \infer{
       \Gamma \vdash \mathit{recD} \cdot T\ t \tpsynth D \to T
      }{
      \Gamma \vdash T \tpsynth \star
      \quad \Gamma \vdash t \tpcheck F \cdot (\mathit{Pair} \cdot D \cdot T) \to T
      }
      \\ \\
      \begin{array}{lcl}
        |\mathit{recD}\ t\ (\mathit{inD}\ t')|
        & =_{\beta\eta}
        & |t\ (\mathit{fmap}\ (\mathit{fork}\ \mathit{id}\ (\mathit{recD}\ t))\ t')|
      \end{array}
    \end{array}
  \]  
  \caption{Generic primitive recursion scheme}
  \label{fig:scheme-rec}
\end{figure}

The typing and computation laws for the generic primitive recursion scheme for
datatype \(D\) with signature \(F\) are given in Figure~\ref{fig:scheme-rec}.
For the typing rule, we see that the primitive recursion scheme allows recursive
functions to be defined in terms an \(F\)-collection of tuples containing both
direct predecessors \emph{and} the recursive results computed from those
predecessors.
This reading is affirmed by the computation law, which states that the
action of \(\mathit{recD}\ t\) over values built from \(\mathit{inD}\ t'\) (for
some \(t' : F \cdot D\)) is to apply \(t\) to the result of tupling each
predecessor (accessed with \(\mathit{fmap}\)) with the result of
\(\mathit{recD}\ t\) (here \(\mathit{id}\) is the polymorphic identity
function).
Assuming \(t\) and \(\mathit{inD}\ t'\) are typed according to the typing
law, the right-hand side of the equation has type \(T\) with the following
assignment of types to sub-expressions.
\begin{itemize}
\item \(\mathit{id} : D \to D\)
  
\item \(\mathit{fork}\ \mathit{id}\ (\mathit{recD}\ t) : D \to \mathit{Pair}
  \cdot D \cdot T\)
  
\item \(\mathit{fmap}\ (\mathit{fork}\ \mathit{id}\ (\mathit{recD}\ t)) : F \cdot
  D \to F \cdot (\mathit{Pair} \cdot D \cdot T)\)
\end{itemize}

With primitive recursion, we can implement the datatype destructor \(\mathit{outD}\):
\[\mathit{outD} = \mathit{recD}\ (\mathit{fmap}\ \mathit{proj1})\]
This simulates the desired computation law \(|\mathit{outD}\
(\mathit{inD}\ t)| =_{\beta\eta} |t|\) only up to the functor identity and
composition laws.
With definitional equality, we obtain a right-hand side of:
\[|\mathit{fmap}\ \mathit{proj1}\ (\mathit{fmap}\ (\mathit{fork}\ \mathit{id}\
  \mathit{outD})\ t)|\]
Additionally, and as we saw for Parigot naturals, this is not an efficient
implementation of the destructor under call-by-value operational semantics since
the predecessors of \(t\) are recursively destructed.

\begin{figure}[h]
  \centering
  \[
    \begin{array}{lcl}
      D
      & =
      & \absu{\mu}{D}{\abs{\forall}{X}{\star}{(F \cdot (\mathit{Pair} \cdot D \cdot X) \to X) \to X}}
      \\ \mathit{recD}
      & =
      & \absu{\Lambda}{X}{\absu{\lambda}{a}{\absu{\lambda}{x}{\mathit{unroll}(x) \cdot X\ a}}}
      \\ \mathit{inD}
      & =
      & \absu{\lambda}{x}{\mathit{roll}(\absu{\Lambda}{X}{\absu{\lambda}{a}{a\
        (\mathit{fmap} \ (\mathit{fork}\ (\mathit{id} \cdot D)\ (\mathit{recD}\ a))\ x)}})}
    \end{array}
  \]
  \caption{Generic Parigot encoding of \(D\) in System~\(\text{F}^{\omega}\) with
    isorecursive types}
  \label{fig:data-rec}
\end{figure}

Using the typing and computation laws for primitive recursion to read an
encoding for \(D\), we obtain a generic supertype of canonical Parigot encodings
in Figure~\ref{fig:data-rec}.
Similar to the case of the Scott encoding, we find that to give the definition
of \(D\) we need some form of recursive types since \(D\) occurs in the premises
of the typing law.
For our derivation, we must further refine the type of \(D\) so that we only
include canonical Parigot encodings (i.e., those built only from \(\mathit{inD}\)).
We take the same approach used for Parigot naturals: \(\mathit{Top}\) and
the Kleene trick help us express satisfaction of the reflection law for untyped
terms, which is in turn used to give a refined definition of \(D\). 

\subsubsection{Characterization criteria.}
We formalize in Cedille the above description of the generic primitive recursion
scheme in Figures~\ref{fig:data-char-primrec-typing} and
\ref{fig:data-char-primrec}.
Definitions for the typing law of the primitive recursion scheme are given in
Figure~\ref{fig:data-char-primrec-typing}, where parameter \(F\) gives the
datatype signature.
Type family \(\mathit{AlgRec}\) gives the shape of the type functions used for
primitive recursion, and \(\mathit{PrimRec}\) gives the shape of the type of
operator \(\mathit{recD}\) itself.

\begin{figure}[h]
  \centering
  \small
\begin{verbatim}
import utils .

module primrec-typing (F: ★ ➔ ★) .

AlgRec ◂ ★ ➔ ★ ➔ ★
= λ D: ★. λ X: ★. F ·(Pair ·D ·X) ➔ X .

PrimRec ◂ ★ ➔ ★
= λ D: ★. ∀ X: ★. AlgRec ·D ·X ➔ D ➔ X .
\end{verbatim}
  \caption{Primitive recursion typing (\texttt{data-char/primrec-typing.ced})}
  \label{fig:data-char-primrec-typing}
\end{figure}

\begin{figure}
  \centering
  \small
\begin{verbatim}
import functor .
import utils .

import data-char/iter-typing .
import data-char/case-typing .

module data-char/primrec
  (F: ★ ➔ ★) (fmap: Fmap ·F) (D: ★) (inD: Alg ·F ·D).

import data-char/primrec-typing ·F .

AlgRecHom ◂ Π X: ★. AlgRec ·D ·X ➔ (D ➔ X) ➔ ★
= λ X: ★. λ a: AlgRec ·D ·X. λ h: D ➔ X.
  ∀ xs: F ·D. { h (inD xs) ≃ a (fmap (fork id h) xs) } .

PrimRecBeta ◂ PrimRec ·D ➔ ★
= λ rec: PrimRec ·D.
  ∀ X: ★. ∀ a: AlgRec ·D ·X. AlgRecHom ·X a (rec a) .

PrimRecEta ◂ PrimRec ·D ➔ ★
= λ rec: PrimRec ·D.
  ∀ X: ★. ∀ a: AlgRec ·D ·X. ∀ h: D ➔ X. AlgRecHom ·X a h ➔
  Π x: D. { h x ≃ rec a x } .

PrfAlgRec ◂ (D ➔ ★) ➔ ★
= λ P: D ➔ ★. Π xs: F ·(Sigma ·D ·P). P (inD (fmap (proj1 ·D ·P) xs)) .

fromAlgCase ◂ ∀ X: ★. AlgCase ·F ·D ·X ➔ AlgRec ·D ·X
= Λ X. λ a. λ xs. a (fmap ·(Pair ·D ·X) ·D (λ x. proj1 x) xs) .

fromAlg ◂ ∀ X: ★. Alg ·F ·X ➔ AlgRec ·D ·X
= Λ X. λ a. λ xs. a (fmap ·(Pair ·D ·X) ·X (λ x. proj2 x) xs) .
\end{verbatim}
  \caption{Primitive recursion characterization
    (\texttt{data-char/primrec.ced})}
  \label{fig:data-char-primrec}
\end{figure}

In Figure~\ref{fig:data-char-primrec}, we now assume that \(F\) has an operation
\(\mathit{fmap}\) for lifting functions, and we take additional module
parameters \(D\) for the datatype and \(\mathit{inD}\) for its constructor.
We import the definitions of the modules defined in Figures~\ref{fig:iter-laws}
and \ref{fig:data-char-case} without specifying the type scheme module
parameter, so it is given explicitly for definitions exported from these modules (e.g.,
\(\mathit{Alg} \cdot F \cdot D\)).
\(\mathit{AlgRecHom}\) gives the shape of the computation law for primitive
recursion with respect to a particular \(a : \mathit{AlgRec} \cdot D \cdot X\),
\(\mathit{PrimRecBeta}\) is a predicate on candidates for the combinator for
primitive recursion stating that they satisfy the computation law with respect
to \emph{all} such functions \(a\), and \(\mathit{PrimRecEta}\) is a predicate
stating that a candidate is the unique such solution up to function
extensionality.

The figure also lists \(\mathit{PrfAlgRec}\), a dependent version of
\(\mathit{AlgRec}\).
Read the type \(\mathit{PrfAlgRec} \cdot P\) as the type of ``\((F,D)\)-proof
algebras for \(P\)''.
It is the type of proofs that take an \(F\)-collection of
\(D\) predecessors tupled with proofs that \(P\) holds for them and produces
a proof that \(P\) holds for the value constructed from these predecessors with
\(\mathit{inD}\).
\(\mathit{PrfAlgRec}\) will be used in the derivations of full induction for
both the generic Parigot and generic Scott encoding.

Finally, as the primitive recursion scheme can be used to simulate both the
iteration and case-distinction schemes, the figure lists the helper functions
\(\mathit{fromAlgCase}\) and \(\mathit{fromAlg}\).
Definition \(\mathit{fromAlgCase}\) converts a function for use in case
distinction to one for use in primitive recursion by ignoring previously
computed results, and \(\mathit{fromAlg}\) converts a function for use in
iteration by ignoring predecessors.

\subsubsection{Generic Parigot encoding.}
\label{sec:parigot-gen}
We now detail the generic derivation of inductive Parigot-encoded data.
The developments of this section are parametrized by a functor \(F\), with
\(\mathit{fmap}\) giving the lifting of functions and \(\mathit{fmapId}\) and
\(\mathit{fmapCompose}\) the proofs that this lifting respects identity and
composition.
The construction is separated into several phases: in
Figure~\ref{fig:parigot-generic-1} we give the computational characterization of
canonical Parigot encodings as a predicate on untyped terms, then prove that the
untyped constructor preserves this predicate; in
Figure~\ref{fig:parigot-generic-2}, we define the type of Parigot encodings, its 
primitive recursion combinator, and its constructors;
in Figure~\ref{fig:parigot-generic-3} we define the inductive subset of Parigot
encodings and its constructor; finally, in Figure~\ref{fig:parigot-generic-4} we
show that every Parigot encoding is already in the inductive subset and prove induction.

\begin{figure}
  \centering
  \small
\begin{verbatim}
import functor .
import utils .

import cast .
import mono .
import recType .

module parigot/generic/encoding
  (F: ★ ➔ ★) (fmap: Fmap ·F)
  {fmapId: FmapId ·F fmap} {fmapCompose: FmapCompose ·F fmap } .

import functorThms ·F fmap -fmapId -fmapCompose .

recU ◂ Top = β{ λ a. λ x. x a } .

inU ◂ Top = β{ λ xs. λ a. a (fmap (fork id (recU a)) xs) } .

reflectU ◂ Top = β{ recU (λ xs. inU (fmap proj2 xs)) } .

DC ◂ Top ➔ ★ = λ x: Top. { reflectU x ≃ x } .

inC ◂ Π xs: F ·(ι x: Top. DC x). DC β{ inU xs }
= λ xs.
  ρ (fmapCompose ·(ι x: Top. DC x) ·(Pair ·Top ·Top) ·Top
       (λ x. proj2 x) (fork (λ x. x.1) (λ x. β{ reflectU x })) xs)
    @x.{ inU x ≃ inU xs }
- ρ (fmapId ·(ι x: Top. DC x) ·Top (λ x. β{ reflectU x.1 }) (λ x. x.2) xs)
    @x.{ inU x ≃ inU xs }
- β{ inU xs } .
\end{verbatim}
  \caption{Generic Parigot encoding (part 1)
    (\texttt{parigot/generic/encoding.ced})}
  \label{fig:parigot-generic-1}
\end{figure}

\paragraphb{Reflection law.}
The definitions \(\mathit{recU}\), \(\mathit{inU}\), and \(\mathit{reflectU}\)
of Figure~\ref{fig:parigot-generic-1} are untyped versions of resp.\ the
combinator for primitive recursion, the generic constructor, and the operation
that builds canonical Parigot encodings by recursively rebuilding the encoding
with the constructor \(\mathit{inU}\) (compare to Figure~\ref{fig:parigot-nat-1}
of Section~\ref{sec:parigot-nat}).
Predicate \(\mathit{DC}\) gives the characterization of canonical Parigot 
encodings that \(\mathit{reflectU}\) behaves extensionally like the
identity function for them.

Even without yet having a type for Parigot-encoded data, we can still
effectively reason about the behaviors of these untyped programs.
This is shown in the proof of \(\mathit{inC}\), which states \(\mathit{inU}\
\mathit{xs}\) satisfies the predicate \(\mathit{DC}\) if \(\mathit{xs}\) is an
\(F\)-collection of untyped terms that satisfy \(\mathit{DC}\).
In the body, the expected type is convertible with the type
\[\{{ \mathit{inU}\ (\mathit{fmap}\ \mathit{proj2}\
    (\mathit{fmap}\ (\mathit{fork}\ \mathit{id}\ \mathit{reflectU})\
    \mathit{xs}))
    ≃ \mathit{inU}\ \mathit{xs}}\}\] 
We rewrite by the functor composition law to fuse the mapping of
\(\mathit{proj2}\) with that of \(\mathit{fork}\ \mathit{id}\
\mathit{reflectU}\), and now the left-hand side of the resulting equation is
convertible (by the computation law for \(\mathit{proj2}\)) with
\[ \mathit{inU}\ (\mathit{fmap}\ \mathit{reflectU}\ \mathit{xs})\]
Here, we can rewrite by the functor identity law using the assumption that, on the
predecessors contained in \(\mathit{xs}\), \(\mathit{reflectU}\) behaves as the
identity function.
Note that we use the Kleene trick so that the proof
\(\mathit{inC}\) is definitionally equal to the untyped constructor
\(\mathit{inU}\).
This allows us to use dependent intersection to form the refinement needed to
type only canonical Parigot encodings.

\begin{figure}[h]
  \centering
  \small
\begin{verbatim}
import data-char/primrec-typing ·F .

DF' ◂ ★ ➔ ★ = λ D: ★. ∀ X: ★. AlgRec ·D ·X ➔ X .
DF  ◂ ★ ➔ ★ = λ D: ★. ι x: DF' ·D. DC β{ x } .

monoDF ◂ Mono ·DF = <..>

D ◂ ★ = Rec ·DF .
rollD   ◂ DF ·D ➔ D = roll -monoDF .
unrollD ◂ D ➔ DF ·D = unroll -monoDF .

recD ◂ PrimRec ·D
= Λ X. λ a. λ x. (unrollD x).1 a .

inD' ◂ F ·D ➔ DF' ·D
= λ xs. Λ X. λ a. a (fmap (fork (id ·D) (recD a)) xs) .

toDC ◂ Cast ·D ·(ι x: Top. DC x)
= intrCast -(λ x. [ β{ x } , (unrollD x).2 ]) -(λ x. β) .

inD ◂ F ·D ➔ D
= λ xs. rollD [ inD' xs , inC (elimCast -(monoFunctor toDC) xs) ] .
\end{verbatim}
  \caption{Generic Parigot encoding (part 2)
    (\texttt{parigot/generic/encoding.ced})}
  \label{fig:parigot-generic-2}
\end{figure}

\paragraphb{Parigot encoding \(\mathit{D}\).}
In Figure~\ref{fig:parigot-generic-2}, \(\mathit{DF'}\) is the type
scheme whose least fixpoint is the solution to \(D\) in Figure~\ref{fig:data-rec}, and
\(\mathit{DF}\) is the refinement of \(\mathit{DF'}\) to those terms satisfying
the reflection law.
Since \(\mathit{DF}\) is a monotone type scheme (\(\mathit{monoDF}\), proof
omitted), we may define \(D\) as its fixpoint with rolling and unrolling
operations \(\mathit{rollD}\) and \(\mathit{unrollD}\).
Following this is \(\mathit{recD}\), the typed combinator for primitive
recursion (see Figure~\ref{fig:data-char-primrec-typing} for the definition of
\(\mathit{PrimRec}\)).

The constructor \(\mathit{inD}\) for \(D\) is defined in two parts.
First, we define \(\mathit{inD'}\) to construct a value of type \(\mathit{DF'}
\cdot D\) from \(\mathit{xs} : F \cdot D\), with the definition similar to that
which we obtained in Figure~\ref{fig:data-rec}.
Then, with the auxiliary proof \(\mathit{toDC}\) that \(D\) is included into the
type of untyped terms satisfying \(\mathit{DC}\), we define the constructor
\(\mathit{inD}\) for \(D\) using the rolling operation and dependent intersection
introduction.
The definition is accepted by virtue of the following definitional equalities:
\[
  \begin{array}{lcl}
    |\mathit{inD'}|
    & =_{\beta\eta}
    & |\mathit{inU}|
    \quad =_{\beta\eta}
    \quad |\mathit{inC}|
    \\ |\mathit{xs}|
    & =_{\beta\eta}
    & |\mathit{elimCast}\ \mhyph (\mathit{monoFunctor}\ \mathit{toDC})\ \mathit{xs}|
  \end{array}
\]
We can use Cedille to confirm that the typed recursion combinator and
constructor are definitionally equal to the corresponding untyped operations.
\begin{verbatim}
_ ◂ { recD ≃ recU } = β .
_ ◂ { inD  ≃ inU }  = β .
\end{verbatim}

\begin{figure}
  \centering
  \small
\begin{verbatim}
import data-char/primrec ·F fmap -fmapId -fmapCompose ·D inD .

IndD ◂ D ➔ ★ = λ x: D. ∀ P: D ➔ ★. PrfAlgRec ·P ➔ P x .

DI ◂ ★ = ι x: D. IndD x .

recDI ◂ ∀ P: D ➔ ★. PrfAlgRec ·P ➔ Π x: DI. P x.1
= Λ P. λ a. λ x. x.2 a .

fromDI ◂ Cast ·DI ·D
= intrCast -(λ x. x.1) -(λ x. β) .

inDI' ◂ F ·DI ➔ D
= λ xs. inD (elimCast -(monoFunctor fromDI) xs) .

indInDI' ◂ Π xs: F ·DI. IndD (inDI' xs)
= λ xs. Λ P. λ a.
  ρ ς (fmapId ·DI ·D (λ x. proj1 (mksigma x.1 (recDI a x))) (λ x. β) xs)
    @x.(P (inD x))
- ρ ς (fmapCompose (proj1 ·D ·P) (λ x: DI. mksigma x.1 (recDI a x)) xs)
    @x.(P (inD x))
- a (fmap ·DI ·(Sigma ·D ·P) (λ x. mksigma x.1 (recDI a x)) xs) .

inDI ◂ F ·DI ➔ DI
= λ xs. [ inDI' xs , indInDI' xs ] .
\end{verbatim}
  \caption{Generic Parigot encoding (part 3)
    (\texttt{parigot/generic/encoding.ced})}
  \label{fig:parigot-generic-3}
\end{figure}

\paragraphb{Inductive Parigot encoding \(\mathit{DI}\).}
In Figure~\ref{fig:parigot-generic-3} we give the definition of \(\mathit{DI}\),
the type for the inductive subset of Parigot-encoded data, and its
constructor \(\mathit{inDI}\).
This definition begins by bringing \(\mathit{PrfAlgRec}\) into scope with an
import, used to define the predicate \(\mathit{IndD}\).
For arbitrary \(x : D\), the property \(\mathit{IndD}\ x\) states that, for all
\(P : D \to \star\), to prove \(P\ x\) it suffices to give an
\((F,D)\)-proof-algebra for \(P\).
Then, we define the inductive subset \(\mathit{DI}\) of Parigot encodings that
satisfy the predicate \(\mathit{IndD}\) using dependent intersections.
Definition \(\mathit{recDI}\) is the induction principle for terms of type \(D\)
in this subset, and corresponds to the definition \(\mathit{recNatI}\) for
Parigot naturals in Figure~\ref{fig:parigot-nat-3}.
With \(\mathit{recDI}\), we can obtain the desired induction scheme for \(D\) if
we show \(\mathit{D}\) is included into \(\mathit{DI}\).

We begin the proof of this type inclusion by defining the constructor
\(\mathit{inDI}\) for the inductive subset.
This is broken into three parts.
First, \(\mathit{inDI'}\) constructs a value of type \(D\) from an
\(F\)-collection of \(\mathit{DI}\) predecessors using the inclusion of type
\(\mathit{DI}\) into \(D\) (\(\mathit{fromDI}\)).
Next, with \(\mathit{indInDI'}\) we prove that the values constructed from
\(\mathit{inDI'}\) satisfy the inductivity predicate \(\mathit{IndD}\) using the
functor identity and composition laws.

In the body of \(\mathit{indInDI'}\), we use equational reasoning to bridge the
gap between the expected type \(P\ (\mathit{inDI'}\ \mathit{xs})\) and the type
of the expression in the final line, which is:
\[ P\ (\mathit{inD}\ (\mathit{fmap}\ (\mathit{proj1} \cdot D \cdot P)\
  (\mathit{fmap} \cdot \mathit{DI} \cdot (\mathit{Sigma} \cdot D \cdot P) (\absu{\lambda}{x}{\mathit{mksigma}\ x.1\
    (\mathit{recDI}\ a\ x)}))\ \mathit{xs}))\]
The main idea here is that we first use the functor composition law to fuse the
two lifted operations, then observe that this results in lifting a single
function,
\[ \absu{\lambda}{x}{\mathit{proj1}\ (\mathit{mksigma}\ x.1\ (\mathit{recDI}\ a\
    x)) : \mathit{DI} \to D}
\]
that is definitionally equal to the identity function (by the computation law
for \(\mathit{proj1}\), Figure~\ref{fig:sigma}, and the erasure of dependent
intersection projections).
We then use the functor identity law, exchanging \(\mathit{inD}\) for
\(\mathit{inDI'}\) in the rewritten type (these two terms are definitionally equal).
As \(\mathit{inDI'}\) and \(\mathit{indInDI'}\) are definitionally equal to each
other (since \(|\mathit{fork}\ \mathit{id}\ (\mathit{recD}\ a)| =_{\beta\eta}
|\absu{\lambda}{x}{\mathit{mksigma}\ x.1\ (\mathit{recDI}\ a\ x)}|\)),
we can define the constructor \(\mathit{inDI}\) using the rolling
operation and dependent intersection introduction. 

\begin{figure}
  \centering
  \small
\begin{verbatim}
reflectDI ◂ D ➔ DI
= recD (λ xs. inDI (fmap ·(Pair ·D ·DI) ·DI (λ x. proj2 x) xs)) .

toDI ◂ Cast ·D ·DI
= intrCast -reflectDI -(λ x. (unrollD x).2) .

indD ◂ ∀ P: D ➔ ★. PrfAlgRec ·D inD ·P ➔ Π x: D. P x
= Λ P. λ a. λ x. recDI a (elimCast -toDI x) .
\end{verbatim}
  \caption{Generic Parigot encoding (part 4)
    (\texttt{parigot/generic/encoding/ced})}
  \label{fig:parigot-generic-4}
\end{figure}

\paragraphb{Reflection and induction.}
We can now show an inclusion of the type \(D\) into the type \(\mathit{DI}\),
giving us induction, by using the fact that terms of type \(D\) are canonical
Parigot encodings.
This is shown in Figure~\ref{fig:parigot-generic-4}.
First, we define the operation \(\mathit{reflectDI}\) which recursively rebuilds
Parigot-encoded data with the constructor \(\mathit{inDI}\) for the inductive
subset, producing a value of type \(\mathit{DI}\).
Then, since \(|\mathit{reflectDI}| =_{\beta\eta} |\mathit{reflectU}|\), we can
use \(\mathit{reflectDI}\) to witness the inclusion of \(D\) into
\(\mathit{DI}\), since every term \(x\) of type \(D\) is itself a proof that
\(\mathit{reflectU}\) behaves extensionally as the identity function on \(x\).
From here, the proof \(\mathit{indD}\) of induction uses
\(\mathit{recDI}\) in combination with this type inclusion.

\subsubsection{Computational and extensional character.}
\label{sec:parigot-props}
We now analyze the properties of our generic Parigot encoding.
In particular, we wish to know the normalization guarantee for terms of type
\(D\) and to confirm that \(\mathit{recD}\) is an efficient and unique solution
to the primitive recursion scheme for \(D\), which can in turn be used to
simulate case distinction and iteration.
We omit the uniqueness proofs, which make heavy use of the functor laws and
rewriting (see the code repository for this paper).

\begin{figure}
\small
\begin{verbatim}
import functor .
import cast .
import recType .
import utils .

module parigot/generic/props
  (F: ★ ➔ ★) (fmap: Fmap ·F)
  {fmapId: FmapId ·F fmap} {fmapCompose: FmapCompose ·F fmap} .

import functorThms ·F fmap -fmapId -fmapCompose .
import parigot/generic/encoding ·F fmap -fmapId -fmapCompose .
import data-char/primrec-typing ·F .

normD ◂ Cast ·D ·(AlgRec ·D ·D ➔ D)
= intrCast -(λ x. (unrollD x).1 ·D) -(λ x. β) .

import data-char/primrec ·F fmap -fmapId -fmapCompose ·D inD .

recDBeta ◂ PrimRecBeta recD
= Λ X. Λ a. Λ xs. β .

reflectD ◂ Π x: D. { recD (fromAlg inD) x ≃ x }
= λ x. (unrollD x).2 .

recDEta ◂ PrimRecEta recD = <..>
\end{verbatim}
  \caption{Characterization of \(\mathit{recD}\) (\texttt{parigot/generic/props.ced})}
  \label{fig:parigot-props-rec}
\end{figure}

\paragraphb{Normalization guarantee.}
In Figure~\ref{fig:parigot-props-rec}, \(\mathit{normD}\) establishes the
inclusion of type \(D\) into the function type \(\mathit{AlgRec} \cdot D \cdot D
\to D\).
By Proposition~\ref{thm:cedille-termination}, this guarantees call-by-name
normalization of closed terms of type \(D\).

\paragraphb{Primitive recursion scheme}
With proof \(\mathit{recDBeta}\), we have that our solution \(\mathit{recD}\)
(Figure~\ref{fig:parigot-generic-2}) satisfies the computation law by
definitional equality.
By inspecting the definitions of \(\mathit{recD}\) and \(\mathit{inD}\), and the
erasures of \(\mathit{roll}\) and \(\mathit{unroll}\)
(Figure~\ref{fig:recType-ax}), we can confirm that the computation law is
satisfied in a constant number of steps under both call-by-name and
call-by-value operational semantics.

Definition \(\mathit{recDEta}\) establishes that this solution is unique up function
extensionality.
The proof follows from induction, the functor laws, and a non-obvious use of the
Kleene trick.
The reflection law is usually obtained as a consequence of uniqueness, but as
the generic Parigot encoding has been defined with satisfaction of this law
baked in, the proof \(\mathit{reflectD}\) proceeds by appealing to that fact
directly.

\begin{figure}[h]
  \centering
  \small
\begin{verbatim}
import data-char/case-typing ·F .
import data-char/case ·F ·D inD .

caseD ◂ Case ·D
= Λ X. λ a. recD (fromAlgCase a) .

caseDBeta ◂ CaseBeta caseD
= Λ X. Λ a. Λ xs.
  ρ (fmapCompose ·D ·(Pair ·D ·X) ·D
       (λ x. proj1 x) (fork (id ·D) (caseD a)) xs)
    @x.{ a x ≃ a xs }
- ρ (fmapId ·D ·D (λ x. proj1 (fork (id ·D) (caseD a) x)) (λ x. β) xs)
    @x.{ a x ≃ a xs }
- β.

caseDEta ◂ CaseEta caseD = <..>
\end{verbatim}
  \caption{Characterization of \(\mathit{caseD}\)
    (\texttt{parigot/generic/props.ced})}
  \label{fig:parigot-props-case}
\end{figure}

\paragraphb{Case-distinction scheme.}
In Figure~\ref{fig:parigot-props-case}, we define the candidate
\(\mathit{caseD}\) for the operation giving case distinction for \(D\) using
\(\mathit{recD}\) and \(\mathit{fromAlgCase}\) (Figure~\ref{fig:data-char-primrec}).
This definition satisfies the computation law only up to the functor laws.
With definitional equality alone, \(\mathit{caseD}\ a\ (\mathit{inD}\
\mathit{xs})\) is joinable with
\[a\ (\mathit{fmap}\ \mathit{proj1}\ (\mathit{fmap}\ (\mathit{fork}\
  \mathit{id}\ (\mathit{caseD}\ a))\ \mathit{xs}))\] 
meaning we have introduced another traversal over the signature with
\(\mathit{fmap}\).
As we have seen before, under call-by-value semantics this would also cause
\(\mathit{caseD}\ a\) to be needlessly computed for all predecessors.
The proof of extensionality, \(\mathit{caseDEta}\), follows from
\(\mathit{recDEta}\).

\begin{figure}[h]
  \centering
  \small
\begin{verbatim}
import data-char/destruct ·F ·D inD .

outD ◂ Destructor = caseD (λ xs. xs) .

lambek1D ◂ Lambek1 outD
= λ xs. ρ (caseDBeta ·(F ·D) -(λ x. x) -xs) @x.{ x ≃ xs } - β .

lambek2D ◂ Lambek2 outD = <..>
\end{verbatim}
  \caption{Characterization of destructor \(\mathit{outD}\)
    (\texttt{parigot/generic/props.ced})}
  \label{fig:parigot-props-out}
\end{figure}

\paragraphb{Destructor.}
In Figure~\ref{fig:parigot-props-out} we define the destructor
\(\mathit{outD}\) using case distinction.
As such, the destructor inherits the caveat that the computation law,
\(\mathit{lambek1D}\), only holds up to the functor laws and is not efficient
under call-by-name semantics.
The extensionality law, \(\mathit{lambek2D}\), holds by induction.

\begin{figure}[h]
  \centering
  \small
\begin{verbatim}
import data-char/iter-typing ·F .
import data-char/iter ·F fmap ·D inD .

foldD ◂ Iter ·D
= Λ X. λ a. recD (fromAlg a) .

foldDBeta ◂ IterBeta foldD
= Λ X. Λ a. Λ xs.
  ρ (fmapCompose ·D ·(Pair ·D ·X) ·X
       (λ x. proj2 x) (fork (id ·D) (foldD a)) xs)
    @x.{ a x ≃ a (fmap (foldD a) xs) }
- β .

foldDEta ◂ IterEta foldD = <..>
\end{verbatim}
  \caption{Characterization of \(\mathit{foldD}\) (\texttt{parigot/generic/props.ced})}
  \label{fig:parigot-initial}
\end{figure}

\paragraphb{Iteration.}
The last property we confirm is that we may use \(\mathit{recD}\) to give a
unique solution for the typing and computation laws of the iteration scheme
(Figures~\ref{fig:iter-typing} and \ref{fig:iter-laws}).
The proposed solution is \(\mathit{foldD}\), given in
Figure~\ref{fig:parigot-initial}.
The computation law, proven with \(\mathit{foldDBeta}\), only holds by the
functor laws, as there are two traversals of the signature with \(\mathit{fmap}\)
instead of one (\(\mathit{fromAlg}\), defined in
Figure~\ref{fig:data-char-primrec}, introduces the additional traversal).

\subsubsection{Example: Rose trees.}
\label{sec:parigot-rosetrees}
We conclude the discussion of Parigot-encoded data by using the generic
derivation to define rose trees with induction.
Rose trees, also called finitely branching trees, are a datatype in which subtrees are
contained within a list, meaning that nodes may have an arbitrary number of
children.
In Haskell, they are defined as:
\begin{verbatim}
data RoseTree a = Rose a [RoseTree a]
\end{verbatim}

There are two motivations for this choice of example.
First, while the rose tree datatype can be put into a form that reveals it is a strictly
positive datatype by using containers \citep{AAG03_Categories-of-Containers} and
a nested inductive definition, we use impredicative encodings for datatypes
(including lists), and so the rose tree datatype we define is not syntactically
strictly positive.
Second, the expected induction principle for rose trees is moderately tricky.
Indeed, it is complex enough to be difficult to synthesize automatically:
additional plugins
\citep{Ull20_Generating-Induction-Principles-for-Nested-Inductive-Types-MetaCoq}
are required for Coq, and in the Agda standard library \citep{ADT21_Agda-StdLib}
the induction principle can be obtained automatically by declaring rose trees as
a size-indexed type \citep{Ab10_Integrating-Sized-and-Dependent-Types}. 

The difficulty lies in giving users access to the inductive hypothesis for the
sub-trees contained within the list.
To work around this, users of Coq or Agda can define a mutually inductive
definition, forgoing reuse for the special-purpose ``list of rose trees''
datatype, or prove the desired induction principle manually for the definition
that uses the standard list type.
In this section, we use the induction principle derived for our
generic Parigot encoding to take this second approach, proving an
induction principle for rose trees in the style expected of a mutual inductive
definition while re-using the list datatype.
The presentation of a higher-level language, based on such a generic encoding, in
which this induction principle could be automatically derived is a matter for
future work.

\begin{figure}[h]
  \centering
  \small
\begin{verbatim}
import utils.

module parigot/examples/list-data (A : ★).

import signatures/list ·A .
import parigot/generic/encoding as R
  ·ListF listFmap -listFmapId -listFmapCompose .

List ◂ ★ = R.D .

nil ◂ List = <..>
cons ◂ A ➔ List ➔ List = <..>

indList
◂ ∀ P: List ➔ ★. P nil ➔ (Π hd: A. Π tl: List. P tl ➔ P (cons hd tl)) ➔
  Π xs: List. P xs
= <..>

recList ◂ ∀ X: ★. X ➔ (A ➔ List ➔ X ➔ X) ➔ List ➔ X
= Λ X. indList ·(λ x: List. X) .
\end{verbatim}
  \caption{Lists (\texttt{parigot/examples/list-data.ced})}
  \label{fig:parigot-list-data}
\end{figure}

\begin{figure}[h]
  \centering
  \small
\begin{verbatim}
import utils .
import functor .

module parigot/examples/list .

import parigot/examples/list-data .

listMap ◂ Fmap ·List
= Λ A. Λ B. λ f.
  recList ·A ·(List ·B) (nil ·B) (λ hd. λ tl. λ xs. cons (f hd) xs) .

listMapId ◂ FmapId ·List listMap = <..>
listMapCompose ◂ FmapCompose ·List listMap = <..>
\end{verbatim}
  \caption{Map for lists (\texttt{parigot/examples/list.ced})}
  \label{fig:parigot-list-map}
\end{figure}

\paragraphb{Lists.}
Figure~\ref{fig:parigot-list-data} shows the definition of the datatype
\(\mathit{List}\), and the types of the list constructors (\(\mathit{nil}\) and
\(\mathit{cons}\)) and induction principle (\(\mathit{indList}\)).
These are defined using the generic derivation of inductive Parigot encodings.
This derivation is brought into scope as ``\(R\)'', and the definitions
within that module are accessed with the prefix ``\(\mathit{R.}\)'', e.g.,
``\(\mathit{R.D}\)'' for the datatype.
Note also that code in the figure refers to \(\mathit{List}\) as a datatype, not
\(\mathit{List} \cdot A\), since \(A\) is a parameter to the module.

For the sake of brevity, we omit the definitions of the list signature
(\(\mathit{ListF}\)), its mapping operation (\(\mathit{listFmap}\)), and the
proofs this mapping satisfies the functor laws (\(\mathit{listFmapId}\) and
\(\mathit{listFmapCompose}\)) (see \texttt{signatures/list.ced} in the code
repository).
In the figure, these are given as module arguments to the generic derivation.
We also define the primitive recursion principle \(\mathit{recList}\) for lists
as a non-dependent use of induction.

In Figure~\ref{fig:parigot-list-map}, we change module contexts (so we may
consider lists with different element types, e.g. \(\mathit{List} \cdot B\)) to
define the list mapping operation \(\mathit{listMap}\) by recursion.
We also prove with \(\mathit{listMapId}\) and \(\mathit{listMapCompose}\) that
this mapping operation obeys the functor laws.

\begin{figure}[h]
  \centering
  \small
\begin{verbatim}
import functor .
import utils .

module signatures/tree
  (A: ★) (F: ★ ➔ ★) (fmap: Fmap ·F)
  {fmapId: FmapId ·F fmap} {fmapCompose: FmapCompose ·F fmap} .

TreeF ◂ ★ ➔ ★ = λ T: ★. Pair ·A ·(F ·T) .

treeFmap ◂ Fmap ·TreeF
= Λ X. Λ Y. λ f. λ t. mksigma (proj1 t) (fmap f (proj2 t)) .

treeFmapId ◂ FmapId ·TreeF treeFmap = <..>
treeFmapCompose ◂ FmapCompose ·TreeF treeFmap = <..>
\end{verbatim}
  \caption{Signature for \(F\)-branching trees (\texttt{signatures/tree.ced})}
  \label{fig:parigot-tree-sig}
\end{figure}

\paragraphb{Signature \(\mathit{TreeF}\).}
Figure~\ref{fig:parigot-tree-sig} gives the signature \(\mathit{TreeF}\) for a
datatype of trees whose branching factor is given by a functor \(F\),
generalizing the rose tree datatype.
The proofs that the lifting operation \(\mathit{treeFmap}\) respects identity
and composition make use of the corresponding proofs for the given \(F\).

\begin{figure}[h]
  \centering
  \small
\begin{verbatim}
import utils .
import list-data .
import list .

module parigot/examples/rosetree-data (A: ★) .

import signatures/tree ·A ·List listMap -listMapId -listMapCompose .

import parigot/generic/encoding as R
  ·TreeF treeFmap -treeFmapId -treeFmapCompose .

RoseTree ◂ ★ = R.D .

rose ◂ A ➔ List ·RoseTree ➔ RoseTree
= λ x. λ t. R.inD (mksigma x t) .

rose' ◂ ∀ P: RoseTree ➔ ★. TreeF ·(Sigma ·RoseTree ·P) ➔ RoseTree
= Λ P. λ xs. R.inD (treeFmap (proj1 ·RoseTree ·P) xs) .

indRoseTree
◂ ∀ P: RoseTree ➔ ★. ∀ Q: List ·RoseTree ➔ ★.
  Q (nil ·RoseTree) ➔
  (Π t: RoseTree. P t ➔ Π ts: List ·RoseTree. Q ts ➔ Q (cons t ts)) ➔
  (Π x: A. Π ts: List ·RoseTree. Q ts ➔ P (rose x ts)) ➔
  Π t: RoseTree. P t
= Λ P. Λ Q. λ n. λ c. λ r.
  R.indD ·P (λ xs.
     indsigma xs ·(λ x: TreeF ·(Sigma ·RoseTree ·P). P (rose' x))
       (λ x. λ ts.
          [conv ◂ List ·(Sigma ·RoseTree ·P) ➔ List ·RoseTree
           = listMap (proj1 ·RoseTree ·P)]
        - [pf ◂ Q (conv ts)
           = indList ·(Sigma ·RoseTree ·P)
               ·(λ x: List ·(Sigma ·RoseTree ·P). Q (conv x))
               n (λ hd. λ tl. λ ih. c (proj1 hd) (proj2 hd) (conv tl) ih) ts]
        - r x (conv ts) pf)) .
\end{verbatim}
  \caption{Rose trees (\texttt{parigot/examples/rosetree-data.ced})}
  \label{fig:parigot-rosetree}
\end{figure}

\paragraphb{Rose trees.}
In Figure~\ref{fig:parigot-rosetree}, we instantiate the module parameters for
the signature of \(F\)-branching trees with \(\mathit{List}\) and
\(\mathit{listMap}\), then instantiate the generic derivation of inductive
Parigot encodings with \(\mathit{TreeF}\) and \(\mathit{treeFmap}\).
We define the standard constructor \(\mathit{rose}\) for rose trees using the
generic constructor \(\mathit{R.inD}\), and give a variant constructor
\(\mathit{rose'}\) which we use in the definition of the induction principle for
rose trees.
This variant uses \(\mathit{treeFmap}\) to remove the tupled proofs that an
inductive hypothesis holds for sub-trees, introduced in the generic induction principle.

Finally, we give the induction principle for rose trees as
\(\mathit{indRoseTree}\) in the figure.
This is a mutual induction principle, with \(P\) the property one desires to
show holds for all rose trees and \(Q\) the invariant maintained for collections
of sub-trees.
We require that:
\begin{itemize}
\item \(Q\) holds for \(\mathit{nil}\), bound as \(n\);
  
\item if \(P\) holds for \(\mathit{t}\) and \(Q\) holds for \(\mathit{ts}\) then \(Q\)
  holds for \(\mathit{cons}\ t\ \mathit{ts}\), bound as \(c\); and that

\item \(P\) holds for \(\mathit{rose}\ x\ \mathit{ts}\) when
  \(Q\) holds for \(\mathit{ts}\), bound as \(r\).

\end{itemize}
In the body of \(\mathit{indRoseTree}\), we use the induction principle
\(\mathit{R.indD}\) for the generic Parigot encoding, then
the induction principle \(\mathit{indsigma}\) (Figure~\ref{fig:sigma}) for
pairs, revealing \(\ann{x}{\mathit{RoseTree}}\) and \(\ann{\mathit{ts}}{\mathit{List}
  \cdot (\mathit{Sigma} \cdot \mathit{RoseTree} \cdot P)}\).
With auxiliary function \(\mathit{conv}\) to convert \(\mathit{ts}\) to a list
of rose trees, we use list induction on \(\mathit{ts}\) to prove \(Q\) holds for
\(\mathit{conv}\ \mathit{ts}\).
With this proved as \(\mathit{pf}\), we can conclude by using \(r\).

\section{Lepigre-Raffalli encoding}
\label{sec:lr}

We now revisit the issue of programming with Scott-encoded data.
Neither the case-distinction scheme, nor the weak induction principle we derived
in Section~\ref{sec:scott}, provide an obvious mechanism for recursion.
In contrast, the Parigot encoding readily admits the primitive recursion scheme,
as it can be viewed as a solution to that scheme.
So despite its significant overhead in space representation, the
Parigot encoding appears to have a clear advantage over the Scott encoding in
total typed lambda calculi.

Amazingly, in some settings this deficit of the Scott encoding is \emph{only}
apparent.
Working with a logical framework, \cite{parigot88} showed how to derive with
``metareasoning'' a strongly normalizing recursor for Scott naturals.
More recently, \cite{lepigre+19} demonstrated a well-typed recursor for Scott
naturals in a Curry-style theory featuring a sophisticated form of subtyping
which utilizes ``circular but well-founded'' derivations.
The Lepigre-Raffalli construction involves a novel impredicative encoding of
datatypes, which we shall call the \emph{Lepigre-Raffalli encoding}, that both
supports recursion \emph{and} is a supertype of the Scott encoding.
In Cedille, we can similarly show an inclusion of the type of Scott encodings into
the type of Lepigre-Raffalli encodings by using weak induction together with the
fact that our derived recursive types are least fixpoints.

\paragraphb{Lepigre-Raffalli recursion.}
We elucidate the construction of the Lepigre-Raffalli encoding by showing its
relationship to the case-distinction scheme.
The computation laws for case distinction over natural numbers
(Figure~\ref{fig:nat-case}) do not form a recursive system of equations.
However, having obtained solutions for Scott naturals and
\(\mathit{caseNat}\), we can \emph{introduce} recursion into the computation laws
by observing that for all \(n\), \(|n| =_{\beta\eta}
|\absu{\lambda}{z}{\absu{\lambda}{s}{\mathit{caseNat}\ z\ s\ n}}|\).
\[
  \begin{array}{lcl}
    \\ |\mathit{caseNat}\ t_1\ t_2\ \mathit{zero}|
    & =_{\beta\eta}
    & |t_1|
    \\ |\mathit{caseNat}\ t_1\ t_2\ (\mathit{suc}\ n)|
    & =_{\beta\eta}
    & |t_2\ (\absu{\lambda}{z}{\absu{\lambda}{s}{\mathit{caseNat}\ z\ s\ n}})|
  \end{array}
\]

Viewing the computation laws this way, we see in the \(\mathit{suc}\) case that
\(t_2\) is given a function which will make a recursive call on \(n\) when
provided a suitable base and step case.
We desire that these be the same base and step cases originally provided,
i.e., that these in fact be \(t_1\) and \(t_2\) again.
To better emphasize this new interpretation, we rename \(\mathit{caseNat}\)
to \(\mathit{recLRNat}\).
By congruence of \(\beta\eta\)-equivalence, the two equations above give us:
\[
  \begin{array}{lcl}
    |\mathit{recLRNat}\ t_1\ t_2\ \mathit{zero}\ t_1\ t_2|
    & =_{\beta\eta}
    & |t_1\ t_1\ t_2|
    \\ |\mathit{recLRNat}\ t_1\ t_2\ (\mathit{suc}\ n)\ t_1\ t_2|
    & =_{\beta\eta}
    & |t_2\ (\absu{\lambda}{z}{\absu{\lambda}{s}{\mathit{recLRNat}\ z\ s\ n}})\
      t_1\ t_2|
  \end{array}
\]

To give types to the terms involved in these equations, observe that we can use
impredicative quantification to address the self-application occurring in the
right-hand sides.
Below, let the type \(T\) of the result we are computing be such that type
variables \(Z\) and \(S\) are fresh with respect to its free variables, and let
``\(?\)'' be a placeholder for a type.
\[
  \begin{array}{lcl}
    t_1
    & : 
    & \abs{\forall}{Z}{\star}{\abs{\forall}{S}{\star}{Z \to S \to T}}
    \\ t_2
    & :
    & \abs{\forall}{Z}{\star}{\abs{\forall}{S}{\star}{\ ? \to Z \to S \to T}}
  \end{array}
\]

This gives an interpretation of \(t_1\) as a constant function
that ignores its two arguments and returns a result of type \(T\).
For \(t_2\), ``\(?\)'' holds the place of the type of its first argument,
\(\absu{\lambda}{z}{\absu{\lambda}{s}{\mathit{recLRNat}\ z\ s\ n}}\).
We are searching for a type that matches our intended reading that \(t_2\) will
instantiate the arguments \(z\) and \(s\) with \(t_1\) and \(t_2\).
Now, \(t_2\) is provided copies of \(t_1\) and \(t_2\) at the universally quantified
types \(Z\) and \(S\), so we make a further refinement:
\[
  \abs{\forall}{Z}{\star}{\abs{\forall}{S}{\star}{(Z \to S \to\ ?) \to Z \to S
      \to T}}
\]
We can complete the type of \(t_2\) by observing that in the system of recursive
equations for the computation law, \(\mathit{recLRNat}\) is a function of five arguments.
Using \(\eta\)-expansion, we can rewrite the equation for the successor case to
match this usage:
\[
  |\mathit{recLRNat}\ t_1\ t_2\ (\mathit{suc}\ n)\ t_1\ t_2| =_{\beta\eta} |t_2\
  (\absu{\lambda}{z}{\absu{\lambda}{s}{\absu{\lambda}{z'}{\absu{\lambda}{s'}{\mathit{recLRNat}\
        z\ s\ n\ z'\ s'}}}})\ t_1\ t_2|
\]
where we understand that, from the perspective of \(t_2\), instantiations of
\(z\) and \(z'\) should have the universally quantified type \(Z\) and that
instantiations of \(s\) and \(s'\) should have universally quantified type \(S\).
We thus obtain the complete definition of the type of \(t_2\).
\[
  \abs{\forall}{Z}{\star}{\abs{\forall}{S}{\star}{(Z \to S \to Z \to S \to T)
      \to Z \to S \to T}}
\]
\begin{figure}
  \centering
  \[
    \begin{array}{c}
      \begin{array}{lcl}
        \mathit{NatZ} \cdot T
        & =
        & \abs{\forall}{Z}{\star}{\abs{\forall}{S}{\star}{Z \to S \to T}}
        \\
        \mathit{NatS} \cdot T
        & =
        & \abs{\forall}{Z}{\star}{\abs{\forall}{S}{\star}{(Z \to S \to Z \to S \to T) \to Z \to
          S \to T}}
      \end{array}
      \\ \\
      \infer{
       \Gamma \vdash \mathit{recLRNat} \cdot T\ t_1\ t_2 \tpsynth \mathit{Nat} \to
      \mathit{NatZ} \cdot T \to \mathit{NatS} \cdot T \to T
      }{
      \Gamma \vdash T \tpsynth \star
      \quad \Gamma \vdash t_1 \tpcheck \mathit{NatZ} \cdot T
      \quad \Gamma \vdash t_2 \tpcheck \mathit{NatS} \cdot T
      }
    \end{array}
  \]
  \caption{Typing law for the Lepigre-Raffalli recursion scheme
    on \(\mathit{Nat}\)}
  \label{fig:scheme-rec-lr-nat}
\end{figure}

Now we are able to construct a typing rule for our recursive combinator, shown
in Figure~\ref{fig:scheme-rec-lr-nat}.
From this, we obtain the type for Lepigre-Raffalli naturals:
\[
  \mathit{Nat} = \abs{\forall}{X}{\star}{\mathit{NatZ} \cdot X \to \mathit{NatS}
  \cdot X \to \mathit{NatZ} \cdot X \to \mathit{NatS} \cdot X \to T}
\]

The remainder of this section is structured as follows.
In Section~\ref{sec:lr-nat}, we show that the type of Lepigre-Raffalli naturals
is a supertype of the type of Scott naturals, and derive the primitive recursion
scheme for Scott naturals from the Lepigre-Raffalli recursion scheme we have just discussed.
In Section~\ref{sec:lr-inductive-nat}, we modify the Lepigre-Raffalli encoding
and derive induction for Scott naturals.
Finally, in Section~\ref{sec:lr-gen} we generalize this modification and derive
induction for generic Scott encodings.

\subsection{Primitive recursion for Scott naturals, concretely}
\label{sec:lr-nat}

Our derivation of primitive recursion for Scott naturals is split into three
parts.
In Figure~\ref{fig:lr-nat-1}, we define the type of Lepigre-Raffalli naturals
and the combinator for Lepigre-Raffalli recursion.
In Figure~\ref{fig:lr-nat-2}, we prove that Scott naturals are a subtype of
Lepigre-Raffalli naturals, giving us Lepigre-Raffalli recursion over them.
Finally, in Figure~\ref{fig:lr-nat-3} we implement primitive recursion for Scott
naturals using Lepigre-Raffalli recursion.

\begin{figure}
  \centering
  \small
\begin{verbatim}
import cast.
import mono.
import recType.

import scott/concrete/nat as S .

module lepigre-raffalli/concrete/nat1 .

NatRec ◂ ★ ➔ ★ ➔ ★ ➔ ★
= λ X: ★. λ Z: ★. λ S: ★. Z ➔ S ➔ Z ➔ S ➔ X .

NatZ ◂ ★ ➔ ★
= λ X: ★. ∀ Z: ★. ∀ S: ★. Z ➔ S ➔ X .

NatS ◂ ★ ➔ ★
= λ X: ★. ∀ Z: ★. ∀ S: ★. NatRec ·X ·Z ·S ➔ Z ➔ S ➔ X .

Nat ◂ ★ = ∀ X: ★. NatRec ·X ·(NatZ ·X) ·(NatS ·X) .

recLRNat ◂ ∀ X: ★. NatZ ·X ➔ NatS ·X ➔ Nat ➔ NatZ ·X ➔ NatS ·X ➔ X
= Λ X. λ z. λ s. λ n. n z s .
\end{verbatim}
  \caption{Primitive recursion for Scott naturals (part 1) (\texttt{lepigre-raffalli/concrete/nat1.ced})}
  \label{fig:lr-nat-1}
\end{figure}

\paragraphb{Lepigre-Raffalli naturals.}
In Figure~\ref{fig:lr-nat-1}, we give in Cedille the definition for the type of
Lepigre-Raffalli encodings we previously obtained.
We use a qualified import of the concrete encoding of Scott naturals
(Section~\ref{sec:scott-nat}), so to access a definition from that development
we use ``\(\mathit{S.}\)'' as a prefix (not to be confused with the quantified
type variable \(S\) that appears with no period).
The common shape \(Z \to S \to Z \to S \to X\) has been refactored into the type
family \(\mathit{NatRec}\), used in the definitions of \(\mathit{NatS}\) and
\(\mathit{Nat}\).
We also give the definition for the recursive combinator \(\mathit{recLRNat}\),
which we observe is definitionally equal to \(\mathit{S.caseNat}\)
(Figure~\ref{fig:scott-nat-comp}):

\begin{figure}
  \centering
  \small
\begin{verbatim}
zero ◂ Nat
= Λ X. λ z. λ s. z ·(NatZ ·X) ·(NatS ·X) .

suc ◂ Nat ➔ Nat
= λ n. Λ X. λ z. λ s.
  s ·(NatZ ·X) ·(NatS ·X) (λ z'. λ s'. recLRNat z' s' n) .

rollNat ◂ Cast ·(S.NatFI ·Nat) ·Nat
= intrCast
    -(λ n. n.1 zero suc)
    -(λ n. n.2 ·(λ x: S.NatF ·Nat. { x zero suc ≃ x }) β (λ m. β)) .

toNat ◂ Cast ·S.Nat ·Nat = recLB -rollNat .
\end{verbatim}
  \caption{Primitive recursion for Scott naturals (part 2)
    (\texttt{lepigre-raffalli/concrete/nat1.ced})}
  \label{fig:lr-nat-2}
\end{figure}

\paragraphb{Inclusion of Scott naturals into Lepigre-Raffalli naturals.}
Figure~\ref{fig:lr-nat-2} shows that Scott naturals (\(\mathit{S.Nat}\)) are a
subtype of Lepigre-Raffalli naturals.
This begins with the constructors \(\mathit{zero}\) and \(\mathit{suc}\), whose
definitions come from computation laws we derived for \(\mathit{recLRNat}\).
In particular, for successor the first argument to the bound \(s\) is 
\(\absu{\lambda}{z'}{\absu{\lambda}{s'}{\mathit{recLRNat}\ z'\ s'\ n}}\), the
handle for making recursive calls on the predecessor \(n\) that awaits a
suitable base and step case.
Because the computation laws for \(\mathit{recLRNat}\) are derived from case
distinction, we have that \(\mathit{zero}\) and \(\mathit{suc}\) are
definitionally equal to \(\mathit{S.zero}\) and \(\mathit{S.suc}\).
\begin{verbatim}
_ ◂ { zero ≃ S.zero } = β .
_ ◂ { suc  ≃ S.suc  } = β .
\end{verbatim}

Recall that in Section~\ref{sec:scott-nat-comp}, we saw that the function which
rebuilds Scott naturals with its constructors behaves extensionally as the
identity function.
We can leverage this fact to define \(\mathit{rollNat}\), which establishes an
inclusion of \(\mathit{S.NatFI} \cdot Nat\) into \(\mathit{Nat}\) by rebuilding 
a term of the first type with the constructors \(\mathit{zero}\) and
\(\mathit{suc}\) for Lepigre-Raffalli naturals.
The proof is given not by \(\mathit{wkIndNat}\), but the even weaker
pseudo-induction principle \(\mathit{S.WkIndNatF} \cdot \mathit{Nat}\ n.1\),
\[\abs{\forall}{P}{\mathit{S.NatF} \cdot \mathit{Nat} \to \star}{P\
    (\mathit{S.zeroF} \cdot \mathit{Nat}) \to (\abs{\Pi}{m}{\mathit{Nat}}{P\
      (\mathit{S.sucF}\ m)) \to P\ n.1}}\]
given by \(n.2\).
We saw in Section~\ref{sec:scott-nat} that \(|\mathit{S.zeroF}| =
_{\beta\eta} |\mathit{S.zero}|\) and \(|\mathit{S.sucF}| =_{\beta\eta}
|\mathit{S.suc}|\), so it follows that \(|\mathit{S.zeroF}| =
_{\beta\eta} |\mathit{zero}|\) and \(|\mathit{S.sucF}| =_{\beta\eta} |\mathit{suc}|\).

With \(\mathit{rollNat}\), we have that \(\mathit{Nat}\) is an
\(\mathit{S.NatFI}\)-closed type.
Since \(\mathit{S.Nat} = \mathit{Rec} \cdot \mathit{S.NatFI}\) is a lower bound
of all such types with respect to type inclusion, using \(\mathit{recLB}\)
(Figure~\ref{fig:recType}) we have a cast from Scott naturals to
Lepigre-Raffalli naturals.

\begin{figure}[h]
  \centering
  \small
\begin{verbatim}
recNatZ ◂ ∀ X: ★. X ➔ NatZ ·(S.Nat ➔ X)
= Λ X. λ x. Λ Z. Λ S. λ z. λ s. λ m. x .

recNatS ◂ ∀ X: ★. (S.Nat ➔ X ➔ X) ➔ NatS ·(S.Nat ➔ X)
= Λ X. λ f. Λ Z. Λ S. λ r. λ z. λ s. λ m.
  f m (r z s z s (S.pred m)) .

recNat ◂ ∀ X: ★. X ➔ (S.Nat ➔ X ➔ X) ➔ S.Nat ➔ X
= Λ X. λ x. λ f. λ n.
  recLRNat ·(S.Nat ➔ X) (recNatZ x) (recNatS f)
    (elimCast -toNat n)
    (recNatZ x) (recNatS f) (S.pred n) .

recNatBeta1
◂ ∀ X: ★. ∀ x: X. ∀ f: S.Nat ➔ X ➔ X. { recNat x f S.zero ≃ x }
= Λ X. Λ x. Λ f. β .

recNatBeta2
◂ ∀ X: ★. ∀ x: X. ∀ f: S.Nat ➔ X ➔ X. ∀ n: S.Nat.
  { recNat x f (S.suc n) ≃ f n (recNat x f n) }
= Λ X. Λ x. Λ f. Λ n. β .
\end{verbatim}
  \caption{Primitive recursion for Scott naturals (part 3)
    (\texttt{lepigre-raffalli/concrete/nat1.ced})}
  \label{fig:lr-nat-3}
\end{figure}

\paragraphb{Primitive recursion for Scott naturals.}
The last step in equipping Scott naturals with primitive recursion is
to translate this scheme to the Lepigre-Raffalli recursion scheme.
This is done in three parts, shown in Figure~\ref{fig:lr-nat-3}.
One complication that must be addressed is that Lepigre-Raffalli recursion
reinterprets the predecessor as a function for making recursive calls, but 
primitive recursion enables direct access to the predecessor.
So that it may serve both roles, we duplicate the predecessor.
This means that if \(T\) is the type of results we wish to compute with
primitive recursion, then we use Lepigre-Raffalli recursion to compute a
function of type \(\mathit{S.Nat} \to T\).

If \(t : T\) is the base case for primitive recursion, then
\(\mathit{recNatZ}\ t\) is a constant polymorphic function that ignores its first
three arguments and returns \(t\).
For the step case \(f : \mathit{S.Nat} \to T \to T\), \(\mathit{recNatS}\ f\)
produces a step case for Lepigre-Raffalli recursion, introducing:
\begin{itemize}
\item type variables \(Z\) and \(S\),
  
\item \(\ann{r}{\mathit{NatRec} \cdot (\mathit{S.Nat} \to T) \cdot Z \cdot S}\),
  the handle for making recursive calls,
  
\item \(z\) and \(s\), the base and step cases at the abstracted types \(Z\) and
  \(S\), and
  
\item \(\ann{m}{S.Nat}\), which we intend to be a duplicate of \(r\).
  
\end{itemize}
In the body of \(\mathit{recNatS}\), \(f\) is given access to the predecessor
\(m\) and the result recursively computed with \(r\), where we decrement \(m\)
as we pass through the recursive call.
Finally, \(\mathit{recNat}\) gives us the primitive recursion scheme for Scott
naturals by translating the base and step cases to the Lepigre-Raffalli style
(and duplicating them), coercing the given Scott natural \(n\) to a
Lepigre-Raffalli natural, and giving also the predecessor of \(n\).

With \(\mathit{recNatBeta1}\) and \(\mathit{recNatBeta2}\), we use Cedille to
confirm that the expected computation laws for the primitive recursion scheme
hold by definition.
To give a more complete understanding of how \(\mathit{recNat}\) computes, we
show some intermediate steps involved for the step case for arbitrary
untyped terms \(t_1\), \(t_2\), and \(n\) in Figure~\ref{fig:lr-nat-rec-reduce}.
In the figure, we omit types and erased arguments, and indeed it should be read
as ordinary (full) \(\beta\)-reduction for untyped terms.
In the last two lines of the figure, we switch the direction of reduction.
Altogether, this shows that \(|\mathit{recNat}\ t_1\ t_2\ (\mathit{S.suc}\ n)|\)
and \(|\mathit{t_2}\ n\ (\mathit{recNat}\ t_1\ t_2\ n)|\) are joinable in a
constant number of reduction steps.

\begin{figure}
  \centering
  \[
    \begin{array}{cl}
      &
        \mathit{recNat}\ t_1\ t_2\ (\mathit{S.suc}\ n)
      \\ \\ \rightsquigarrow_\beta^*
      & \mathit{recLRNat}\ (\mathit{recNatZ}\ t_1)\ (\mathit{recNatS}\ t_2)\
        (\mathit{elimCast}\ (\mathit{S.suc}\ n))\ 
      \\ & \quad \quad (\mathit{recNatZ}\ t_1)\ (\mathit{recNatS}\
           t_2)\ (\mathit{S.pred}\ (\mathit{S.suc}\ n))
      \\ \\ \rightsquigarrow_\beta^*
      & \mathit{S.suc}\ n\ (\mathit{recNatZ}\ t_1)\ (\mathit{recNatS}\ t_2)\
        (\mathit{recNatZ}\ t_1)\ (\mathit{recNatS}\ t_2)\ n
      \\ \\ \rightsquigarrow_\beta^*
      & \mathit{recNatS}\ t_2\ n\ (\mathit{recNatZ}\ t_1)\ (\mathit{recNatS}\
        t_2)\ n
      \\ \\ \rightsquigarrow_\beta^*
      & t_2\ n\ (n\ (\mathit{recNatZ}\ t_1)\ (\mathit{recNatS}\
        t_2)\ (\mathit{recNatZ}\ t_1)\ (\mathit{recNatS}\ t_2)\
           (\mathit{S.pred}\ n))
      \\ \\ \leftrsquigarrow_\beta^*
      & t_2\ n\ (\mathit{recLRNat}\ (\mathit{recNatZ}\ t_1)\ (\mathit{recNatS}\
        t_2)\ (\mathit{elimCast}\ n)\
      \\ \\ & \quad \quad \quad (\mathit{recNatZ}\ t_1)\ (\mathit{recNatS}\ t_2)\ (\mathit{S.pred}\ n))
      \\ \\ \leftrsquigarrow_\beta^*
      & t_2\ n\ (\mathit{recNat}\ t_1\ t_2\ n)
    \end{array}
  \]
  \caption{Reduction of \(\mathit{recNat}\) for the successor case}
  \label{fig:lr-nat-rec-reduce}
\end{figure}

\subsection{Induction for Scott naturals, concretely}
\label{sec:lr-inductive-nat}

In this section, we describe modifications to the Lepigre-Raffalli encoding
allowing us to equip Scott naturals with induction.
For completeness, we show the full derivation, but as this development is 
similar to what preceded, we shall only highlight the differences.

\begin{figure}
  \centering
  \small
\begin{verbatim}
import cast.
import mono.
import recType.

import scott/concrete/nat as S .

module lepigre-raffalli/concrete/nat2 .

NatRec ◂ (S.Nat ➔ ★) ➔ S.Nat ➔ ★ ➔ ★ ➔ ★
= λ P: S.Nat ➔ ★. λ x: S.Nat. λ Z: ★. λ S: ★.
  Z ➔ S ➔ Z ➔ S ➔ P x.

NatZ ◂ (S.Nat ➔ ★) ➔ ★
= λ P: S.Nat ➔ ★. ∀ Z: ★. ∀ S: ★. Z ➔ S ➔ P S.zero .

NatS ◂ (S.Nat ➔ ★) ➔ ★
= λ P: S.Nat ➔ ★.
  ∀ Z: ★. ∀ S: ★. Π n: (ι x: S.Nat. NatRec ·P x ·Z ·S).
  Z ➔ S ➔ P (S.suc n.1) .

Nat ◂ ★ = ι x: S.Nat. ∀ P: S.Nat ➔ ★. NatRec ·P x ·(NatZ ·P) ·(NatS ·P) .

recLRNat ◂ ∀ P: S.Nat ➔ ★. NatZ ·P ➔ NatS ·P ➔ Π n: Nat. NatZ ·P ➔ NatS ·P ➔ P n.1
= Λ X. λ z. λ s. λ n. n.2 z s .

zero ◂ Nat
= [ S.zero , Λ X. λ z. λ s. z ·(NatZ ·X) ·(NatS ·X) ] .

suc ◂ Nat ➔ Nat
= λ n.
  [ S.suc n.1
  , Λ P. λ z. λ s.
    s ·(NatZ ·P) ·(NatS ·P) [ n.1 , λ z. λ s. recLRNat z s n ] ] .

rollNat ◂ Cast ·(S.NatFI ·Nat) ·Nat
= intrCast
    -(λ n. n.1 zero suc)
    -(λ n. n.2 ·(λ x: S.NatF ·Nat. { x zero suc ≃ x }) β (λ m. β)) .

toNat ◂ Cast ·S.Nat ·Nat = recLB -rollNat .
\end{verbatim}
  \caption{Induction for Scott naturals (part 1)
    (\texttt{lepigre-raffalli/concrete/nat2.ced})}
  \label{fig:lr-nat-ind-1}
\end{figure}

In Figure~\ref{fig:lr-nat-ind-1}, we begin our modification by making
\(\mathit{NatRec}\) dependent: \(\mathit{NatRec} \cdot P\ n \cdot Z \cdot S\)
is the type of functions taking two arguments each of type \(Z\) and \(S\) and
returning a proof that \(P\) holds for \(n\).
The next and most significant modification is to the type family \(\mathit{NatS}\).
We want that the handle \(n\) for invoking our inductive hypothesis will produce
a proof that \(P\) holds for the predecessor --- which is \(n\) itself!
We can express the dual role of the predecessor as data (\(n\)) and
function (\(\absu{\lambda}{z}{\absu{\lambda}{s}{\mathit{recLRNat}\ z\ s\ n}}\))
with dependent intersections, which recovers the view of the predecessor as a
Scott natural.
This duality is echoed in \(\mathit{Nat}\), which is defined
with dependent intersections as the type of Scott naturals \(x\) which, for an
arbitrary predicate \(P\), will act as a function taking two base
\((\mathit{NatZ} \cdot P)\) and step (\(\mathit{NatS} \cdot P\)) cases and
produce a proof that \(P\) holds of \(x\).

By its definition, the type \(\mathit{Nat}\) of the modified Lepigre-Raffalli
encoding is a subtype of Scott encodings.
Using the same approach as in Section~\ref{sec:lr-nat}, we can show a type
inclusion in the other direction.
We define the constructors \(\mathit{zero}\) and \(\mathit{suc}\), noting that
the innermost intersection introduction in \(\mathit{suc}\) again shows the dual
role of the predecessor, then show that \(\mathit{Nat}\) is \(\mathit{S.NatFI}\)
closed by proving that rebuilding a term of type \(\mathit{S.NatFI} \cdot
\mathit{Nat}\) with these constructors reproduces the original term at type
\(\mathit{Nat}\).

\begin{figure}
  \centering
  \small
\begin{verbatim}
indNatZ ◂ ∀ P: S.Nat ➔ ★. P S.zero ➔ NatZ ·P
= Λ X. λ x. Λ Z. Λ S. λ z. λ s. x .

indNatS ◂ ∀ P: S.Nat ➔ ★. (Π n: S.Nat. P n ➔ P (S.suc n)) ➔ NatS ·P
= Λ P. λ f. Λ Z. Λ S. λ r. λ z. λ s. f r.1 (r.2 z s z s) .

indNat
◂ ∀ P: S.Nat ➔ ★. P S.zero ➔ (Π m: S.Nat. P m ➔ P (S.suc m)) ➔
  Π n: S.Nat. P n
= Λ P. λ x. λ f. λ n.
  recLRNat ·P (indNatZ x) (indNatS f) (elimCast -toNat n)
    (indNatZ x) (indNatS f) .
\end{verbatim}
  \caption{Induction for Scott naturals (part 2)
    (\texttt{lepigre-raffalli/concrete/nat2.ced})}
  \label{fig:lr-nat-ind-2}
\end{figure}

Finally, we derive true induction for Scott naturals in
Figure~\ref{fig:lr-nat-ind-2}.
Playing the same roles as \(\mathit{recNatZ}\) and \(\mathit{recNatS}\)
(Figure~\ref{fig:lr-nat-3}), \(\mathit{indNatZ}\) and \(\mathit{indNatS}\)
convert the usual base and step cases of an inductive proof into forms suitable
for Lepigre-Raffalli-style induction.
In particular, note that in \(\mathit{indNatS}\) the bound \(r\) plays the role
of both predecessor (\(r.1\)) and handle for the inductive hypothesis (\(r.2\)).

\subsection{Induction for Scott-encoded data, generically}
\label{sec:lr-gen}

In this section, we formulate a \emph{generic} variant of the Lepigre-Raffalli
encoding and use this to derive induction for our generic Scott encoding.
To explain the generic encoding, we start by deriving a Lepigre-Raffalli recursion
scheme from the observation that, for the solution \(\mathit{caseD}\) for the
case-distinction scheme given in Figure~\ref{fig:scott-props-case}, \(|t|
=_{\beta\eta} |\absu{\lambda}{a}{\mathit{caseD}\ a\ t}|\) for all \(t\).
This allows us to introduce recursion into the computation law for case
distinction.

Let \(D\) be the type of Scott-encoded data whose signature is \(F\), let
\(t : \mathit{AlgCase} \cdot D \cdot T\) for some \(T\), and let \(t' : F \cdot
D\).
For the sake of exposition, we will for now assume that \(F\) is a functor.
Using the functor identity law, we can form the following equation from the
computation law of case distinction.
\[|\mathit{caseD}\ t\ (\mathit{inD}\ t')| \rightsquigarrow |t\ t'|
  =_{\mathit{FmapId}} |t\ (\mathit{fmap}\
  (\absu{\lambda}{x}{\absu{\lambda}{a}{\mathit{caseD}\ a\ x}})\ t')|\]

Renaming \(\mathit{caseD}\) to \(\mathit{recLRD}\), we can read the above as a
computational characterization of the generic Lepigre-Raffalli recursion scheme
and use this to derive a suitable typing law.
We desire that \(t\) should be a function which will be able to accept itself as
a second argument in order to make recursive calls on predecessors.
\[|\mathit{recLRD}\ t\ (\mathit{inD}\ t')\ t| \rightsquigarrow |t\ t'\ t|
  =_{\mathit{FmapId}} |t\ (\mathit{fmap}\
  (\absu{\lambda}{x}{\absu{\lambda}{a}{\mathit{recLRD}\ a\ x}})\ t')\ t|\]
Under this interpretation, we are able to give the following type for terms
\(t\) used in generic Lepigre-Raffalli recursion to compute a value of type \(T\).
\[\abs{\forall}{Y}{\star}{F \cdot (Y \to Y \to T) \to Y \to T}\]
Compare this to Lepigre-Raffalli recursion for naturals (Figure~\ref{fig:scheme-rec-lr-nat}).
\begin{itemize}
\item Quantification over \(Y\) replaces quantification over \(Z\) and \(S\) for
  the base and step case of naturals. Here, we intend that \(Y\) will be impredicatively
  instantiated with the type \(\abs{\forall}{Y}{\star}{F \cdot (Y \to Y \to T)
    \to Y \to T}\) again.
  
\item The single handle for making a recursive call on the natural number
  predecessor becomes an \(F\)-collection of handles of type \(Y \to Y \to T\)
  for making recursive calls, obtained from the \(F\)-collection of \(D\) predecessors.
\end{itemize}
This leads to the typing law for \(\mathit{recLRD}\) listed in
Figure~\ref{fig:scheme-lr-gen}.
For the computation law, we desire it be precisely the same as that for
\(\mathit{caseD}\) --- meaning that we will not need to require \(F\) to be a
functor for Lepigre-Raffalli recursion (or induction) over datatype \(D\).

\begin{figure}[h]
  \centering
  \[
    \begin{array}{c}
      \infer{
       \Gamma \vdash \mathit{recLRD} \cdot T\ t \tpsynth D \to (\abs{\forall}{Y}{\star}{F
      \cdot (Y \to Y \to T) \to Y \to T}) \to T
      }{
      \Gamma \vdash T \tpsynth \star
      \quad \Gamma \vdash t \tpcheck \abs{\forall}{Y}{\star}{F \cdot (Y \to Y
      \to T) \to Y \to T}
      }
    \end{array}
  \] 
  \caption{Typing law for the generic Lepigre-Raffalli recursion scheme}
  \label{fig:scheme-lr-gen}
\end{figure}

Unlike the other schemes we have considered, we are unaware of any standard
criteria for characterizing Lepigre-Raffalli recursion, and the development of a
categorical semantics for this scheme is beyond the scope of this paper.
Instead, and under the assumption that \(F\) is a functor, we will show that
from this scheme and the related induction principle we can give efficient and
provably unique solutions to the iteration and primitive recursion schemes.

The derivations of this section are separated into three parts.
In Figure~\ref{fig:lr-generic-1}, we give the type of our generic variant of the
Lepigre-Raffalli encoding for a monotone signature.
In Figure~\ref{fig:lr-generic-2}, we derive Lepigre-Raffalli induction for
Scott encodings.
Finally, in Figure~\ref{fig:lr-generic-3} we assume the stronger condition that
the datatype signature is a functor and derive a standard induction principle
for Scott encodings.

\begin{figure}[h]
  \centering
  \small
\begin{verbatim}
import cast .
import mono .
import recType .

module scott-rec/generic/encoding (F: ★ ➔ ★) {mono: Mono ·F} .

import scott/generic/encoding as S ·F -mono .

DRec ◂ (S.D ➔ ★) ➔ S.D ➔ ★ ➔ ★
= λ P: S.D ➔ ★. λ x: S.D. λ Y: ★. Y ➔ Y ➔ P x .

inDRec ◂ ∀ P: S.D ➔ ★. ∀ Y: ★. F ·(ι x: S.D. DRec ·P x ·Y) ➔ S.D
= Λ P. Λ Y. λ xs.
  [c ◂ Cast ·(ι x: S.D. DRec ·P x ·Y) ·S.D
   = intrCast -(λ x. x.1) -(λ x. β)]
- S.inD (elimCast -(mono c) xs) .

PrfAlgLR ◂ (S.D ➔ ★) ➔ ★
= λ P: S.D ➔ ★.
  ∀ Y: ★. Π xs: F ·(ι x: S.D. DRec ·P x ·Y). Y ➔ P (inDRec xs) .

D ◂ ★ = ι x: S.D. ∀ P: S.D ➔ ★. DRec ·P x ·(PrfAlgLR ·P) .

recLRD ◂ ∀ P: S.D ➔ ★. PrfAlgLR ·P ➔ Π x: D. PrfAlgLR ·P ➔ P x.1
= Λ P. λ a. λ x. x.2 a .
\end{verbatim}
  \caption{Generic Lepigre-Raffalli-style induction for Scott encodings (part 1)
    (\texttt{lepigre-raffalli/generic/encoding.ced})}
  \label{fig:lr-generic-1}
\end{figure}

\paragraphb{Generic Lepigre-Raffalli encoding.}
In Figure~\ref{fig:lr-generic-1}, we begin by importing the generic Scott
encoding, using prefix ``\(\mathit{S.}\)'' to access definitions in that
module.
Type family \(\mathit{DRec}\) gives the shape of the types of handles for
invoking an inductive hypothesis for a particular term \(x\) of type
\(\mathit{S.D}\) and predicate \(P : \mathit{S.D} \to \star\)
(compare this to \(\mathit{NatRec}\) in Figure~\ref{fig:lr-nat-ind-1}).
Next, for all predicates \(P\) over Scott encodings \(\mathit{S.D}\),
\(\mathit{PrfAlgLR} \cdot P\) is the type of Lepigre-Raffalli-style proof
algebras for \(P\), corresponding to \(\mathit{NatZ} \cdot P\) and
\(\mathit{NatS} \cdot P\) together in Figure~\ref{fig:lr-nat-ind-1}.
A term of this type is polymorphic in a type \(Y\) (which we interpret as
standing in for \(\mathit{PrfAlgLR} \cdot P\) itself) and takes an
\(F\)-collection \(\mathit{xs}\) of terms which, with the use of dependent
intersection types, are each interpreted both as a predecessor and a handle for
accessing the inductive hypothesis for that predecessor.
The argument of type \(Y\) is the step case to be given to these handles.
To state that the result should be a proof that \(P\) holds for the value
constructed from these predecessors, we need a variant constructor
\(\mathit{inDRec}\) that first casts the predecessors to the type
\(\mathit{S.D}\) using monotonicity of \(F\).
Note that we have \(|\mathit{inDRec}| =_{\beta\eta} |\mathit{S.inD}|\).

Type \(D\) is our generic Lepigre-Raffalli encoding, again defined with
dependent intersection as the type for Scott encodings \(x\) that also act as
functions that, for all properties \(P\), produce a proof that \(P\) holds of
\(x\) when given two Lepigre-Raffalli-style proof algebras for \(P\).
Finally, \(\mathit{recLRD}\) is the Lepigre-Raffalli-style induction principle
restricted to those Scott encodings which have type \(D\).
To obtain true Lepigre-Raffalli induction, it remains to show that \emph{every}
term of type \(\mathit{S.D}\) has type \(D\).

\begin{figure}[h]
  \centering
  \small
\begin{verbatim}
fromD ◂ Cast ·D ·S.D
= intrCast -(λ x. x.1) -(λ x. β) .

instDRec ◂ ∀ P: S.D ➔ ★. Cast ·D ·(ι x: S.D. DRec ·P x ·(PrfAlgLR ·P))
= Λ P. intrCast -(λ x. [ x.1 , λ a. recLRD a x ]) -(λ x. β) .

inD ◂ F ·D ➔ D
= λ xs.
  [ S.inD (elimCast -(mono fromD) xs)
  , Λ P. λ a.
    a ·(PrfAlgLR ·P) (elimCast -(mono (instDRec ·P)) xs) ].

rollD ◂ Cast ·(S.DFI ·D) ·D
= intrCast
    -(λ x. x.1 inD)
    -(λ x. x.2 ·(λ x: S.DF ·D. { x inD ≃ x }) (λ xs. β)) .

toD ◂ Cast ·S.D ·D = recLB -rollD .

indLRD ◂ ∀ P: S.D ➔ ★. PrfAlgLR ·P ➔ Π x: S.D. PrfAlgLR ·P ➔ P x
= Λ P. λ a. λ x. recLRD a (elimCast -toD x) .
\end{verbatim}
  \caption{Generic Lepigre-Raffalli-style induction for Scott encodings (part 2)
    (\texttt{lepigre-raffalli/generic/encoding.ced})}
  \label{fig:lr-generic-2}
\end{figure}

\paragraphb{Lepigre-Raffalli induction.}
In Figure~\ref{fig:lr-generic-2}, we begin the process of demonstrating an
inclusion of the type \(\mathit{S.D}\) into \(D\) by defining the constructor
\(\mathit{inD}\) for the generic Lepigre-Raffalli encoding.
This definition crucially uses the auxiliary function \(\mathit{instDRec}\) to
produce the two views of a given predecessor \(\ann{x}{D}\) as subdata (\(x.1\))
and as a handle for the inductive hypothesis associated to that predecessor
(\(\absu{\lambda}{a}{\mathit{recLRD}\ a\ x}\)) for a given predicate \(P\);
compare this to the definition of \(\mathit{suc}\) in Figure~\ref{fig:lr-nat-ind-1}.
As expected, we have that \(\mathit{inD}\) and \(\mathit{S.inD}\) (and also
\(\mathit{S.inDF}\)) are definitionally equal.

With the constructor defined, we show with \(\mathit{rollD}\) that \(D\) is an
\(\mathit{S.DFI}\)-closed type (\(\mathit{S.DFI}\) is defined in
Figure~\ref{fig:scott-generic-3}) by giving a proof that rebuilding a term of
type \(\mathit{S.DFI} \cdot D\) with constructor \(\mathit{inD}\) reproduces the
same term.
As \(\mathit{S.D} = \mathit{Rec} \cdot \mathit{S.DFI}\) is a lower bound of all
such types, we thus obtain a proof \(\mathit{toD}\) of an inclusion of the type
\(\mathit{S.D}\) into \(D\) using \(\mathit{recLB}\) (Figure~\ref{fig:recType}).
With this, we have Lepigre-Raffalli-style induction as \(\mathit{indLRD}\).

\begin{figure}
  \centering
  \small
\begin{verbatim}
import functor .
import cast .
import mono .
import utils .

module lepigre-raffalli/generic/induction
  (F: ★ ➔ ★) (fmap: Fmap ·F)
  {fmapId : FmapId ·F fmap} {fmapCompose: FmapCompose ·F fmap} .

import functorThms ·F fmap -fmapId -fmapCompose .

import scott/generic/encoding as S ·F -monoFunctor .
import lepigre-raffalli/generic/encoding ·F -monoFunctor .

import data-char/primrec-typing ·F .
import data-char/primrec ·F fmap -fmapId -fmapCompose ·S.D S.inD .

applyDRec ◂ ∀ P: S.D ➔ ★. ∀ Y: ★. Y ➔ (ι x: S.D. DRec ·P x ·Y) ➔ Sigma ·S.D ·P
= Λ P. Λ Y. λ y. λ x. mksigma x.1 (x.2 y y) .

fromPrfAlgRec
◂ ∀ P: S.D ➔ ★. PrfAlgRec ·P ➔ PrfAlgLR ·P
= Λ P. λ a. Λ Y. λ xs. λ y.
  ρ ς (fmapId ·(ι x: S.D. DRec ·P x ·Y) ·S.D
         (λ x. proj1 (applyDRec y x)) (λ x. β) xs)
    @x.(P (S.inD x))
- ρ ς (fmapCompose (proj1 ·S.D ·P) (applyDRec ·P y) xs)
    @x.(P (S.inD x))
- a (fmap (applyDRec ·P y) xs) .

indD ◂ ∀ P: S.D ➔ ★. PrfAlgRec ·P ➔ Π x: S.D. P x
= Λ P. λ a. λ x. indLRD (fromPrfAlgRec a) x (fromPrfAlgRec a) .
\end{verbatim}
  \caption{Generic induction for Scott encodings
    (\texttt{lepigre-raffalli/generic/induction.ced})}
  \label{fig:lr-generic-3}
\end{figure}

\paragraphb{Standard induction for Scott encodings.}
To derive the usual induction principle for Scott encodings using
Lepigre-Raffalli induction, we change module contexts in
Figure~\ref{fig:lr-generic-3} and now assume that \(F\) is a functor (see
Section~\ref{sec:functor-sigma} for the definitions of the functor laws).
As we did for the derivation of induction for Scott naturals, our
approach here is to convert a proof algebra of the form for standard induction
\[\mathit{PrfAlgRec} \cdot \mathit{S.D}\ \mathit{S.inD} \cdot P =
  \abs{\Pi}{xs}{F \cdot (\mathit{Sigma} \cdot \mathit{S.D} \cdot P)}{P\
    (\mathit{S.inD}\ (\mathit{fmap}\ \mathit{proj1}\ \mathit{xs}))}\]
into one of the form for Lepigre-Raffalli induction.

Function \(\mathit{fromPrfAlgRec}\) gives the conversion of proof algebras, and
its body is best read bottom-up.
The bound \(\mathit{xs}\) is an \(F\)-collection of predecessors playing dual
roles as subdata and handles for inductive hypotheses that require proof
algebras at the universally quantified type \(Y\), and the bound \(\ann{y}{Y}\)
is ``self-handle'' of the Lepigre-Raffalli proof algebra we are defining that
gives us access to those inductive hypotheses.
With \(\mathit{applyDRec}\) we separate these two roles, producing a dependent
pair of type \(\mathit{Sigma} \cdot \mathit{S.D} \cdot P\) as expected for the
usual formulation of induction.

The type of the final line of \(\mathit{fromPrfAlgRec}\) is:
\[P\ (\mathit{S.inD}\ (\mathit{fmap}\ \mathit{proj1}\ (\mathit{fmap}\
  (\mathit{applyDRec}\ y)\ \mathit{xs})))
\]
We use the functor composition law to fuse the two 
mappings of \(\mathit{proj1}\) and \(\mathit{applyDRec}\ y\).
Then, observing that this results in the mapping of a function that is
definitionally equal \(\absu{\lambda}{x}{x}\) (by the computation law of
\(\mathit{proj1}\), Figure~\ref{fig:sigma}), we use the functor identity law to
remove the mapping completely.
The resulting type is convertible with the expected type \(P\ (\mathit{inDRec}\
\mathit{xs})\) (since \(|\mathit{inDRec}| =_{\beta\eta}\ |\mathit{S.inD}|\)).
With the conversion complete, in \(\mathit{indD}\) we equip Scott encodings with
the standard induction principle by invoking Lepigre-Raffalli induction on
two copies of the converted proof algebra.

\subsubsection{Computational and extensional character.}
\label{sec:lr-gen-char}

With the standard induction principle derived for Scott encodings, we can now
show that Scott-encoded datatypes enjoy the same characterization, up to
propositional equality, as do Parigot-encoded datatypes.
Concerning the efficiency of solutions to recursion schemes, we have already
seen that Scott encodings offer a superior simulation of case distinction.
We now consider primitive recursion and iteration.
For the module listed in Figures~\ref{fig:lr-props} and \ref{fig:lr-init}, the
generic Scott encoding is imported without qualification, and
``\(\mathit{LR.}\)'' qualifies the definitions imported from generic
Lepigre-Raffalli encoding.

\begin{figure}
  \centering
  \small
\begin{verbatim}
import functor .
import cast .
import recType .
import utils .

module lepigre-raffalli/generic/propos
  (F: ★ ➔ ★) (fmap: Fmap ·F)
  {fmapId: FmapId ·F fmap} {fmapCompose: FmapCompose ·F fmap} .

import functorThms ·F fmap -fmapId -fmapCompose .
import scott/generic/encoding ·F -monoFunctor .
import lepigre-raffalli/generic/encoding as LR ·F -monoFunctor .
import lepigre-raffalli/generic/induction ·F fmap -fmapId -fmapCompose .

import data-char/primrec-typing ·F .
import data-char/primrec ·F fmap -fmapId -fmapCompose ·D inD .

recD ◂ PrimRec ·D
= Λ X. λ a. indD ·(λ x: D. X) a .

recDBeta ◂ PrimRecBeta recD
= Λ X. Λ a. Λ xs. β .

recDEta ◂ PrimRecEta recD = <..>
\end{verbatim}
  \caption{Characterization of \(\mathit{recD}\) (\texttt{lepigre-raffalli/generic/props.ced})}
  \label{fig:lr-props}
\end{figure}

\paragraphb{Primitive recursion.}
The solution \(\mathit{recD}\) in Figure~\ref{fig:lr-props} for the combinator
for primitive recursion is a non-dependent instance of standard induction.
As we saw for primitive recursion on Scott naturals in Section~\ref{sec:lr-nat},
the computation law for generic primitive recursion, proved by
\(\mathit{recDBeta}\), does not hold by reduction in the operational semantics
alone, but \emph{does} hold up to joinability using a constant number of
\(\beta\)-reduction steps. 
We illustrate for arbitrary (untyped) terms \(t\) and \(t'\).

\[
  \begin{array}{cl}
    & \mathit{recD}\ t\ (\mathit{inD}\ t')
    \\ \rightsquigarrow^*_\beta
    & \mathit{indLRD}\ (\mathit{fromPrfAlgRec}\ t)\ (\mathit{inD}\ t')\ (\mathit{fromPrfAlgRec}\ t)
    \\ \rightsquigarrow^*_\beta
    & \mathit{fromPrfAlgRec}\ t\ t'\ (\mathit{fromPrfAlgRec}\ t)
    \\ \rightsquigarrow^*_\beta
    & t\ (\mathit{fmap}\ (\mathit{applyDRec}\ (\mathit{fromPrfAlgRec}\ t))\ t')
    \\ \rightsquigarrow^*_\beta
    & t\ (\mathit{fmap}\ (\absu{\lambda}{x}{\mathit{mksigma}\ x\ (x\
      (\mathit{fromPrfAlgRec}\ t)\ (\mathit{fromPrfAlgRec}\ t))})\ t')
    \\ \leftrsquigarrow^*_\beta
    & t\ (\mathit{fmap}\ (\absu{\lambda}{x}{\mathit{mksigma}\ x\ (\mathit{recD}\
      t)})\ t')
    \\ \leftrsquigarrow^*_\beta
    & t\ (\mathit{fmap}\ (\mathit{fork}\ \mathit{id}\ (\mathit{recD}\ a))\ t')
  \end{array}
\]
The extensionality law \(\mathit{recDEta}\) (proof omitted) follows from
induction.

\begin{figure}[h]
  \centering
  \small
\begin{verbatim}
import data-char/iter-typing ·F .
import data-char/iter ·F fmap ·D inD .

lrFromAlg ◂ ∀ X: ★. Alg ·X ➔ LR.PrfAlgLR ·(λ x: D. X)
= Λ X. λ a. Λ Y. λ xs. λ y.
  a (fmap ·(ι x: D. LR.DRec ·(λ x: D. X) x ·Y) ·X (λ x. x.2 y y) xs) .

foldD ◂ Iter ·D
= Λ X. λ a. λ x. LR.indLRD ·(λ x: D. X) (lrFromAlg a) x (lrFromAlg a) .

foldDBeta ◂ IterBeta foldD
= Λ X. Λ a. Λ xs. β .

algHomLemma
◂ ∀ X: ★. ∀ a: Alg ·X. ∀ h: D ➔ X. AlgHom ·X a h ➔ AlgRecHom ·X (fromAlg a) h
= <..>

foldDEta ◂ IterEta foldD = <..>
\end{verbatim}
  \caption{Characterization of \(\mathit{foldD}\) (\texttt{lepigre-raffalli/generic/props.ced})}
  \label{fig:lr-init}
\end{figure}

\paragraphb{Iteration.}
While primitive recursion can be used to simulate iteration, we saw in
Section~\ref{sec:parigot-props} that this results in a definition of
\(\mathit{foldD}\) that obeys the expected computation law only up to the
functor laws.
We now show in Figure~\ref{fig:lr-init} that with Lepigre-Raffalli recursion, we
can do better and obtain a solution obeying the computation law by definitional
equality alone.
The first definition, \(\mathit{lrFromAlg}\), converts a function of type
\(\mathit{Alg} \cdot X\) (used in iteration) to a function for use in
Lepigre-Raffalli recursion.
It maps over the \(F\)-collection of dual-role predecessors, applying
each to two copies of the handle \(y\) specifying the next step of recursion.
For the solution \(\mathit{foldD}\) for iteration, we use Lepigre-Raffalli
recursion on two copies of the converted \(a : \mathit{Alg} \cdot X\).

As was the case for \(\mathit{recD}\), with \(\mathit{foldD}\) the left-hand and
right-hand sides of the computation law for iteration are joinable using a
constant number of full \(\beta\)-reductions.
This means that the proof \(\mathit{foldDBeta}\) holds by definitional equality
alone.
The proof of the extensionality law, \(\mathit{foldDEta}\), follows as a
consequence of \(\mathit{recDEta}\) and a lemma that any function \(h : D \to
X\) which satisfies the computation law for iteration with respect to some \(a :
\mathit{Alg} \cdot X\) also satisfies the computation law for primitive
recursion with respect to \(\mathit{fromAlg}\ a\)
(Figure~\ref{fig:data-char-primrec}).

\section{Scott encoding vs. Parigot encoding.}
Satisfaction of the computation and extensionality laws of iteration, the
destructor, case distinction, and primitive recursion by both the Scott and
Parigot encoding establishes that both are adequate representations of inductive
datatypes in Cedille.
However, there are compelling reasons for preferring the Scott encoding.
First, and as discussed earlier, the Parigot encoding suffers from significant
space overhead: Parigot naturals are represented in exponential space compared
to the linear-space Scott naturals.
Second, efficiency of the destructor for Scott encodings does not depend on the
choice of evaluation strategy, and the computation law for iteration is satisfied
by definitional equality.
Finally, not all monotone type schemes in Cedille are functors.
For such a type scheme \(F\), we cannot use \(F\) as a datatype signature for
the generic Parigot encoding.
However, we \emph{can} use \(F\) as a signature for the generic Scott encoding, and
although we cannot obtain the (generic) standard induction principle for the
resulting datatype, we still may still use Lepigre-Raffalli induction.

With our final example, we demonstrate that this last concern is not
hypothetical: the set of monotone type schemes is a \emph{strict} superset of
the set of functorial type schemes.
Consider a datatype for infinitely branching trees, which in Haskell would be
defined as:
\begin{verbatim}
data ITree = Leaf | Node (Nat -> ITree)
\end{verbatim}
The signature of \(\mathit{ITree}\) is positive.
However, because the recursive occurrence of \(\mathit{ITree}\) in the node
constructor's argument type occurs within an arrow, we cannot prove that the
functor identity law as formulated in Figure~\ref{fig:functor} holds.

To see why this is the case, assume we have \(f : \mathit{Nat} \to S\) and \(g :
S \to T\).
Lifting \(g\) over the type scheme \(\abs{\lambda}{X}{\star}{\mathit{Nat} \to
  X}\) and applying the result to \(f\), we obtain the expression:
\[\absu{\lambda}{x}{g\ (f\ x) : \mathit{Nat} \to T}\]
To prove the functor identity law, we must show that the expression above is
propositionally equal to \(f\) while assuming only that \(g\) behaves
extensionally like the identity function on terms of type \(S\).
Since Cedille's equality is intensional, we cannot conclude from this last
assumption that \(g\) is itself equal to \(\absu{\lambda}{x}{x}\).

In fact, in the presence of the \(\delta\) axiom we can \emph{prove} that the
covariant mapping for function types, and thus the mapping for the signature of
infinitary trees, does not satisfy the functor laws.
Recall that \(\delta\ \mhyph\ t\) (Figure~\ref{fig:cdle-equality}) can be
checked against any type if \(t\) proves an absurd equation, and that the
Cedille implementation  uses the B\"ohm-out algorithm to determine when an
equation is absurd.
The counter-example proceeds by picking closed instances of \(f\) and \(g\) such
that \(g\) behaves extensionally as the identity function, but nonetheless \(f\)
and \(\absu{\lambda}{x}{g\ (f\ x)}\) are B\"ohm-separable.

\paragraphb{Not every monotone scheme is functorial.}
\begin{figure}
  \centering
  \[
    \begin{array}{c}
      \infer{
       \Gamma \vdash \mathit{Sum} \cdot S \cdot T \tpsynth \star
      }{
      \Gamma \vdash S \tpsynth \star
      \quad \Gamma \vdash T \tpsynth \star
      }
      \\ \\
      \begin{array}{cc}
        \infer{
         \Gamma \vdash \mathit{in1} \cdot S \cdot T\ s \tpsynth \mathit{Sum}
        \cdot S \cdot T
        }{
        \Gamma \vdash \mathit{Sum} \cdot S \cdot T \tpsynth \star
        \quad \Gamma \vdash s \tpcheck S
        }
        &
        \infer{
         \Gamma \vdash \mathit{in2} \cdot S \cdot T\ t \tpsynth \mathit{Sum}
        \cdot S \cdot T
        }{
        \Gamma \vdash \mathit{Sum} \cdot S \cdot T \tpsynth \star
        \quad \Gamma \vdash t \tpcheck T
        }
      \end{array}
      \\ \\
      \infer{
       \Gamma \vdash \mathit{indsum}\ s \cdot P\ t_1\ t_2 \tpsynth P\ s
      }{
      \begin{array}{c}
      \Gamma \vdash s \tpsynth \mathit{Sum} \cdot S \cdot T
        \quad \Gamma \vdash P \tpsynth \mathit{Sum} \cdot S \cdot T \to \star
        \\ \Gamma \vdash t_1 \tpcheck \abs{\Pi}{x}{S}{P\ (\mathit{in1}\ x)}
        \quad \Gamma \vdash t_2 \tpcheck \abs{\Pi}{x}{T}{P\ (\mathit{in2}\ x)}
      \end{array}
      }
      \\ \\
      \begin{array}{lcl}
        |\mathit{indsum}\ (\mathit{in1}\ s)\ t_1\ t_2|
        & =_{\beta\eta}
        & |t_1\ s|
        \\ |\mathit{indsum}\ (\mathit{in2}\ t)\ t_1\ t_2|
        & =_{\beta\eta}
        & |t_2\ t|
      \end{array}
    \end{array}
  \]
  \caption{\(\mathit{Sum}\), axiomatically (\texttt{utils/sum.ced})}
  \label{fig:sum}
\end{figure}

\begin{figure}
  \centering
  \small
\begin{verbatim}
module signatures/itree .

import functor .
import cast .
import mono .
import utils .

import scott/concrete/nat .

ITreeF ◂ ★ ➔ ★ = λ X: ★. Sum ·Unit ·(Nat ➔ X) .

itreeFmap ◂ Fmap ·ITreeF
= Λ X. Λ Y. λ f. λ t.
  indsum t ·(λ _: ITreeF ·X. ITreeF ·Y) (λ u. in1 u) (λ x. in2 (λ n. f (x n))) .

monoITreeF ◂ Mono ·ITreeF
= Λ X. Λ Y. λ c.
  intrCast
    -(itreeFmap (elimCast -c))
    -(λ t. indsum t ·(λ x: ITreeF ·X. { itreeFmap (elimCast -c) x ≃ x })
             (λ u. β) (λ x. β)) .

t1 ◂ ITreeF ·Nat = in2 (λ x. x) .
t2 ◂ ITreeF ·Nat = itreeFmap (caseNat zero suc) t1 .

itreeFmapIdAbsurd ◂ FmapId ·ITreeF itreeFmap ➔ ∀ X: ★. X
= λ fid. Λ X.
  [pf ◂ { t2 ≃ t1 } = fid (caseNat zero suc) reflectNat t1]
- δ - pf .
\end{verbatim}
  \caption{Counter-example: a monotone type scheme which is not a functor
    (\texttt{signatures/itree.ced})}
  \label{fig:itree}
\end{figure}

In Figure~\ref{fig:sum}, we give an axiomatic summary of derivable sum
(coproduct) types with induction in Cedille.
The constructors are \(\mathit{in1}\) and \(\mathit{in2}\), and the induction
principle is \(\mathit{indsum}\) and follows the expected computation laws.
In Figure~\ref{fig:itree} we define \(\mathit{ITreeF}\), the signature for
infinitely branching trees, and \(\mathit{itreeFmap}\), its corresponding
mapping operation (here \(\mathit{Unit}\) is the single-element type).
Following this, we prove with \(\mathit{monoITreeF}\) that this type scheme is
monotonic.
The function from \(\mathit{ITreeF} \cdot X\) to \(\mathit{ITreeF} \cdot Y\)
that realizes the type inclusion is \(\mathit{itreeFmap}\ (\mathit{elimCast}\
\mhyph c)\), and the proof that it behaves as the identity
function follows by induction on \(\mathit{Sum}\).
Note that for the node case in particular, we know
that the function we are mapping (\(\mathit{elimCast}\ \mhyph c\)) is
\emph{definitionally} equal to the identity function.

For the counterexample \(\mathit{itreeFmapIdAbsurd}\), we
consider two terms of type \(\mathit{ITreeF} \cdot Nat\): the first, \(t_1\), is
defined using the second coproduct injection on the identity function for
\(\mathit{Nat}\), and the second, \(t_2\), is the result of mapping
\(\mathit{caseNat}\ \mathit{zero}\ \mathit{suc}\) over \(t_1\).
From \(\mathit{reflectNat}\) (Section~\ref{sec:scott-nat-comp}), we know that
\(\mathit{caseNat}\ \mathit{zero}\ \mathit{suc}\) behaves as the identity
function on \(\mathit{Nat}\).
If \(\mathit{itreeFmap}\) satisfied the functor identity law, we would thus
obtain a proof that \(t_1\) and \(t_2\) are propositionally equal.
However, the erasures of \(t_1\) and \(t_2\) are closed untyped terms which are
\(\beta\eta\)-inequivalent, so we use \(\delta\) to derive a contradiction.

This establishes that our generic derivation of Scott-encoded data, together with
Lepigre-Raffalli-style recursion and induction, allow for programming with a
strictly larger set of inductive datatypes than does our generic Parigot
encoding.
Though the generic formulation of the standard induction principle is not
derivable without assuming functoriality of the datatype signature, we observe
that in some cases that Lepigre-Raffalli induction can be used to give an
ordinary (datatype-specific) induction principle.
For example, for \(\mathit{ITree}\) we can derive: 
{\small%
\begin{verbatim}
indITree
◂ ∀ P: ITree ➔ ★. P leaf ➔
  (Π f: Nat ➔ ITree. (Π n: Nat. P (f n)) ➔ P (node f)) ➔ Π x: ITree. P x
= <..>
\end{verbatim}%
}%
\noindent where \(\mathit{leaf}\) and \(\mathit{node}\) are the constructors for
\(\mathit{ITree}\) (see \texttt{lepigre-raffalli/examples/itree.ced} in the code
repository).

\section{Related Work}
\label{sec:related}
\paragraphb{Monotone inductive types.}
\citet{matthes02} employs
Tarski's fixpoint theorem to motivate the construction of a typed lambda calculus
with monotone recursive types.
The gap between this order-theoretic result and type theory is bridged using
category theory, with evidence that a type scheme is monotonic corresponding to
the morphism-mapping rule of a functor. 
Matthes shows that as long as the reduction rule eliminating an
\textit{unroll} of a \textit{roll} incorporates the monotonicity witness in a
certain way, strong normalization of System~F is preserved by extension
with monotone isorecursive types.
Otherwise, he shows a counterexample to normalization.

In contrast, we establish that type inclusions (zero-cost casts) induce a
preorder \emph{within} the type theory of Cedille, and carry out a modification
of Tarski's order-theoretic result directly within it.
Evidence of monotonicity is given by an operation lifting type inclusions, not
arbitrary functions, over a type scheme.
As mentioned in the introduction, deriving monotone recursive types within the
type theory of Cedille has the benefit of guaranteeing that they enjoy precisely
the same meta-theoretic properties as enjoyed by Cedille itself -- no additional
work is required.

\paragraphb{Impredicative encodings and datatype recursion schemes.}
Our use of casts in deriving recursive types guarantees that the
\textit{rolling} and \textit{unrolling} operations take constant time,
permitting the definition of efficient data accessors for inductive datatypes
defined with them.
However, when using recursive types to encode datatypes one usually desires
efficient solutions to datatype recursion schemes, and the
derivation in Section~\ref{sec:tarskitrans} does not on its own provide this.

Independently, \citet{Me91_Predicative-Universes-Primitive-Recursion} and
\citet{Ge92_Inductive-Coinductive-Types-Iteration-Recursion} developed the
category-theoretic notion of recursive \(F\)-algebras to give the semantics of
the primitive recursion scheme for inductive datatypes, and
\citet{Ge92_Inductive-Coinductive-Types-Iteration-Recursion} and
\citet{matthes02} use this notion in extending a typed lambda calculus with
typing and computation laws for the primitive recursion scheme for datatypes.  
The process of using a particular recursion scheme to directly obtain an
impredicative encoding is folklore knowledge \citep[c.f.][Section
3.6]{AMU05_Iteration-and-Coiteration-for-HO-and-Nested-Datatypes}.
\citet{Ge14_Church-Scott-Encoding} used this process to examine the close
connection between the (co)case-distinction scheme and the Scott encoding, and
the primitive (co)recursion scheme and the Parigot encoding, for (co)inductive
types in a theory with primitive positive recursive types.
Using derived monotone recursive types in Cedille, we follow the same approach
but for encodings of datatypes with induction principles.
This allows us to establish within Cedille that the solution to the primitive
recursion scheme is unique (i.e., that we obtain \emph{initial} recursive
\(F\)-algebras).

\paragraphb{Recursor for Scott-encoded data.}
The non-dependent impredicative encoding we used to equip Scott naturals with
primitive recursion in Section~\ref{sec:lr-nat} is based on a result by
\cite{lepigre+19}.
We thus dub it the \emph{Lepigre-Raffalli} encoding, though they report that the
encoding is in fact due to Parigot.
In earlier work, \citet{parigot88} demonstrated the lambda term that realizes
the primitive recursion scheme for Scott naturals, but the typing of this term
involved reasoning outside the logical framework being used.
The type system in which \cite{lepigre+19} carry out this construction has
built-in notions of least and greatest type fixpoints and a sophisticated form
of subtyping that utilizes ordinals and well-founded circular typing derivations based on
cyclic proof theory
\citep{San02_A-Calculus-of-Circular-Proofs-and-its-Categorical-Semantics}.
Roughly, the correspondence between their type system and that of Cedille's is
so: both theories are Curry-style, enabling a rich subtyping relation which in Cedille
is internalized as \(\mathit{Cast}\); and in defining recursor for
Scott naturals, we replace the circular subtyping derivation with an internally
realized proof of the fact that our derived recursive types are least fixpoints
of type schemes.

Our two-fold generalization of the Lepigre-Raffalli encoding (making it both
generic \emph{and} dependent) is novel.
To the best of our knowledge, the observation that the Lepigre-Raffalli-style
recursion scheme associated to this encoding can be understood as introducing
recursion into the computation laws for the case-distinction scheme is also
novel.
This characterization informs both our minor modification to the encoding in
Section~\ref{sec:lr-nat} for natural numbers (we quantify over types for
both base and step cases of recursion, instead of quantifying over only the
latter as done by \cite{lepigre+19}) and the generic formulation.
Presently, this connection between Lepigre-Raffalli recursion and case
distinction serves a mostly pedagogical role; we leave as future
work the task of providing a more semantic (e.g., category-theoretic) account of
Lepigre-Raffalli recursion.

\paragraphb{Lambda encodings in Cedille.} Work prior to ours describes the
generic derivation of induction for lambda-encoded data in Cedille. This was
first accomplished by \cite{firsov18} for the Church and Mendler encodings,
which do not require recursive types as derived in this paper.
The approach used for the Mendler encoding was then further refined by
\cite{firsov18b} to enable
efficient data accessors, resulting in the first-ever example of a lambda
encoding in type theory with derivable induction, constant-time destructor, and
whose representation requires only linear space. To the best of our knowledge,
this paper establishes that the Scott encoding is the second-ever example of a
lambda encoding enjoying these same properties.
\cite{firsov18} and \cite{firsov18b} also provide a computational
characterization for their encodings, limited to Lambek's lemma and the
computation law for Mendler-style iteration.
For the Parigot and Scott encodings we have presented, we have shown Lambek's
lemma and both the computation \emph{and} extensionality laws for the
case-distinction, iteration, and primitive recursion schemes.

\section{Conclusion and Future Work}
\label{sec:conclusion}
We have shown how to derive monotone recursive types with constant-time
\emph{roll} and \emph{unroll} operations within the type theory of 
Cedille by applying Tarski's least fixpoint theorem to a preorder on types
induced by an internalized notion of type inclusion.
By deriving recursive types within a theory, rather than extending it, we do not
need to rework existing meta-theoretic results to ensure logical consistency or
provide a normalization guarantee.
As applications, we used the derived monotone recursive types to derive two
recursive representations of data in lambda calculus, the Parigot and Scott
encoding, \emph{generically} in a signature $F$.
For both encodings, we derived induction and gave a thorough characterization of
the solutions they admit for case distinction, iteration, and primitive
recursion.
In particular, we showed that with the Scott encoding all three of these schemes
can be efficiently simulated.
This derivation, which builds on a result described by \cite{lepigre+19},
crucially uses the fact that recursive types in Cedille provide \emph{least}
fixpoints of type schemes.

In the authors' opinion, we have demonstrated that lambda encodings in Cedille
provide an adequate basis for programming with inductive datatypes.
Building from this basis to a convenient surface language requires the
generation of monotonicity witnesses from datatype declarations.
Our experience coding such proofs in Cedille leads us to believe they can be
readily mechanized for a large class of positive datatypes,
including infinitary trees and the usual formulation of rose trees.
However, more care is needed for checking monotonicity of nested datatypes
\citep{BM98_Nested-Datatypes,
  AMU05_Iteration-and-Coiteration-for-HO-and-Nested-Datatypes}.

Finally, we believe our developments have raised two interesting questions
for future investigation.
The first of these is the development of a categorical semantics for
Lepigre-Raffalli recursion, which would allow a complete the characterization of
Scott-encoded datatypes whose signatures are positive but non-functorial.
Second, in the derivation of primitive recursion for Scott encodings, leastness
of the fixpoint formed from our derived recursive type operator plays a similar
role to the cyclic subtyping rules of \cite{lepigre+19}.
If this correspondence generalizes, some subset of their type system might be
translatable to Cedille, opening the way to a surface language with a rich
notion of subtyping in the presence of recursive types that is based on
internalized type inclusions.

\section*{Financial Aid}
  We gratefully acknowledge NSF
  support under award 1524519, and DoD support under award
  FA9550-16-1-0082 (MURI program).

\section*{Competing Interests}
The authors declare none.

\bibliography{biblio}

\begin{thebibliography}{}

\bibitem[Abadi et~al., 1993]{ACP93_Types-for-Scott}
Abadi, M., Cardelli, L., and Plotkin, G. (1993).
\newblock Types for the {S}cott numerals.
\newblock Unpublished note.

\bibitem[Abbott et~al., 2003]{AAG03_Categories-of-Containers}
Abbott, M.~G., Altenkirch, T., and Ghani, N. (2003).
\newblock Categories of containers.
\newblock In Gordon, A.~D., editor, {\em Foundations of Software Science and
  Computational Structures, 6th International Conference, {FOSSACS} 2003 Held
  as Part of the Joint European Conference on Theory and Practice of Software,
  {ETAPS} 2003, Warsaw, Poland, April 7-11, 2003, Proceedings}, volume 2620 of
  {\em Lecture Notes in Computer Science}, pages 23--38. Springer.

\bibitem[Abel, 2010]{Ab10_Integrating-Sized-and-Dependent-Types}
Abel, A. (2010).
\newblock Mini{A}gda: Integrating sized and dependent types.
\newblock In Komendantskaya, E., Bove, A., and Niqui, M., editors, {\em
  Partiality and Recursion in Interactive Theorem Provers, PAR@ITP 2010,
  Edinburgh, UK, July 15, 2010}, volume~5 of {\em EPiC Series}, pages 18--33.
  EasyChair.

\bibitem[Abel et~al.,
  2005]{AMU05_Iteration-and-Coiteration-for-HO-and-Nested-Datatypes}
Abel, A., Matthes, R., and Uustalu, T. (2005).
\newblock Iteration and coiteration schemes for higher-order and nested
  datatypes.
\newblock {\em Theor. Comput. Sci.}, 333(1-2):3--66.

\bibitem[Allen et~al., 2006]{allen+06}
Allen, S.~F., Bickford, M., Constable, R.~L., Eaton, R., Kreitz, C., Lorigo,
  L., and Moran, E. (2006).
\newblock Innovations in computational type theory using {N}uprl.
\newblock {\em J. Applied Logic}, 4(4):428--469.

\bibitem[Atkey, 2018]{Atk18_Quantitative-Type-Theory}
Atkey, R. (2018).
\newblock Syntax and semantics of quantitative type theory.
\newblock In Dawar, A. and Gr{\"{a}}del, E., editors, {\em Proceedings of the
  33rd Annual {ACM/IEEE} Symposium on Logic in Computer Science, {LICS} 2018,
  Oxford, UK, July 09-12, 2018}, pages 56--65. {ACM}.

\bibitem[Barendregt et~al., 2013]{BDS13_Lambda-Calculus-with-Types}
Barendregt, H.~P., Dekkers, W., and Statman, R. (2013).
\newblock {\em Lambda Calculus with Types}.
\newblock Perspectives in logic. Cambridge University Press.

\bibitem[Bird and Meertens, 1998]{BM98_Nested-Datatypes}
Bird, R.~S. and Meertens, L. G. L.~T. (1998).
\newblock Nested datatypes.
\newblock In {\em Mathematics of Program Construction, MPC'98, Marstrand,
  Sweden, June 15-17, 1998, Proceedings}, pages 52--67.

\bibitem[Breitner et~al., 2016]{breitner+16}
Breitner, J., Eisenberg, R.~A., Jones, S.~P., and Weirich, S. (2016).
\newblock Safe zero-cost coercions for {H}askell.
\newblock {\em J. Funct. Program.}, 26:e15.

\bibitem[Böhm et~al., 1979]{BDPR79_Bohm-Algorithm}
Böhm, C., Dezani-Ciancaglini, M., Peretti, P., and Rocca, S.~D. (1979).
\newblock A discrimination algorithm inside λ-β-calculus.
\newblock {\em Theoretical Computer Science}, 8(3):271 -- 291.

\bibitem[Crary et~al., 1999]{Crary+99}
Crary, K., Harper, R., and Puri, S. (1999).
\newblock {What is a Recursive Module?}
\newblock In {\em Proceedings of the ACM SIGPLAN 1999 Conference on Programming
  Language Design and Implementation (PLDI)}, pages 50--63, New York, NY, USA.
  ACM.

\bibitem[Dybjer and Palmgren, 2016]{dybjer16}
Dybjer, P. and Palmgren, E. (2016).
\newblock {Intuitionistic Type Theory}.
\newblock In Zalta, E.~N., editor, {\em The Stanford Encyclopedia of
  Philosophy}. Metaphysics Research Lab, Stanford University, winter 2016
  edition.

\bibitem[Firsov et~al., 2018]{firsov18b}
Firsov, D., Blair, R., and Stump, A. (2018).
\newblock {Efficient Mendler-Style Lambda-Encodings in Cedille}.
\newblock In Avigad, J. and Mahboubi, A., editors, {\em Interactive Theorem
  Proving - 9th International Conference, {ITP} 2018, Held as Part of the
  Federated Logic Conference, FloC 2018, Oxford, UK, July 9-12, 2018,
  Proceedings}, volume 10895 of {\em Lecture Notes in Computer Science}, pages
  235--252. Springer.

\bibitem[Firsov and Stump, 2018]{firsov18}
Firsov, D. and Stump, A. (2018).
\newblock Generic derivation of induction for impredicative encodings in
  cedille.
\newblock In Andronick, J. and Felty, A.~P., editors, {\em Proceedings of the
  7th {ACM} {SIGPLAN} International Conference on Certified Programs and
  Proofs, {CPP} 2018, Los Angeles, CA, USA, January 8-9, 2018}, pages 215--227.
  {ACM}.

\bibitem[Geuvers, 1992]{Ge92_Inductive-Coinductive-Types-Iteration-Recursion}
Geuvers, H. (1992).
\newblock Inductive and coinductive types with iteration and recursion.
\newblock In {\em Proceedings of the 1992 {W}orkshop on {T}ypes for {P}roofs
  and {P}rograms, {B}astad}, pages 183--207. Bastad, Chalmers University of
  Technology.

\bibitem[Geuvers, 2001]{geuvers01}
Geuvers, H. (2001).
\newblock {Induction Is Not Derivable in Second Order Dependent Type Theory}.
\newblock In Abramsky, S., editor, {\em Typed Lambda Calculi and Applications
  (TLCA)}, volume 2044 of {\em Lecture Notes in Computer Science}, pages
  166--181. Springer.

\bibitem[Geuvers, 2014]{Ge14_Church-Scott-Encoding}
Geuvers, H. (2014).
\newblock The {C}hurch-{S}cott representation of inductive and coinductive
  data.
\newblock Unpublished manuscript.

\bibitem[Ghani et~al., 2012]{GJF12_Generic-Fibrational-Induction}
Ghani, N., Johann, P., and Fumex, C. (2012).
\newblock Generic fibrational induction.
\newblock {\em Logical Methods in Computer Science}, 8(2).

\bibitem[Kleene, 1965]{kleene65}
Kleene, S. (1965).
\newblock {Classical Extensions of Intuitionistic Mathematics}.
\newblock In Bar-Hillel, Y., editor, {\em LMPS 2}, pages 31--44. North-Holland
  Publishing Company.

\bibitem[Kopylov, 2003]{kopylov03}
Kopylov, A. (2003).
\newblock Dependent intersection: {A} new way of defining records in type
  theory.
\newblock In {\em 18th {IEEE} Symposium on Logic in Computer Science {(LICS)}},
  pages 86--95.

\bibitem[Lambek, 1968]{Lam68_A-Fixpoint-Theorem-for-Complete-Categories}
Lambek, J. (1968).
\newblock A fixpoint theorem for complete categories.
\newblock {\em Mathematische Zeitschrift}, 103(2):151--161.

\bibitem[Lassez et~al., 1982]{lassez82}
Lassez, J.-L., Nguyen, V., and Sonenberg, E. (1982).
\newblock Fixed point theorems and semantics: a folk tale.
\newblock {\em Information Processing Letters}, 14(3):112 -- 116.

\bibitem[Leivant, 1983]{leivant83}
Leivant, D. (1983).
\newblock Reasoning about functional programs and complexity classes associated
  with type disciplines.
\newblock In {\em 24th Annual Symposium on Foundations of Computer Science
  (FOCS)}, pages 460--469. {IEEE} Computer Society.

\bibitem[Lepigre and Raffalli, 2019]{lepigre+19}
Lepigre, R. and Raffalli, C. (2019).
\newblock Practical subtyping for {C}urry-style languages.
\newblock {\em ACM Trans. Program. Lang. Syst.}, 41(1):5:1--5:58.

\bibitem[Matthes, 1999]{matthes98}
Matthes, R. (1999).
\newblock {Monotone Fixed-Point Types and Strong Normalization}.
\newblock In Gottlob, G., Grandjean, E., and Seyr, K., editors, {\em Computer
  Science Logic, 12th International Workshop, {CSL} '98, Annual Conference of
  the EACSL, Brno, Czech Republic, August 24-28, 1998, Proceedings}, volume
  1584 of {\em Lecture Notes in Computer Science}, pages 298--312. Springer.

\bibitem[Matthes, 2002]{matthes02}
Matthes, R. (2002).
\newblock Tarski's fixed-point theorem and lambda calculi with monotone
  inductive types.
\newblock {\em Synthese}, 133(1-2):107--129.

\bibitem[\mbox{The Coq development team}, 2018]{coq}
\mbox{The Coq development team} (2018).
\newblock {\em The Coq proof assistant reference manual}.
\newblock LogiCal Project.
\newblock Version 8.7.2.

\bibitem[{Mendler}, 1991]{Me91_Predicative-Universes-Primitive-Recursion}
{Mendler}, N.~P. (1991).
\newblock Predictive type universes and primitive recursion.
\newblock In {\em [1991] Proceedings Sixth Annual IEEE Symposium on Logic in
  Computer Science}, pages 173--184.

\bibitem[Miquel, 2001]{miquel01}
Miquel, A. (2001).
\newblock {The Implicit Calculus of Constructions Extending Pure Type Systems
  with an Intersection Type Binder and Subtyping}.
\newblock In Abramsky, S., editor, {\em {Typed Lambda Calculi and
  Applications}}, volume 2044 of {\em Lecture Notes in Computer Science}, pages
  344--359. Springer.

\bibitem[Parigot, 1988]{parigot88}
Parigot, M. (1988).
\newblock Programming with proofs: a second order type theory.
\newblock In Ganzinger, H., editor, {\em {European Symposium On Programming
  (ESOP)}}, volume 300 of {\em Lecture Notes in Computer Science}, pages
  145--159. Springer.

\bibitem[Parigot, 1989]{parigot89}
Parigot, M. (1989).
\newblock On the representation of data in lambda-calculus.
\newblock In B\"orger, E., B\"uning, H., and Richter, M., editors, {\em
  {Computer Science Logic (CSL)}}, volume 440 of {\em Lecture Notes in Computer
  Science}, pages 309--321. Springer.

\bibitem[Parigot, 1992]{parigot1992}
Parigot, M. (1992).
\newblock Recursive programming with proofs.
\newblock {\em Theoretical Computer Science}, 94(2):335--356.

\bibitem[Pierce, 2002]{pierce02}
Pierce, B.~C. (2002).
\newblock {\em Types and programming languages}.
\newblock {MIT} Press.

\bibitem[Pierce and Turner, 2000]{pierce+00}
Pierce, B.~C. and Turner, D.~N. (2000).
\newblock Local type inference.
\newblock {\em {ACM} Trans. Program. Lang. Syst.}, 22(1):1--44.

\bibitem[Santocanale,
  2002]{San02_A-Calculus-of-Circular-Proofs-and-its-Categorical-Semantics}
Santocanale, L. (2002).
\newblock A calculus of circular proofs and its categorical semantics.
\newblock In Nielsen, M. and Engberg, U., editors, {\em Foundations of Software
  Science and Computation Structures, 5th International Conference, {FOSSACS}
  2002. Held as Part of the Joint European Conferences on Theory and Practice
  of Software, {ETAPS} 2002 Grenoble, France, April 8-12, 2002, Proceedings},
  volume 2303 of {\em Lecture Notes in Computer Science}, pages 357--371.
  Springer.

\bibitem[Scott, 1962]{Sco62_A-System-of-Functional-Abstraction}
Scott, D. (1962).
\newblock A system of functional abstraction.
\newblock {L}ectures delivered at University of California, Berkeley.

\bibitem[S{\o}rensen and Urzyczyn, 2006]{Sorensen06}
S{\o}rensen, M.~H. and Urzyczyn, P. (2006).
\newblock {\em {Lectures on the Curry-Howard Isomorphism, Volume 149 (Studies
  in Logic and the Foundations of Mathematics)}}.
\newblock Elsevier Science Inc., New York, NY, USA.

\bibitem[Splawski and Urzyczyn, 1999]{SU99_Type-Fixpoints-Iteration-Recursion}
Splawski, Z. and Urzyczyn, P. (1999).
\newblock Type fixpoints: Iteration vs. recursion.
\newblock In R{\'{e}}my, D. and Lee, P., editors, {\em Proceedings of the
  fourth {ACM} {SIGPLAN} International Conference on Functional Programming
  {(ICFP} '99), Paris, France, September 27-29, 1999}, pages 102--113. {ACM}.

\bibitem[Stump, 2017]{stump17}
Stump, A. (2017).
\newblock {The Calculus of Dependent Lambda Eliminations}.
\newblock {\em J. Funct. Program.}, 27:e14.

\bibitem[Stump, 2018a]{Stu18_From-Realizibility-to-Induction}
Stump, A. (2018a).
\newblock From realizability to induction via dependent intersection.
\newblock {\em Ann. Pure Appl. Logic}, 169(7):637--655.

\bibitem[Stump, 2018b]{stump18b}
Stump, A. (2018b).
\newblock Syntax and typing for {C}edille core.
\newblock {\em CoRR}, abs/1811.01318.

\bibitem[Stump and Fu,
  2016]{SF16_Efficiency-of-Lambda-Encodings-in-Total-Type-Theory}
Stump, A. and Fu, P. (2016).
\newblock Efficiency of lambda-encodings in total type theory.
\newblock {\em J. Funct. Program.}, 26:e3.

\bibitem[Stump and Jenkins, 2021]{stump18c}
Stump, A. and Jenkins, C. (2021).
\newblock Syntax and semantics of {C}edille.
\newblock {\em CoRR}, abs/1806.04709.

\bibitem[Tarski, 1955]{tarski55}
Tarski, A. (1955).
\newblock A lattice-theoretical fixpoint theorem and its applications.
\newblock {\em Pacific Journal of Mathematics}, 5(2):285--309.

\bibitem[{The Agda Team}, 2021]{ADT21_Agda-StdLib}
{The Agda Team} (2021).
\newblock The {A}gda standard library, v1.5.
\newblock \url{https://github.com/agda/agda-stdlib}.

\bibitem[Ullrich,
  2020]{Ull20_Generating-Induction-Principles-for-Nested-Inductive-Types-MetaCoq}
Ullrich, M. (2020).
\newblock Generating induction principles for nested inductive types in
  {MetaCoq}.
\newblock Bachelor's thesis.

\bibitem[Uustalu and Vene,
  1999]{UV99_Primitive-Corecurusion-and-CoV-Coiteration}
Uustalu, T. and Vene, V. (1999).
\newblock Primitive (co)recursion and course-of-value (co)iteration,
  categorically.
\newblock {\em Informatica}, 10(1):5--26.

\bibitem[Wadler, 1990]{Wad90_Recursive-Types-for-Free}
Wadler, P. (1990).
\newblock Recursive types for free!
\newblock Unpublished manuscript.

\end{thebibliography}

\end{document}